\documentclass{article}
\usepackage{geometry}
\usepackage{amsmath,amsfonts, amsthm}
\usepackage{algorithm}
\usepackage{hyperref}
\usepackage{algcompatible}
\usepackage{algpseudocode}
\usepackage{array} 
\usepackage{textcomp}
\usepackage{stfloats}
\usepackage{url}
\usepackage{verbatim} 
\usepackage{graphicx} 
\usepackage{MnSymbol,bbding,pifont}
\usepackage{color}
\usepackage{multirow} 
\usepackage{booktabs}
\newcolumntype{P}[1]{>{\centering\arraybackslash}p{#1}}
\usepackage{stmaryrd}
\usepackage[sort]{cite}
\usepackage[normalem]{ulem}
\usepackage{subfigure}
\usepackage{lipsum}
\usepackage{amsfonts}
\usepackage{graphicx}
\usepackage{epstopdf}
\usepackage{color, colortbl}
\usepackage{times}
\definecolor{LightCyan}{rgb}{0.88,1,1}
\definecolor{gray}{rgb}{0.925,0.925,0.925}
\makeatletter
\newcommand{\pushright}[1]{\ifmeasuring@#1\else\omit\hfill$\displaystyle#1$\fi\ignorespaces}
\newcommand{\pushleft}[1]{\ifmeasuring@#1\else\omit$\displaystyle#1$\hfill\fi\ignorespaces}
\makeatother

\newcommand{\TT}{\mathsf{T}}

\newcommand{\norm}[1]{\left\Vert{#1}\right\Vert}
\newcommand{\bb}[1]{\mathbb{#1}}

\newcommand{\bd}[1]{\mathbf{#1}}
\newcommand{\bld}[1]{\boldsymbol{#1}}
\newcommand{\cl}[1]{\mathcal{#1}}

\newtheorem{theorem}{Theorem}
\newtheorem{proposition}{Proposition}
\newtheorem{corollary}{Corollary}
\newtheorem{lemma}{Lemma}
\newtheorem{definition}{Definition}

\title{Tight-frame-like Analysis-Sparse Recovery Using \\Non-tight Sensing Matrices}
\author{Nareddy  Kartheek Kumar Reddy, Abijith Jagannath Kamath, and  Chandra Sekhar Seelamantula \\ Department of Electrical Engineering, Indian Institute of Science, Bengaluru 560 012, India \\ Email:\{nareddyreddy, abijithj, css\}@iisc.ac.in}
\date{}    

\begin{document}

\maketitle
% ----------------------------------------
% ----------------------------------------
% 	             Abstract
% ----------------------------------------
% ----------------------------------------
\begin{abstract} 
The choice of the sensing matrix is crucial in compressed sensing.
Random Gaussian sensing matrices satisfy the \emph{restricted isometry property}, which is crucial for solving the sparse recovery problem using convex optimization techniques. However, \emph{tight-frame sensing matrices} result in minimum mean-squared-error recovery given \emph{oracle} knowledge of the support of the sparse vector. 
\emph{If the sensing matrix is not tight, could one achieve the recovery performance assured by a tight frame by suitably designing the recovery strategy?} --- This is the key question addressed in this paper.  
{\color{black} We consider the analysis-sparse $\ell_1$-minimization problem with a generalized $\ell_2$-norm-based data-fidelity and show that it effectively corresponds to using a \emph{tight-frame} sensing matrix.}
The new formulation offers improved performance bounds when the number of non-zeros is large.
{\color{black} One could develop a tight-frame variant of a known sparse recovery algorithm using the proposed formalism. We solve the analysis-sparse recovery problem in an unconstrained setting using proximal methods.}
Within the \emph{tight-frame sensing framework}, we rescale the gradients of the data-fidelity loss in the iterative updates to further improve the accuracy of analysis-sparse recovery. Experimental results show that the proposed algorithms offer superior analysis-sparse recovery performance. Proceeding further, we also develop deep-unfolded variants, with a convolutional neural network as the sparsifying operator. On the application front, we consider compressed sensing image recovery. Experimental results on Set11, BSD68, Urban100, and DIV2K datasets show that the proposed techniques outperform the state-of-the-art techniques, with performance measured in terms of peak signal-to-noise ratio and structural similarity index metric.
\end{abstract}
\section{Introduction}
 
Compressed sensing (CS) \cite{donoho2006cs, candes2008cs, baraniuk2007compressive, blu2008sparse, kutyniok2013theory, candes} allows for the recovery of signals using fewer measurements than required in the classical sampling paradigm. This is achieved by exploiting the inherent structure of the signal, such as its sparsity in a suitable bases. An important milestone has been the development of convex optimization techniques for efficiently solving sparse recovery problems \cite{tibshirani96regression,candes2008cs, daubechies2004iterative,beck2009fast, yin2008bregman, zhang2010bregmanized, aujol2015stability, ji2008bayesian, goldstein2009split}. Cand\`es {\it et al.} consider the convex $\ell_1$-norm minimization problem, and provide an upper bound on the $\ell_2$ error between the recovered signal and the ground truth, which guarantees accurate recovery when the signals are sparse in overcomplete dictionaries \cite{candes_analysis}. Compressed sensing has found numerous applications in medical imaging \cite{deep_ADMM_Net}, compressed sensing image recovery (CS-IR) \cite{dong2014compressive}, signal and image processing \cite{elad2010}, optical-coherence tomography \cite{Mukharjee2014,kamilov2016optical}, seismic imaging \cite{mache2021durin,mache2021nuspan,yuan2013spectral, mache2023introducing}, snapshot imaging \cite{snapshot}, magnetic resonance imaging \cite{vaswani2009}, unlimited sampling \cite{rudresh2018wavelet}, single-pixel camera \cite{duarte2008single}, radar imaging \cite{herman2009high, potter2010sparsity, de2019compressed, rossi2014spatial}, etc.\\
% --------------------------------------------------------------------------------------------------------------------------------------------------------------------------------
{\color{black}
\indent In analysis-sparse signal recovery, one considers the recovery of a vector $\bld x$ from compressed linear measurements of the type
\begin{equation}\label{eq:noise_model_vanilla}
\bld y = \bd A \bld x^* + \bld w,
\end{equation}
where $\bd A\in\bb R^{m\times n}, \; m \ll n$ denotes the \emph{sensing matrix}, and $\bld w\in\bb R^m$ denotes the measurement noise. The linear system of equations is underdetermined, which makes the problem ill-posed. Additional information of the structure of the signal is required for recovery, and we consider the signal $\bld x^*$ to be a sparse linear combination of atoms of a dictionary $\bd D$.
In applications such as compressed sensing image recovery, the signal of interest can be \emph{sparsified} using a suitable transformation, which results in the so-called {\it analysis-sparse recovery} problem \cite{candes_analysis}. Stated mathematically, the ground-truth signal $\bld x^*$ is not sparse in the canonical basis, but it admits the representation $\bld x^* = \bd D \bld\alpha^*$, where $\bd D$ is an overcomplete dictionary and $\bld\alpha^*$ is sparse in the canonical basis. This representation considers a linear sparsifying transform $\bd D$, but one could also consider a nonlinear operator, for instance, a neural network. The recovery of $\bld x^*$ is of greater interest than the uniqueness of $\bld\alpha^*$, and subsequently, the recovery guarantees on $\bld x^*$ depend on the properties of the sensing matrix $\bd A$ rather than the combined linear transformation $\bd A\bd D$ \cite{candes_analysis}. If the columns of $\bd D$ constitute a tight frame, i.e., $\bd D\bd D^\TT = \bd I$, then the analysis-sparse recovery problem is posed as the $\ell_1$-norm minimization problem:
\begin{equation}\label{eq:P1-analysis}
   \underset{\boldsymbol{x}\in\bb R^n}{\text{minimize }} \| \bd D^\TT \boldsymbol{x}\|_1, \ \text{subject to }\ \| \mathbf{A}\boldsymbol{x} - \boldsymbol{y} \|_{2} \leq \epsilon,
\end{equation}
where $\epsilon>0$ is an upper bound on the $\ell_2$ data fidelity. In an unconstrained setting, the optimization problem can be written as 
\begin{equation} \label{eq:unconstrained-P1-analysis}
    \underset{\boldsymbol{x}\in\bb R^n}{\text{minimize }} \frac{1}{2}\norm{\bd A\bld x - \bld y}_{2}^2 + \lambda\|\bd D^\TT\bld x\|_1.
\end{equation}
{\color{black} The minimizer of the unconstrained problem for regularization parameter $\lambda>0$ is also the minimizer of the constrained problem in \eqref{eq:P1-analysis} \cite[Proposition~3.2]{foucart2013mathematical}.}\\ 
\indent Sensing matrices with a low coherence guarantee robust sparse recovery \cite{incoherence1,incoherence2,tropp_greed,liu2019alista,elad_opt_proj}. Although random Gaussian sensing matrices are relatively easy to construct and also satisfy the restricted isometry property (RIP), tight-frame sensing matrices offer the minimum mean-squared error (MSE) \cite{chen2013projection,UNTF_sensing}, and are a favorite choice in several applications \cite{tropp_structured_TF,steiner_TF,PISTA_MRI,UNTF_CS}. In this paper, our objective is to perform sparse recovery with an arbitrary sensing matrix, such as the Gaussian sensing matrix, but with \emph{tight-frame-like} recovery performance.\par}
%------------------------------------------------------------------------------------------------------------------------------------------------------------------
{\color{black}Canonical-sparse recovery can be thought of as a special case of analysis-sparse recovery, where the dictionary $\bd D$ is the identity matrix $\bd I$.
Several algorithms have been proposed in the literature for solving the canonical-sparse recovery problem --- greedy methods such as Dantzig selector \cite{candes_dantzing}, orthogonal matching pursuit \cite{mallat1993matching, tropp2007omp}, CoSAMP \cite{needell2009cosamp, baraniuk2010model}; iterative methods such as LASSO algorithm \cite{tibshirani96regression} and basis pursuit denoising \cite{chen2001atomic} for solving \eqref{eq:P1-analysis}; proximal gradient methods such as iterative shrinkage thresholding algorithm (ISTA) \cite{daubechies2004iterative}, its fast variant FISTA \cite{beck2009fast}, approximate message passing (AMP) \cite{donoho2009message, rangan2019vector} for solving \eqref{eq:unconstrained-P1-analysis}. Nonconvex penalties (in place of $\ell_1$) have also been considered to obtain more accurate estimates of the sparse vector \cite{bredies2008iterated, foucart2009sparsest,fan2001variable,zhang2010nearly,chouzenoux2013mm, yin2015minimization, selesnick2017sparse, praveen2019firmnet}. With the recent advent of deep learning, deep-unfolding methods that employ neural networks \cite{gregor2010lista,monga2021algorithm,LAMP,TISTA,eldar1,eldar2,eldar4,eldar5, liu2019alista} have been developed to solve the sparse recovery problem.}
%------------------------------------------------------------------------------------------------------------------------------------------------------------------
\subsection{Notations}
\label{subsec:notations}
Boldface lowercase symbols represent vectors and boldface uppercase symbols denote matrices. $\bld 0$ is the null vector and $\bld 1$ is a vector of all ones. The symbol $x_i$ denotes the $i$\textsuperscript{th} element of the vector $\bld x$, and $A_{ij}$ refers to the $(i,j)$\textsuperscript{th} element of the matrix $\bld A$. The $\ell_1$-norm of the vector $\bld x$ is denoted by $\norm{\bld x}_1$. $(\bld x)_s\in\bb R^n$ is a vector comprising the $s$ largest coefficients (in magnitude) of $\bld x \in \bb R^n$:
\begin{equation} \label{eq:x_s}
	(\bld x)_s\in\arg \min_{\norm{\bld z}_0 \leq s} \norm{\bld x - \bld z}_2.
\end{equation}
A diagonal matrix constructed using a vector $\bld v$ is denoted by $\text{diag}(\bld v)$, and element-wise product (Hadamard product) is denoted by $\odot$. The operators $\text{sgn}, \text{min}, \text{max}$, and $|\cdot|$ refer to the signum, minimum, maximum, and modulus, respectively, and when the arguments are vectors or matrices, the operators must be interpreted in the element-wise sense.
%------------------------------------------------------------------------------------------------------------------------------------------------------------------
\subsection{Recovery Guarantees}
\label{subsec:recovery_gaurantees}
{\color{black}The recovery guarantees of an analysis-sparse vector from undersampled measurements depend on the properties of the sensing matrix $\bd A$. Although it is not possible to perfectly recover all sparse vectors from noisy compressed measurements, one can provide upper bounds on the $\ell_2$ recovery error. These bounds depend on dictionary-restricted isometry constant (D-RIC) of the sensing matrix, sparsity of the signal, and the measurement noise level. We recall the definitions of the RIP \cite{candes} and dictionary-RIP (D-RIP) next \cite{foucart2013mathematical}.
\begin{definition}[Restricted Isometry Property \cite{candes}]
For a sensing matrix $\bd A\in\bb R^{m\times n}$, the s-restricted isometry constant $\delta_s$ of $\bd A$ is the smallest quantity in $(0,1)$ such that
\begin{equation}\label{def:RIP}
	\left( 1 - \delta_s \right) \| \bld x\|_2^2 \leq \| \bd A \bld x\|_2^2 \leq \left( 1 + \delta_s \right) \| \bld x\|_2^2,
\end{equation}
for all $\bld x\in\bb R^n$ such that $\norm{\bld x}_0\leq s$.
\end{definition}
%------------------------------------------------------------------------------------------------------------------------------------------------------------------
\begin{definition}[Dictionary RIP \cite{candes_analysis}]
For a sensing matrix $\bd A\in\bb R^{m\times n}$, and a dictionary $\bd D\in\bb R^{n \times d}$, the s-restricted isometry constant ${\delta}_s$ of $\bd A$ adapted to $\bd D$, is the smallest quantity in $(0,1)$ such that
	\begin{equation}\label{eq:RIP-mod}
		\left(1 - {\delta}_s\right)\| \bd D \bld\alpha\|_2^2 \leq \| \bd A  \bd D \bld\alpha\|_2^2 \leq \left(1+{\delta}_s\right)\|\bd D \bld\alpha\|_2^2
	\end{equation}
for all $\bld \alpha\in\bb R^d$ such that $\norm{\bld \alpha}_0\leq s$.
\end{definition}
The sensing matrix $\bd A$ is said to satisfy D-RIP if $\delta_s$ is sufficiently small for large values of $s$. Random Gaussian matrices satisfy the D-RIP with a high probability provided that $m = \cl O(s \log\frac ds) $ \cite{candes_analysis}. A sparse signal $\bld x$ may be measured using a Gaussian sensing matrix to obtain $\bld y$ and recovered accurately with a high probability when the number of measurements is proportional to the level of sparsity in $\bld x$.
%------------------------------------------------------------------------------------------------------------------------------------------------------------------
Cand\`es {\it et al.} \cite{candes_analysis} provide a bound on the recovery performance when $\bd D$ is a tight frame. We recall the following result from \cite{candes_analysis}.
\begin{theorem}[Performance bounds \cite{candes_analysis}] 
Let $\bd D\in\bb R^{n \times d}$ be a tight frame and $\bd A\in\bb R^{m\times n}$ be a sensing matrix satisfying D-RIP with ${\delta}_{2s} < 0.08$. The minimizer $\hat{\bld x}$ of \eqref{eq:P1-analysis} satisfies 
\begin{equation}
	\| \hat{\bld x} - \bld x^* \|_2 \leq C_0 \epsilon + C_1 \frac{\| \bd D^\TT \bld x^* - \left( \bd D^\TT \bld x^* \right)_s \|_1}{\sqrt{s}},
\end{equation}
where the constants $C_0$ and $C_1$ may only depend on ${\delta}_{2s}$, and where $\left( \bd D^\TT \bld x^* \right)_s$ is the vector comprising the $s$ largest entries of $\bd D^\TT \bld x^*$ in magnitude as in \eqref{eq:x_s}.
\label{thm:mse_candes}
\end{theorem}
The exact values of the constants $C_0$ and $C_1$ can be calculated using the formulas given in the proof of Theorem~\ref{thm:mse_candes} \cite{candes_analysis}. Solving the optimization problem in \eqref{eq:P1-analysis} results in exact recovery in the absence of observation noise provided that the ground-truth signal $\bld x^*$ is exactly $s$-sparse in $\bd D$, i.e.,  $\bd D^\TT \bld x^* = \left( \bd D^\TT \bld x^* \right)_s$.
%------------------------------------------------------------------------------------------------------------------------------------------------------------------
\subsection{Contributions of This Paper}
\label{subsec:contributions}
{\color{black}The key observation that forms the basis for this paper is that solving the system of equations $\bld y = \bd A\bld x$ is equivalent to solving $\bd B\bld y = \bd B\bd A\bld x$, if $\bd B$ is an invertible matrix, with the effective sensing matrix now becoming $\bd{BA}$ and the effective measurement becoming $\bd B\bld y$. The new system of equations can have advantages over the original one if $\bd B$ is chosen appropriately --- this is the crux of our contribution. In the backdrop of this argument, we consider the generalized $\ell_2$-norm-based formulation:
\begin{equation} \tag{P1} \label{eq:P1-B}
\underset{\boldsymbol{x}\in\bb R^n}{\text{minimize }} \| \bd D^\TT \boldsymbol{x}\|_1, \ \text{subject to }\ \| \mathbf{A}\boldsymbol{x} - \boldsymbol{y} \|_{\bd B} \leq \epsilon,
\end{equation}
where $\norm{\cdot}_{\bd B} = \norm{\bd B (\cdot) }_2$ denotes the $\bd B$-norm \cite{horn2012matrix}, $\bd B\in\bb R^{m\times m}$ being a positive-definite matrix. We show that this generalized formulation effectively corresponds to using a tight-frame sensing matrix $(\bd B \bd A)$ by a careful choice of $\bd B$, even if $\bd A$ is not tight. A consequence is that the $\ell_2$-error in the recovery can be lowered without having to design or hand-craft a tight-frame sensing matrix. We then develop and derive corresponding performance bounds (Section~\ref{subsec:performance_bounds}) on the minimizer of the proposed formulation. \eqref{eq:P1-B} can be transformed into an unconstrained problem, which can be solved iteratively using proximal gradient methods. We also propose rescaled gradient versions of these algorithms, which show further performance improvements. Simulation results demonstrate superior recovery accuracy compared to the benchmark techniques (Section~\ref{sec:experiemental_validation}). Proceeding further, we consider the scenario where the sparsifying transform is a deep neural network and present an application to compressed sensing image recovery (CS-IR). The reconstructed image quality obtained using the proposed methods surpasses that of the benchmark CS-IR methods in terms of the standard performance metrics, viz. peak signal-to-noise ratio (PSNR) and structural similarity (SSIM) (Section~\ref{sec:image_recovery}).}
% ------------------------------------------------------------------------------------------------------------------------------------------------------------
% ------------------------------------------------------------------------------------------------------------------------------------------------------------
% 	      									Tight Frames
% ------------------------------------------------------------------------------------------------------------------------------------------------------------
% ------------------------------------------------------------------------------------------------------------------------------------------------------------
\section{Tight Frames}
A set of $n$ vectors $\{\bld v_i \in \bb R^m\}_{i=1}^{n}, \; n\geq m$, is said to constitute a \emph{frame} for $\bb R^m$ if the following partial energy equivalence holds \cite{kovacevic2007lifePart1,kovacevic2007lifePart2,vetterli2014fsp}:
\begin{equation}\label{eq:frame_condition}
    \alpha\norm{\bld x}_2^2 \leq \sum_{i=1}^n |\langle \bld x,\bld v_i\rangle|^2\leq \beta \norm{\bld x}_2^2, \forall \bld x\in\bb R^m,
\end{equation}
where $0<\alpha\leq\beta<\infty$. The largest such $\alpha$ and the smallest such $\beta$ are called the \emph{frame bounds}. For $n > m$, the set of frame vectors $\{\bld v_i \in \bb R^m\}_{i=1}^{n}$ is redundant/overcomplete. An overcomplete system provides more degrees of freedom than required for representing the signal. It allows for redundancy and robustness to perturbations, which is beneficial in various signal processing tasks such as compression, coding, denoising, and feature extraction \cite{aggelos2014UNTF, Tsiligianni2017approximate, cai2014applied, li2014compressed, strohmer2003grassmannian, casazza2006harmonic}.
The set of vectors $\{\bld v_i\}_{i=1}^n$ is said to constitute a \emph{tight frame} if $\alpha=\beta$, and a \emph{Parseval tight frame} if $\alpha=\beta=1$ \cite{kovacevic2007lifePart1,kovacevic2007lifePart2, casazza1999theart}. 
{\color{black} In terms of matrices, the columns of a matrix $\bd V \in \bb R^{m \times n}$ constitute a Parseval tight frame if $\bd V\bd V^\TT = \bd I\in\bb R^{m\times m}$. Effectively, $\bd V$ is the left-inverse of $\bd V^\TT$. This property is reminiscent of orthonormal matrices. However, it is important to note that the converse is not true, i.e., $\bd V^\TT \bd V \neq \bd I\in\bb R^{n\times n}$.}
% ------------------------------------------------------------------------------------------------------------------------------------------------------------
{\color{black} 
The choice of the sensing matrix plays an important role in the recovery of sparse signals in the compressed sensing setting. Incoherence between the columns of the sensing matrix and a smaller singular value spread are important from the perspective of signal recovery. The coherence of the matrix $\bd A = [\bld a_1\;\bld a_2\;\cdots\;\bld a_n] \in \bb R^{m\times n}$ is defined as
\begin{equation} \label{eq:welch_bound}
	\mu \stackrel{\text {def.}}{=} \max _{\substack{1 \leq i, j \leq n \\ i \neq j}} \frac{\left|\langle \bld a_i, \bld a_j\rangle\right|}{\left\|\bld a_i\right\|_2\left\|\bld a_j\right\|_2} \geq \sqrt{\frac{n-m}{m(n-1)}} = \mu_{\text{min}},
\end{equation}
where the inequality follows from the Welch bound \cite{foucart2013mathematical}.
Equiangular tight frames (ETFs) achieve the Welch bound and have a condition number of unity \cite{foucart2013mathematical}, making them ideal for compressed sensing \cite{elad2010}. The gram matrix of an ETF will have 1 along the diagonal, and the off-diagonal entries will all be equal to the Welch bound.
}
Bandeira {\it et al.} \cite{bandeira2013raod} showed that RIP constants of ETFs satisfy $\delta_s \leq (s-1)\mu_{\text{min}}$, where $\mu_\text{min}$ is the Welch bound given in \eqref{eq:welch_bound}.
However, the main challenge lies in the construction of ETFs. Further, ETFs may not even exist for arbitrary choices of $m$ and $n$ \cite{ conway1996packing, lemmens1973equi}. This hurdle is one reason why tight frames are not so prevalent in compressed sensing. Instead, random Gaussian matrices are preferred because they are relatively easy to construct and also satisfy the RIP with a very high probability \cite{foucart2013mathematical}. }
% ------------------------------------------------------------------------------------------------------------------------------------------------------------
% ------------------------------------------------------------------------------------------------------------------------------------------------------------
% 	      								PROBLEM FORMULATION
% ------------------------------------------------------------------------------------------------------------------------------------------------------------
% ------------------------------------------------------------------------------------------------------------------------------------------------------------
 \section{Problem Formulation}
\label{sec:problem_formulation}
Consider the analysis-sparse recovery problem:
\begin{equation}\label{eq:P1-TF}
   \underset{\boldsymbol{x}\in\bb R^n}{\text{minimize }} \| \bd D^\TT \boldsymbol{x}\|_1, \ \text{subject to }\ \| \mathbf{A}\boldsymbol{x} - \boldsymbol{y} \|_{\bd B} \leq \epsilon,
\end{equation}
where $\bd B = \left( \bd A \bd A^\TT \right)^{-\frac{1}{2}}$, $\bd D$ is tight, and $\epsilon > 0$ is small. The sensing matrix $\bd A\in \bb R^{m\times n}$ is assumed to have full row-rank and unit-$\ell_2$-norm columns. For the particular choice of $\bd B = \left( \bd A \bd A^\TT \right)^{-\frac{1}{2}}$, we have 
\begin{align*}
   f_{\bd B}\left(\bld x \right) \overset{\text{def.}}{=}\frac12\norm{\bld y-\bd A\bld x}_{\bd B}^2 = \frac12\norm{\bd A^\dagger(\bld y-\bd A\bld x)}_2^2,
\end{align*}
where $\bd A^\dagger = \bd A^\TT \left( \bd A \bd A^\TT\right)^{-1}$, which is the back-projection data-fidelity term introduced in \cite{tirer2020back, tirer2021convergence}. The constraint is convex, as it is the composition of an affine transformation with an $\ell_2$-norm ball, and the objective is a convex function. Therefore, \eqref{eq:P1-TF} is a convex optimization problem.
% ------------------------------------------------------------------------------------------------------------------------------------------------------------
\subsection{Effective Sensing Matrix}
\label{subsec:eff_sensing_matrix}
The proposed formulation is equivalent to \eqref{eq:P1-analysis} with the \emph{effective sensing matrix} $\bd B\bd A = \left( \bd A\bd A^\TT\right)^{-\frac{1}{2}} \bd A $ $\in\bb R^{m\times n}$, and measurements $\bd B\bld y$. It is important to note that $\bd B \bd A$ constitutes a Parseval tight frame because it satisfies the property: $\left( \bd B\bd A\right)\left( \bd B\bd A\right)^\TT = \bd I$. Chen {{\it et al.}} \cite{chen2013projection} showed that the recovery error in the standard CS setting is minimum when the sensing matrix is a Parseval tight frame.
The formulation in \eqref{eq:P1-TF} is as though a tight sensing matrix has been used, by virtue of using the $\bd B$-norm for the measurement constraint. This gives the best of both worlds --- the ease of constructing a sensing matrix by sampling from a Gaussian distribution $\cl N(\bd 0, m^{-1}\bd I)$, and the performance advantage that accrues from a tight-frame sensing matrix.
Some studies have aimed at obtaining a preconditioned linear systems with an effectively tight sensing matrix.
Tsiligianni {\it et al.} proposed an iterative technique to obtain a preconditioner that results in an approximate incoherent unit-norm tight-frame (UNTF) sensing matrix \cite{aggellos2015precondition}. Alemazkoor and Meidani proposed a similar approach, but with the added constraint of the preconditioner being close to the identity in Frobenius norm, with the objective of obtaining an effective sensing matrix that is close to an ETF \cite{meidani2018precondition}. Such methods for finding the preconditioner are iterative, and each iteration involves solving an optimization problem using conjugate gradient descent. On the other hand, the formulation considered in \eqref{eq:P1-TF} effectively results in a Parseval tight-frame sensing matrix, without requiring iterative optimization.\\
% ------------------------------------------------------------------------------------------------------------------------------------------------------------
{\color{black}\indent For a random Gaussian sensing matrix, the concentration of singular values is given by the following theorem:
\begin{theorem}[Concentration of singular values \cite{foucart2013mathematical}]\label{thm:foucart_singular_value_1}
Let $\bd A^\TT\in\bb R^{n\times m}$ be a Gaussian matrix with $m < n$, and let $\sigma_{\min}$ and $\sigma_{\max}$ denote the smallest and the largest singular values of the renormalized matrix $\frac{1}{\sqrt{n}} \bd A^\TT$. Then, for $t > 0 $
\begin{equation} \label{eq:concentration_inequality}
	\mathbb{P}\left(\sigma_{\max} \geq 1 + \sqrt{\frac{m}{n}}+ t \right) \leq e^{-nt^2/2} \quad \text{and} \quad \mathbb{P}\left(\sigma_{\min} \leq 1 - \sqrt{\frac{m}{n}}- t \right) \leq e^{-nt^2/2}.
\end{equation}
\end{theorem}
In the above theorem, $\bd A^\TT$ is a tall matrix, which matches with the proposed formulation where $\bd A$ is considered to be rectangular. The non-zero singular values of $\bd A$ and $\bd A^\TT$ are the same. We define the {\it restricted condition number} of a matrix as the ratio of the largest and smallest non-zero singular values. In our case, $\bd B$ is chosen such that $\bd B\bd A$ becomes a tight-frame with restricted condition number being unity. On the other hand, the restricted condition number of the random Gaussian sensing matrix in general is greater than one, following the above theorem.}
% ------------------------------------------------------------------------------------------------------------------------------------------------------------
\subsection{Performance Bounds}
\label{subsec:performance_bounds}
In order to establish performance bounds within the proposed tight frame formulation, we introduce the definition of RIP using the  $\bd B$-norm.
% ------------------------------------------------------------------------------------------------------------------------------------------------------------
\begin{definition}{($\bd B$-norm Restricted Isometry)}
The $\bd B$-norm $s$-restricted isometry constant $\hat{\delta}_s$ of a sensing matrix $\bd A\in\bb R^{m\times n}$, where $\bd B = \left(\bd A \bd A^\TT \right)^{{-\frac{1}{2}}}$, is the smallest quantity in $(0,1)$ such that
\begin{equation}\label{def:B-norm-RIP}
	\left( 1 - \hat{\delta}_s \right) \| \bld x\|_2^2 \leq \| \bd A \bld x\|_{\bd B}^2 \leq \| \bld x\|_2^2
\end{equation} 
for all $\bld x\in\bb R^n$ such that $\norm{\bld x}_0\leq s$.
\end{definition}
% ------------------------------------------------------------------------------------------------------------------------------------------------------------
We say that $\bd A$ satisfies the $\bd B$-norm RIP if $\hat{\delta}_s$ is sufficiently small for large values of $s$. The $\bd B$-norm restricted isometry constant (RIC) can be interpreted as the RIC associated with the effective sensing matrix $\bd B \bd A$. Note that, by virtue of the choice of $\bd B = \left(\bd A \bd A^\TT \right)^{{-\frac{1}{2}}}$ and the Parseval tight-frame property of $\bd B\bd A$, the upper bound in the $\bd B$-norm RIP condition in \eqref{def:B-norm-RIP} is different from that in the standard RIP condition in \eqref{def:RIP}. 
The isometry constants, $\delta_s$ in \eqref{def:RIP} and $\hat{\delta}_s$ in \eqref{def:B-norm-RIP}, can be related using the spectral norm of $\bd A \bd A^\TT$. For all $s$-sparse vectors $\bld x \in \bb R^n$, the following holds:
\begin{align}
\| \bd A \bld x\|_{\bd B}^2 = \left\| \left(\bd A \bd A^\TT \right)^{\frac{-1}{2}} \bd A \bld x\right\|_2^2 
\geq \lambda_{\min} \left\{\left(\bd A \bd A^\TT \right)^{-1} \right\}  \| \bd A \bld x\|_2^2
=\frac{\| \bd A \bld x\|_2^2} {\lambda_{\max} \left\{\bd A \bd A^\TT \right\} } = \frac{\| \bd A \bld x\|_2^2} {\| \bd A\|_2^2} \geq \frac{\left( 1- \delta_s \right)} {\| \bd A\|_2^2} \|\bld x\|_2^2, \nonumber 
\end{align}
where $\lambda_{\text{min}}\{\cdot\}$ and $\lambda_{\text{max}}\{\cdot\}$ denote the minimum and maximum eigenvalues, respectively, of the argument. The above calculation results in the following inequality for all $s$-sparse $\bld x$:
\begin{equation}
	 \frac{\left( 1- \delta_s \right)} {\| \bd A\|_2^2} \|\bld x\|_2^2 \leq \| \bd A \bld x\|_{\bd B}^2.
\end{equation}
Invoking the $\bd B$-norm RIP from \eqref{def:B-norm-RIP}, we  obtain 
\begin{equation}\label{eq:delta_relation}
\hat{\delta}_s = 1 - \left( 1 - \delta_s \right) / \| \bd A\|_2^2.
\end{equation} 
\indent The standard RIP condition allows an isometry gap of $2{\delta}_s \|\bld x \|_2^2$. Under the $\bd B$-norm RIP condition, the isometry gap is $ \hat{\delta}_s \|\bld x \|_2^2$. We compare the two isometry gaps to identify the smaller one by considering the difference:
\begin{equation}
\begin{split}
 \hat{\delta}_s - 2{\delta}_s ~=~ 1 - \frac{1-{\delta}_s}{\| \bd A \|_2^2} - 2 {\delta}_s 
 ~=~ \left( 1 - \frac{1}{ \| \bd A \|_2^2 } \right) - \delta_s \left( 2 - \frac{1}{\| \bd A \|_2^2} \right). \nonumber
\end{split}
\end{equation}
When $ \delta_s > \displaystyle \frac{\| \bd A \|_2^2  - 1}{2\| \bd A \|_2^2  - 1}$, we have $\hat{\delta}_s < 2{\delta}_s$, i.e., the isometry gap in the $\bd B$-norm RICs is smaller than that of the vanilla RICs. In this scenario, using the effective sensing matrix $\bd B \bd A$ is better than using $\bd A$. The converse happens when $ \delta_s < \displaystyle \frac{\| \bd A \|_2^2  - 1}{2\| \bd A \|_2^2  - 1}$, in which case $\bd A$ is a better sensing matrix than $\bd B \bd A$. 
% ------------------------------------------------------------------------------------------------------------------------------------------------------------
{\color{black} The signals get denser as $s$ increases. The RICs $\delta_s$ and $\hat{\delta}_s$ increase with $s$. The use of sensing matrices with smaller RICs has been shown to result in superior reconstruction guarantees \cite[Theorem 1.2]{candes}. Therefore, there is clearly a certain regime in which the $\bd B$-norm formulation offers enhanced recovery as the value of $s$ increases, i.e., as the signals get denser.} \\
\indent {\color{black} Tirer and Giryes \cite{tirer2021convergence} examine the restricted minimum eigenvalue of $\mathbf{A}\mathbf{A}^\TT$ using $\ell_2$-loss and back-projection loss. The restricted smallest eigenvalue does not have any dependency on the sparse signal under consideration, whereas the restricted isometry constant is associated with the sparsity of the signal. The isometry gap, $2\delta_s$ or $\hat{\delta}_s$, discussed in this paper is specifically relevant to the compressed sensing scenario. The relation between the restricted minimum eigenvalue ($\lambda_{\text{min}}^r$) and $\bd B$-norm RIC ($\hat{\delta}_s$) can be given as:
\begin{equation}
    \lambda_{\text{min}}^r \norm{\bld x}_2^2 \leq (1 - \hat{\delta}_s) \norm{\bld x}_2^2 \leq  \| \bd A \bld x\|_{\bd B}^2 \leq \| \bld x\|_2^2,
\end{equation}
for all $s$-sparse vectors $\bld x \in \bb R^n$ and $\hat{\delta}_s \in (0,1)$. We observe that $\bd B$-norm RIC offers tighter bounds for sparse recovery compared to restricted minimum eigenvalue of Tirer and Giryes.\\
}
We now define the dictionary counterpart of $\bd B$-norm-based RIP within the tight-frame formulation to handle the analysis-sparse model.
% ------------------------------------------------------------------------------------------------------------------------------------------------------------
\begin{definition}[$\bd B$-norm D-RIP]
	For a sensing matrix $\bd A\in\bb R^{m\times n}$, and dictionary $\bd D\in\bb R^{n \times d}$, the $\bd B$-norm s-restricted isometry constant $\hat{\delta}_s$ of $\bd A$, where $\bd B = \left(\bd A \bd A^\TT \right)^{{-\frac{1}{2}}}$, is the smallest quantity in $(0,1)$ such that
	\begin{equation}\label{def:RIP-mod}
		\left( 1 - \hat{\delta}_s \right) \| \bd D \bld v\|_2^2 \leq \| \bd A  \bd D \bld v\|_{\bd B}^2 \leq \|  \bd D \bld v\|_2^2
	\end{equation}
	holds for all $s$-sparse vectors $\bld v \in \bb R^d$.
\end{definition}
% ------------------------------------------------------------------------------------------------------------------------------------------------------------
A similar relation between $\bd B$-norm D-RIP constant $\hat{\delta}_s$ and D-RIP constant $\delta_s$ holds as with RIP constants in \eqref{eq:delta_relation}, i.e.,  $\hat{\delta}_s = 1 - \left( 1 - \delta_s \right) / \| \bd A\|_2^2$.
For all $s$-sparse vectors $\bld v \in \bb R^n$, the following holds:
\begin{align}
\| \bd A \bd D \bld v\|_{\bd B}^2 &= \left\| \left(\bd A \bd A^\TT \right)^{-1/2} \bd A \bd D \bld v\right\|_2^2 
\geq \lambda_{\min} \left\{\left(\bd A \bd A^\TT \right)^{-1} \right\}  \| \bd A \bd D \bld v\|_2^2, \nonumber \\
&=\frac{\| \bd A \bd D \bld v\|_2^2} {\lambda_{\max} \left\{\bd A \bd A^\TT \right\} } = \frac{\| \bd A \bd D \bld v\|_2^2} {\| \bd A\|_2^2} \geq \frac{\left( 1- \delta_s \right)} {\| \bd A\|_2^2} \|\bd D \bld v\|_2^2.\nonumber 
\end{align}
{\color{black}The D-RIP for Gaussian sensing matrices is shown to hold with high probability when $m = \cl O(s\log(d/s))$ \cite[Lemma 2.1]{rauhut2008CS}, \cite{candes_analysis}. }
The counterpart of \eqref{eq:P1-analysis} relying on the $\bd B$-norm of the error takes the following form:
\begin{equation}
   \underset{\boldsymbol{x}\in\bb R^n}{\text{minimize }} \| \bd D^\TT \boldsymbol{x}\|_1, \ \text{subject to }\ \| \mathbf{A}\boldsymbol{x} - \boldsymbol{y} \|_{\bd B} \leq \epsilon.
\end{equation}
The corresponding performance bound, following that of Cand\`es {\it et al.} \cite{candes_analysis}, is provided next.
% ------------------------------------------------------------------------------------------------------------------------------------------------------------
\begin{theorem}\label{thm:MSE_D}
Let $\bd D\in\bb R^{n\times d}$ be an arbitrary tight frame and $\bd A\in\bb R^{m\times n}$ be a sensing matrix satisfying the $\bd B$-norm D-RIP with $\hat{\delta}_{2s} < 0.08$. Then, the solution $\hat{\bld x}$ to the optimization problem in \eqref{eq:P1-TF} satisfies
\begin{equation} 
    \norm{\hat{\bld x} - \bld x^*}_2 \leq C_{0}^{\textsc{TF}} \epsilon + C_{1}^{\textsc{TF}} \frac{\norm{\bd D^\TT \bld x^* - \left( \bd D^\TT \bld x^* \right)_s}_1}{\sqrt{s}},
\end{equation}
where $\bld x^*$ is the ground-truth signal, $C_{0}^{\textsc{TF}}$ and $C_{1}^{\textsc{TF}}$ are positive numerical constants that depend on $\hat{\delta}_{2s}$, and $\left( \bd D^\TT \bld x^* \right)_s$ denotes the vector comprising the $s$ largest entries (magnitude-wise) of $ \bd D^\TT \bld x^*$ as defined in \eqref{eq:x_s}.
\end{theorem}
% ------------------------------------------------------------------------------------------------------------------------------------------------------------
\begin{figure}[t]
    \begin{center}
		\begin{tabular}[t]{P{.22\linewidth}P{.22\linewidth}P{.22\linewidth}P{.22\linewidth}}
			$\quad \  1\%$ & $\quad  2\%$ & $\quad  5\%$ & $\quad  10\%$ \\
			\cellcolor{white}{\includegraphics[width=1.0\linewidth]{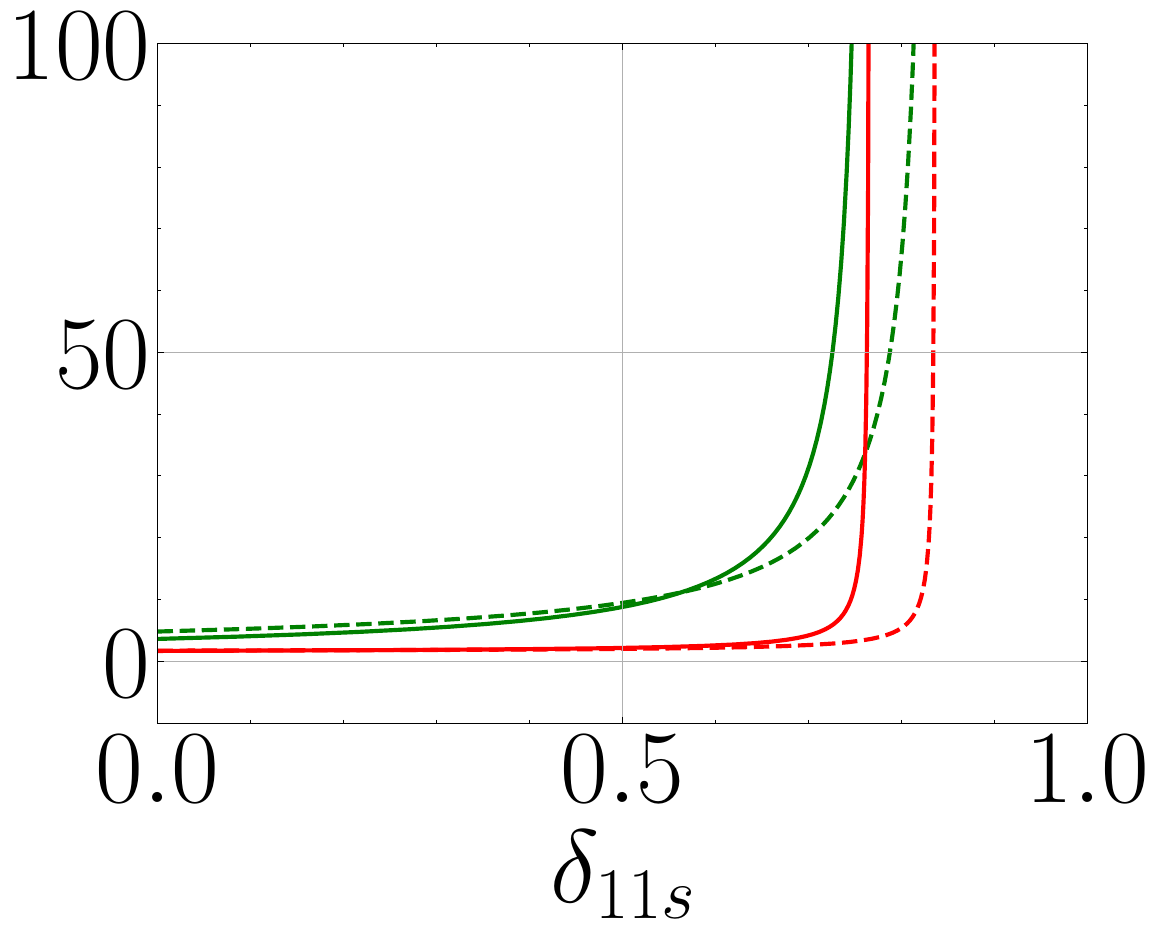}}&
			\cellcolor{white}{\includegraphics[width=1.0\linewidth]{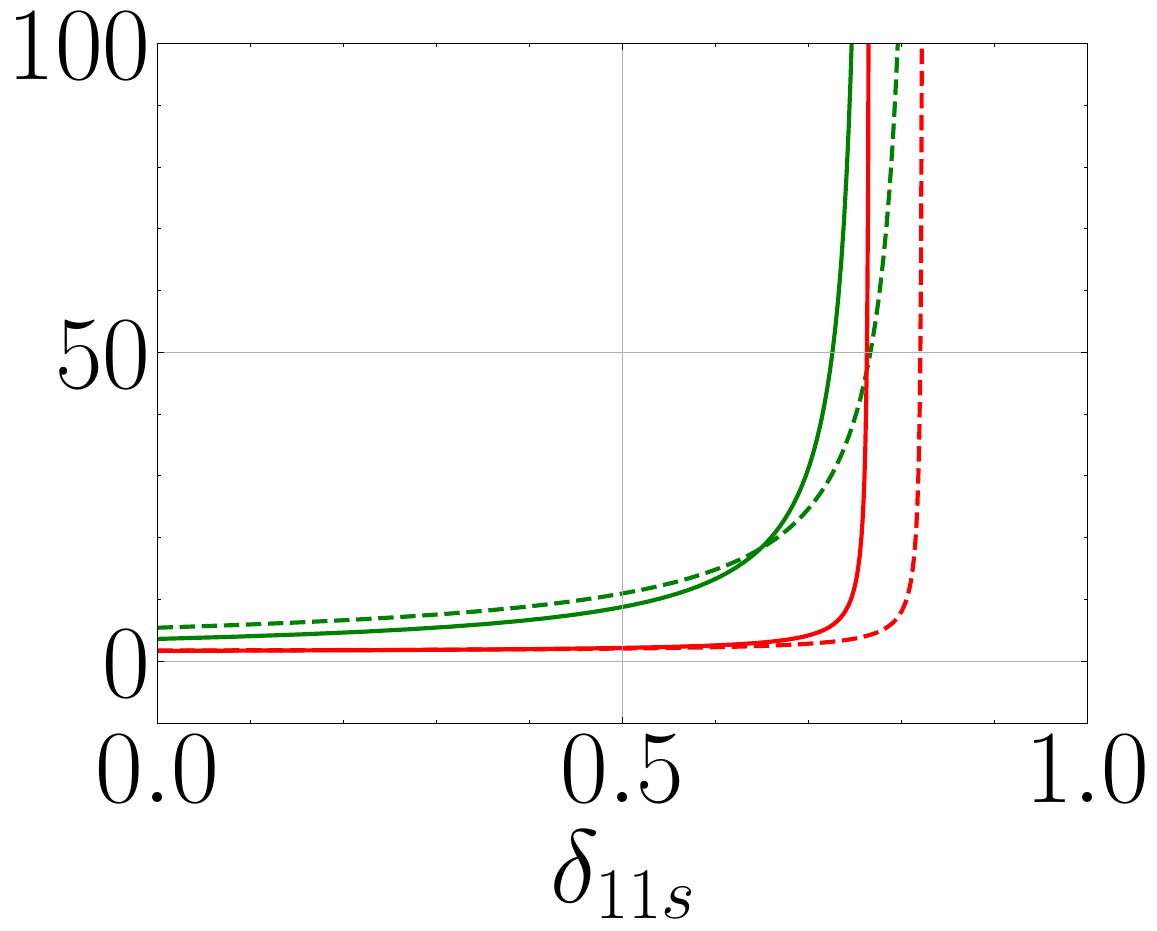}}&
			\cellcolor{white}{\includegraphics[width=1.0\linewidth]{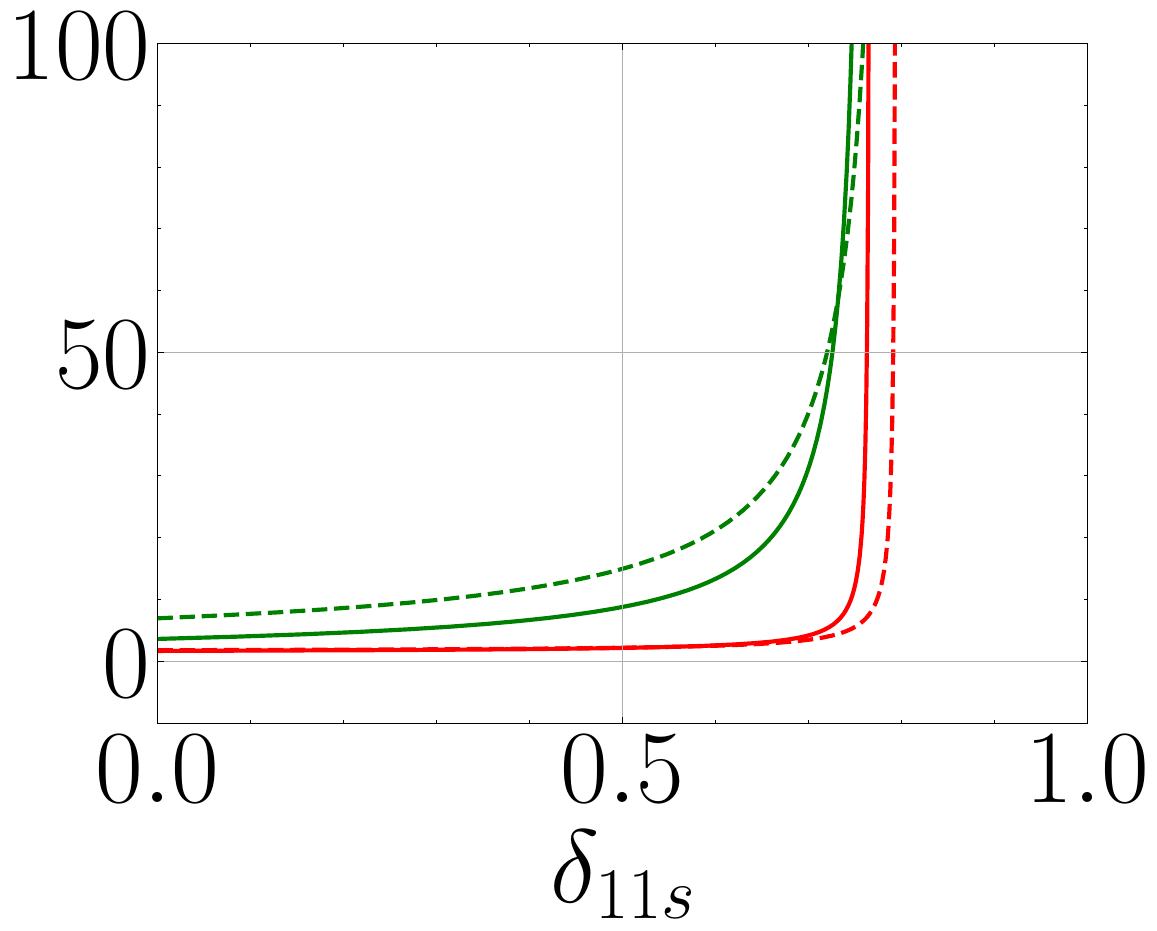}}&
            \cellcolor{white}{\includegraphics[width=1.0\linewidth]{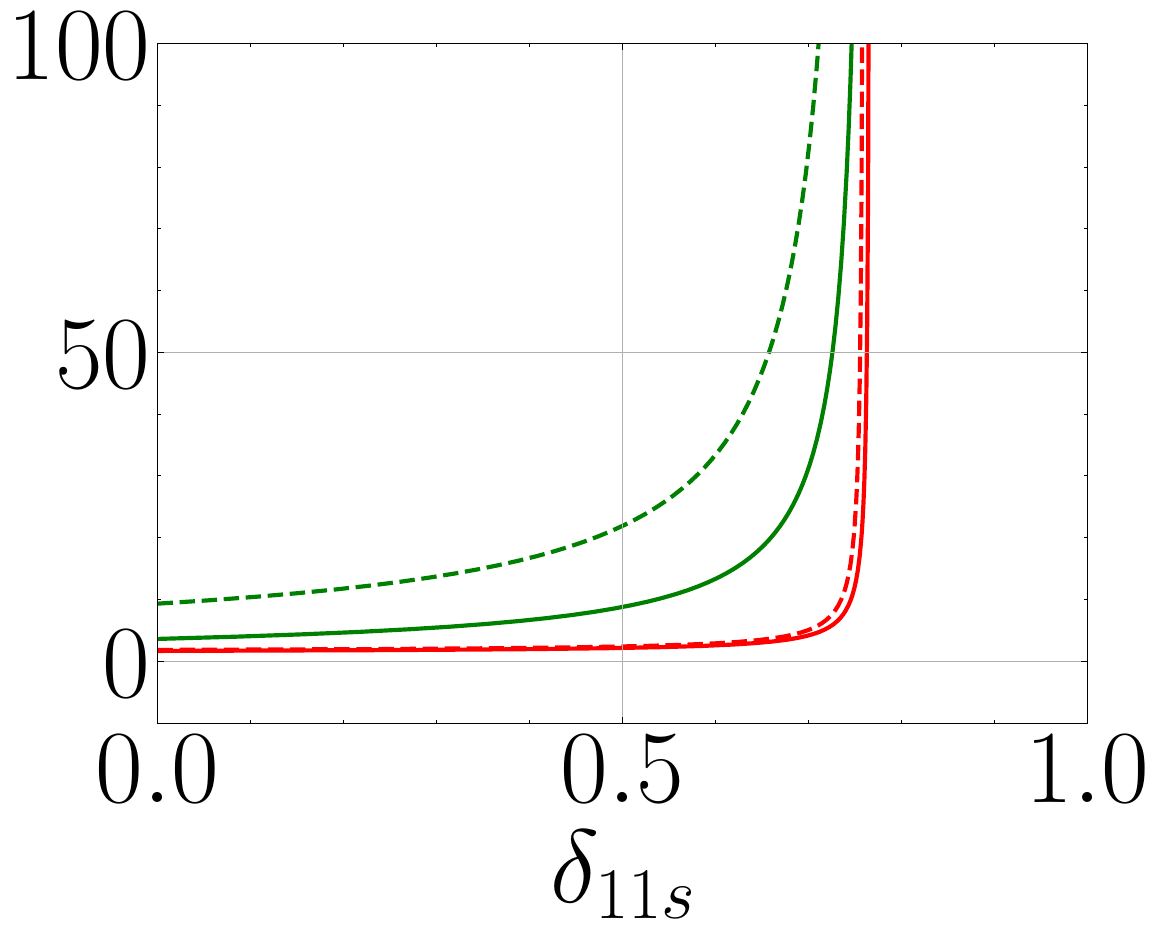}}\\
            \cellcolor{white}{\includegraphics[width=1.0\linewidth]{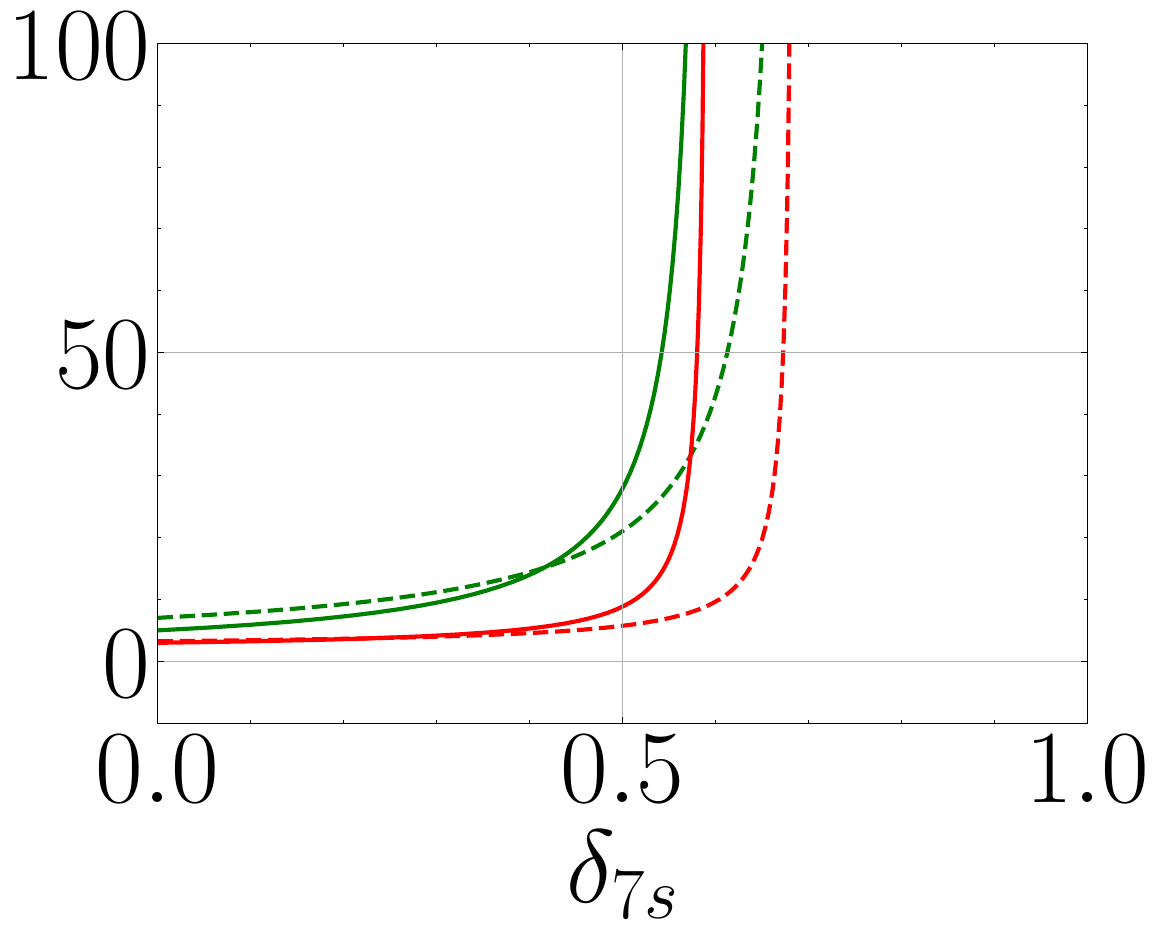}}&
			\cellcolor{white}{\includegraphics[width=1.0\linewidth]{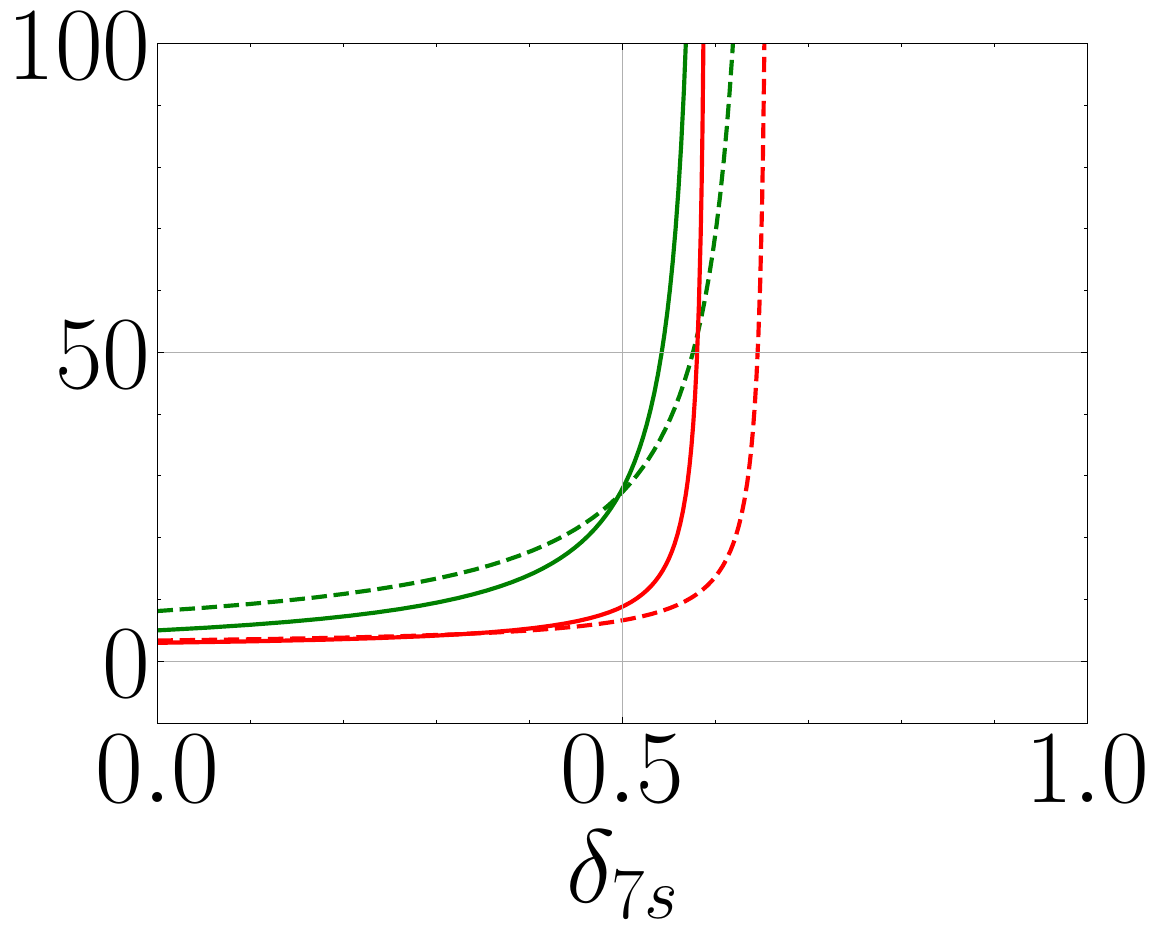}}&
			\cellcolor{white}{\includegraphics[width=1.0\linewidth]{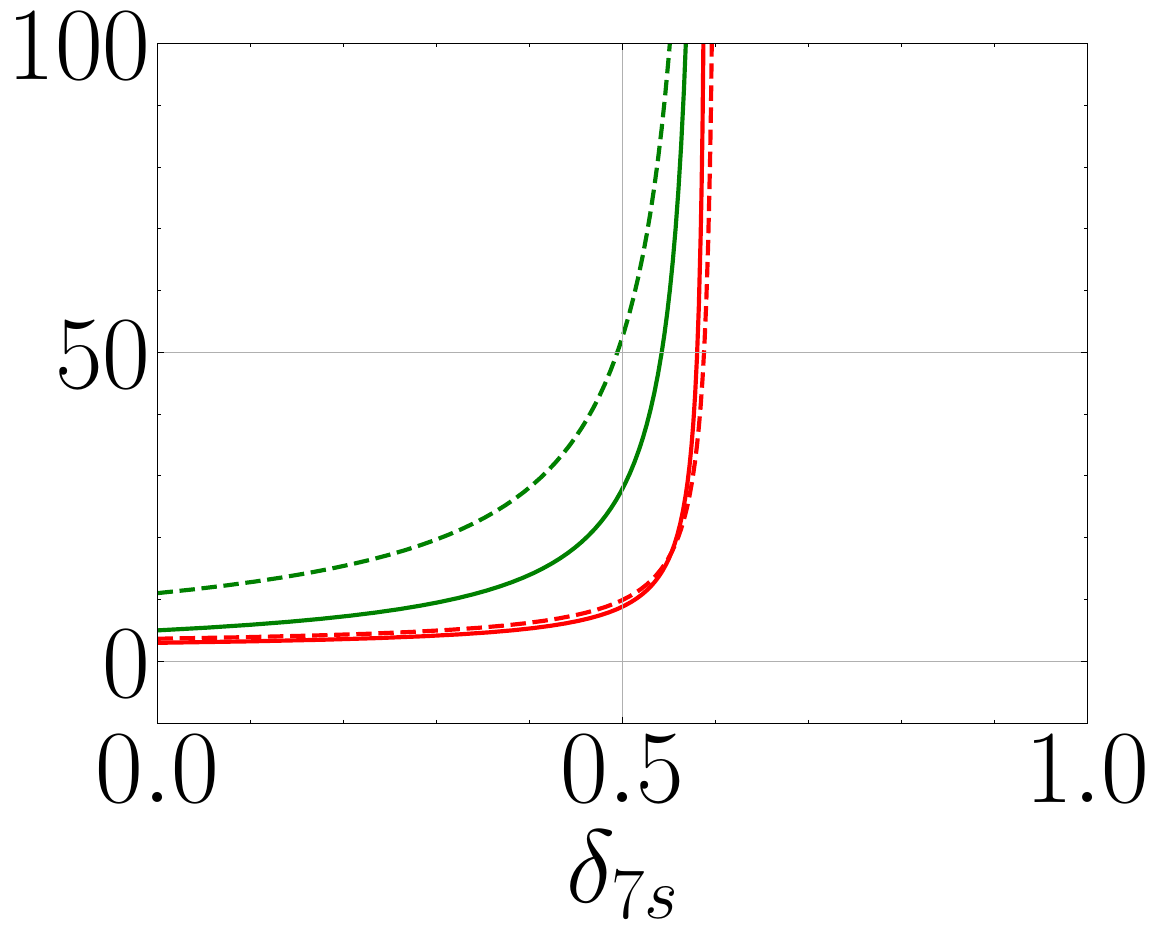}}&
            \cellcolor{white}{\includegraphics[width=1.0\linewidth]{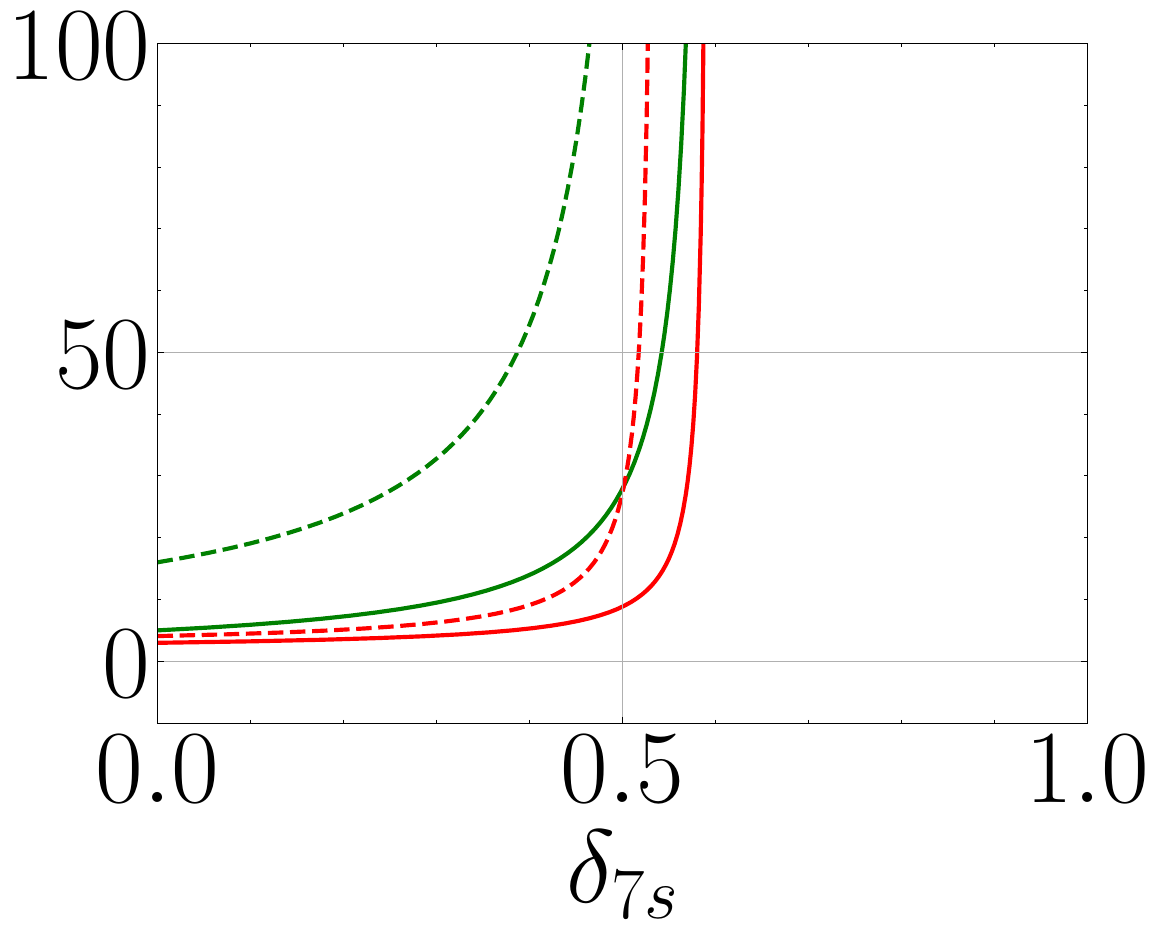}}\\
            &\multicolumn{2} {c} {\includegraphics[width=0.45\linewidth]{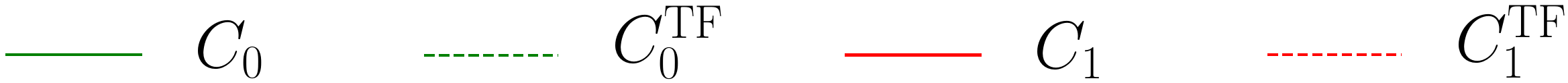}} &\\
		\end{tabular} 
 		\caption[]{Variation of the constants $C_0$, $C_1$, $C_0^\textsc{TF}$ and $C_1^\textsc{TF}$ versus the RIC for various compression factors: $C_0$ and $C_0^\textsc{TF}$ correspond to the error due to the noise in the observation, while $C_1$ and $C_1^\textsc{TF}$ correspond to the sparsity model mismatch. Smaller values of the constants correspond to a superior reconstruction. We observe that the value of $C_0^{\textsc{TF}}$ is smaller than $C_0$ for higher compression ratios. The value of $C_1^{\textsc{TF}}$ is smaller than $C_1$ for several scenarios.}
		\label{fig:phase_plots}
	\end{center}
\end{figure}
%% ------------------------------------------------------------------------------------------------------------------------------------------------------------
\textcolor{black}{The proof of Theorem~\ref{thm:MSE_D} is provided in Appendix~\ref{appendix:mse_bounds} and follows more or less along the same lines as \cite{candes_analysis}, but with minor differences. 
Our proof provides adapted versions of the lemmas pertaining to restricted isometry constants, tube constraint, and the D-RIP. The revised lemmas are coherently integrated in the main proof. A comparison between the constants derived from the proposed optimization problem \eqref{eq:P1-TF} and that from \eqref{eq:P1-analysis} in \cite{candes_analysis} is performed by invoking inequalities pertaining to the concentration of singular values of the sensing matrix \cite{foucart2013mathematical}}.\\
\indent {It is known that $m \times n$ Gaussian matrices with $m = \cl O \left( s \log{ \left( \frac{n}{s}\right)} \right)$ satisfy the RIP \cite{foucart2011hard}.} This result is similar to Theorem~\ref{thm:mse_candes}, where the constants $C_{0}^{\text{TF}}$, $C_{1}^{\text{TF}}$ depend on the $\bd B$-norm RIP constant $\hat{\delta}_s$, while $C_0$ and $C_1$ depend on the compression ratio $\frac{m}{n}$ and $\hat{\delta}_s$. We analyze the constants as functions of ${\delta}_s$ for various compression ratios. Figure~\ref{fig:phase_plots} shows the upper bounds for the constants $C_0$, $C_1$, $C_{0}^{\text{TF}}$, $C_{1}^{\text{TF}}$ plotted as functions of RICs for various compression ratios. \\
\indent {\color{black} The exact expressions for $C_0^{\text{TF}}$ and $C_1^{\text{TF}}$ are given in Appendix \ref{appendix:mse_bounds} (\eqref{eq:final_expression_for_perf_bounds}). The expressions for $C_0$ and $C_1$ can be found in \cite{candes_analysis}. 
A comparison of the constants arising from both formulations \eqref{eq:P1-analysis} and \eqref{eq:P1-TF} is given in Fig.~\ref{fig:phase_plots}. 
The constants depend on the sparsity factor through the corresponding RICs. The reason for plotting the constants vs. $\delta_{7s}$ and $\delta_{11s}$ will become clear from Appendix \ref{appendix:mse_bounds} (the text following \eqref{eq:final_expression_for_perf_bounds}).\\
\indent We observe that when the compression ratio is high, the constants $C_{0}^{\text{TF}}$ and $C_{1}^{\text{TF}}$ are either comparable or smaller than $C_0$ and $C_1$, respectively, which makes the upper bound on the $\ell_2$ error in the reconstruction smaller. Further, in the high compression ratio regime, the constant $C_1^{\text{TF}}$ is consistently smaller than $C_1$, indicating superior recovery guarantees  for signals that are not exactly compressible through the proposed TF formulation. On the other hand, when the compression ratio is small, the constants $C_{0}^{\text{TF}}$ and $C_{1}^{\text{TF}}$ are larger than $C_0$ and $C_1$, respectively. Hence, the reconstruction quality of \eqref{eq:P1-analysis} is superior to that of proposed formulation \eqref{eq:P1-TF}. Effectively, when the measurements are few, the proposed formulation is certainly a better choice than the standard formulation \eqref{eq:P1-analysis} from the point of view of reconstruction accuracy.\\ 
\indent {\color{black} Tirer and Giryes \cite{tirer2021convergence} examined the relationship between compression ratio and recovery performance considering both $\ell_2$-loss and back-projection (BP) loss, focusing mainly on the $\ell_2$-regularizer, whereas our approach considers the $\ell_1$-prior, specifically with the objective of sparse signal recovery. The BP loss is shown to have an advantage over the least-squares (LS) loss when the noise level is low and when the $\ell_2$ prior is employed \cite{tirer2018image}. Our analysis from Figure~\ref{fig:phase_plots} supports this relation in the case of the analysis-sparse $\ell_1$ prior. The iterate convergence results pertaining to a lower semi-continuous prior with back-projection loss are given in \cite[Theorem 3.1]{tirer2021convergence}, where the convergence of the sequence resulting from \cite[Equation 2.5]{tirer2021convergence} is shown by \cite[Theorem 3.1]{tirer2021convergence}. 
\cite[Theorem 3.1]{tirer2021convergence} is specific to the proximal gradient-descent algorithm and the sequence generated by the approach described in \cite[(2.5)]{tirer2021convergence}. On the contrary, Theorem \ref{thm:MSE_D} depends on the optimization problem and is not specific to an algorithm employed to solve the problem.}}

% %------------------------------------------------------------
\section{Experimental Validation}
\label{sec:experiemental_validation}
Consider signals of length $n = 1024$ and number of measurements $m = 500$ using sensing matrix $\mathbf{A}\in \bb R^{m\times n}$. We consider an undersampled orthogonal discrete cosine transform (DCT) as the dictionary $\mathbf D \in \bb R^{n \times d}$, where $d = 4n$. The ground-truth signals $\bld x^* \in \bb R^n$ are generated as follows: $\bld x^* = \bd D \bld \alpha$, where $\bld \alpha \in \bb R^d$ is a canonical sparse vector. We generate the canonical sparse signal $\boldsymbol{\alpha}^*$ with support determined by a Bernoulli distribution with parameter $p$ and output symbols $\{0, 1\}$, and the amplitudes of the nonzero entries drawn from the standard normal distribution. We consider a Gaussian sensing matrix with columns rescaled to have unit $\ell_2$-norm. The noisy measurements are modeled as $\bld y = \bd A\bld x^* + \bld w$, where $\bld w$ has i.i.d. Gaussian entries with zero mean and variance $\sigma^2$. The measurement signal-to-noise ratio (SNR) is defined as $\text{SNR} = 20 \log_{10} \left( \frac{  \| {\bd A\bld x^*} \|_2}{ \| \bld w \|_2} \right) \text{ dB}$. 
% %------------------------------------------------------------
\paragraph{Algorithms and Methods}
\textcolor{black}{The benchmark methods used for the analysis-sparse recovery are ISTA \cite{daubechies2004iterative, chaari2009solving, elad2007analysis, selesnick2009signal,weiss2009efficient}, Loris {\it et al.} \cite{Loris2011generalization}, NESTA \cite{becker2011NESTA}, and SFISTA \cite{tan2014smoothing}. ISTA can be used for solving analysis-sparse recovery problems only if the dictionary $\bd D$ is tight \cite{chaari2009solving, elad2007analysis, selesnick2009signal,weiss2009efficient}. Loris {\it et al.}'s method (henceforth, referred to as Loris) is applicable for non-separable penalties, unlike ISTA. Each iteration of Loris involves four matrix-vector products (whereas ISTA requires two) and a thresholding step. NESTA integrates Nesterov's smoothing technique with Nesterov's accelerated first-order algorithm. NESTA has the added advantage over the other methods in accounting for the book-keeping of the past gradients. SFISTA employs smoothing of the monotone version of FISTA (MFISTA), since MFISTA does not have a closed-form solution for the analysis-sparse model. The smoothing relaxes the original sparse recovery problem and implements MFISTA on the relaxed formulation. \\
We incorporated the TF methodology into two of the benchmark methods, namely ISTA and Loris, resulting in four variants: TF-ISTA, RTF-ISTA, TF-Loris, and RTF-Loris. We give the details of these methods next. 
\paragraph{ISTA variants} The iterative shrinkage-thresholding algorithm (ISTA) solves the \eqref{eq:unconstrained-P1-analysis} optimization problem using proximal gradient method \cite{chaari2009solving, elad2007analysis, selesnick2009signal,weiss2009efficient}. Starting with the initialization, $\bld x^{(0)} = \bd A^\TT \bld y$, step-size parameter $\eta = 1/\norm{\bd A}_2^2$, and threshold parameter $\lambda$, the reconstruction is updated as follows:
\begin{equation} \label{eq:ista_update}
	\bld{x}^{(k+1)} = \mathcal{T}_{\eta\lambda} \left( \boldsymbol{x}^{(k)} - \eta \mathbf{A}^\TT \left(\mathbf{A}\boldsymbol{x}^{(k)} - \boldsymbol{y} \right) \right),
\end{equation}
where $\mathcal{T}_{\eta\lambda}$ is the shrinkage function given by
\begin{equation} \label{eq:soft_threshold}
	\mathcal{T}_{\eta\lambda}(\bld x) = \text{sgn}(\bld x) \cdot \text{max}(|\bld x| - \eta\lambda \bld 1, \bld 0),
\end{equation}
with the operations considered element-wise.\\
\indent ISTA can be used for solving the analysis-sparse recovery problem in \eqref{eq:unconstrained-P1-analysis} if the dictionary $\bd D$ is tight. The update takes the form
\begin{equation} \label{eq:ISTA-update}
	\displaystyle \boldsymbol{x}^{(k+1)} = \bd D \mathcal{T}_{\eta\lambda} \left( \bd D^\TT \left( \boldsymbol{x}^{(k)} - \eta \mathbf{A}^\TT (\mathbf{A}\boldsymbol{x}^{(k)} - \boldsymbol{y}) \right) \right),	
\end{equation} 
where $0 < \eta < \frac 1L $ is the step-size parameter ($L = \norm{\bd A}_2$). The tight-frame ISTA (TF-ISTA) solves the unconstrained formulation of \eqref{eq:P1-TF}:
\begin{equation} \label{eq:P1-TF-unconstrained}
    \underset{\boldsymbol{x}\in\bb R^n}{\text{minimize }} \norm{\bd A \bld x - \bld y}_{\bd B}^2 + \lambda \norm{\bd D^\TT \bld x}_1,
\end{equation}
where $\lambda$ is the regularization parameter. TF-ISTA update is given by  
\begin{equation} \label{eq:TF-ISTA-update}
	\displaystyle \boldsymbol{x}^{(k+1)} = \bd D \mathcal{T}_{\eta\lambda} \left( \bd D^\TT \left( \boldsymbol{x}^{(k)} - \eta \mathbf{A}^\TT \left(\mathbf{A}\mathbf{A}^\TT\right)^{-1}  (\mathbf{A}\boldsymbol{x}^{(k)} - \boldsymbol{y}) \right) \right),	
\end{equation}
where $0 < \eta < 1 $ is the step-size parameter. The Lipschitz constant of  $\nabla f_{\bd B}(\bld x)$ is unity \cite{tirer2018image, tirer2020back}. We observe that the TF-ISTA updates from \eqref{eq:TF-ISTA-update} can not be obtained from the ISTA updates in \eqref{eq:ISTA-update} through pre-conditioning \cite{tirer2020back}. 
{\color{black} The gradient updates of TF-ISTA are similar to the gradient updates from \cite{liao2014generalized}, where the generalized alternating projection (GAP) method is combined with weighted $\ell_{2,1}$ prior for image and video compressed sensing. 
The orthogonal AMP (OAMP) method in \cite{ma2017orthogonal} is based on de-correlated linear estimation and divergence-free non-linear estimation for state evolution, whose linear updates coincide with those in \eqref{eq:TF-ISTA-update} for $\bd D = \bd I$.}\\
\indent The rescaled tight-frame ISTA (RTF-ISTA) update is given by
\begin{equation} \nonumber
	\displaystyle \boldsymbol{x}^{(k+1)} = \bd D \mathcal{T}_{\eta\lambda} \left( \bd D^\TT \left( \boldsymbol{x}^{(k)} - \eta \bd C^{-1} \mathbf{A}^\TT \left(\mathbf{A}\mathbf{A}^\TT\right)^{-1}  (\mathbf{A}\boldsymbol{x}^{(k)} - \boldsymbol{y}) \right) \right),
\end{equation}
where $\bd C = \left(\bd A^\dagger\bd A \right) \odot \bd I$, where $\odot$ denotes the Hadamard/element-wise product.
%------------------------- Loris variants ----------------------
\paragraph{Loris variants} The algorithm uses fixed-point solutions to the variational equations corresponding to the optimization problem. A proximal ascent step follows a gradient descent in the algorithm. In the end, another gradient-descent step completes the algorithm. The updates of Loris  can be found in \cite{Loris2011generalization}. The updates of the TF variant of Loris (TF-Loris) are given below:
\begin{equation} \label{eq:TF_Loris}
    \begin{split}
        \Bar{\bld x}^{(k+1)} ~&=~ {\bld x}^{(k)} - \eta \mathbf{A}^\dagger  (\mathbf{A}\boldsymbol{x}^{(k)} - \boldsymbol{y}) - \bd D^\TT {\bld w}^{(k)}, \\
        {\bld w}^{(k+1)} ~&=~ {\bld w}^{(k)} + \bd D \Bar{\bld x}^{(k+1)} - \cl T_{\lambda}\left({\bld w}^{(k)} + \bd D \Bar{\bld x}^{(k+1)} \right),\\
        {\bld x}^{(k+1)} ~&=~ {\bld x}^{(k)} - \eta \mathbf{A}^\dagger  (\mathbf{A}\boldsymbol{x}^{(k)} - \boldsymbol{y}) - \bd D^\TT {\bld w}^{(k+1)}.
    \end{split}
\end{equation}
Similarly, the RTF-Loris updates are given by
\begin{equation} \label{eq:RTF_Loris}
    \begin{split}
        \Bar{\bld x}^{(k+1)} ~&=~ {\bld x}^{(k)} - \eta \bd C^{-1}\mathbf{A}^\dagger  (\mathbf{A}\boldsymbol{x}^{(k)} - \boldsymbol{y}) - \bd D^\TT {\bld w}^{(k)}, \\
        {\bld w}^{(k+1)} ~&=~ {\bld w}^{(k)} + \bd D \Bar{\bld x}^{(k+1)} - \cl T_{\lambda}\left({\bld w}^{(k)} + \bd D \Bar{\bld x}^{(k+1)} \right),\\
        {\bld x}^{(k+1)} ~&=~ {\bld x}^{(k)} - \eta \bd C^{-1} \mathbf{A}^\dagger  (\mathbf{A}\boldsymbol{x}^{(k)} - \boldsymbol{y}) - \bd D^\TT {\bld w}^{(k+1)}.
    \end{split}
\end{equation}
}
%------------------------- NESTA variants ----------------------
{\color{black}\paragraph{NESTA variants} The algorithm uses Nesterov's smoothing technique and Nesterov's accelerated first-order algorithm. 
Let $g(\bld x) = \norm{\bd D^\TT \bld x}_1$, and $g_{\mu}(\bld x)$ be a smooth approximation of $g(\bld x)$ as considered in \cite{becker2011NESTA}. The updates of NESTA can be found in \cite{becker2011NESTA}. 
The TF variant of NESTA (TF-NESTA) updates are given below ($\beta_k = \frac{1}{2(k+1)}$ and ${\tau}_{k} = \frac{2}{k+3}$):
\begin{equation} \label{eq:TF_NESTA}
    \begin{split}
        {\bld u}^{(k+1)} ~&=~ \left ( \bd I - \frac{\lambda}{\lambda + 1} \bd A^\dagger \bd A \right ) \left( \lambda \bd A^\dagger \bld y + {\bld x}^{(k)} - \nabla g_{\mu} (\bld x^{(k)}) \right),\\
        % {\bld y}^{(k+1)} ~&=~ \arg \min_{\bld x} \frac{L}{2}\norm{x - \bld x^{(k)}}_2^2 + \left \langle \nabla f_{\bd B} (\bld x^{(k)}) , \bld x - \bld x^{(k)} \right \rangle \\
        {\bld z}^{(k+1)} ~&=~ \arg \min_{\bld x : \norm{\bd A \bld x - \bld y}_{\bd B} \leq \epsilon} \norm{\bd D^\TT \bld x}_1 + \sum_{i=1}^k \beta_i  \langle \nabla g_{\mu} (\bld x^{(i)}) , \bld x - \bld x^{(i)}  \rangle, \\
        {\bld x}^{(k+1)} ~&=~  {\tau}_{k+1} {\bld z}^{(k+1)} + (1- {\tau}_{k+1}) \bld u^{(k+1)}.
    \end{split}
\end{equation}
Similarly, the RTF-NESTA updates ($L = \norm{\bd C^{-1} \bd A^\dagger \bd A}_2$) are given by
\begin{equation} \label{eq:RTF_NESTA}
    \begin{split}
        {\bld u}^{(k+1)} ~&=~ \left ( \bd I - \frac{\lambda}{\lambda + L} \bd C^{-1}\bd A^\dagger \bd A \right ) \left( \frac{\lambda}{L}\bd C^{-1} \bd A^\dagger \bld y + {\bld x}^{(k)} - \frac 1L \nabla g_{\mu} (\bld x^{(k)}) \right), \\
        {\bld z}^{(k+1)} ~&=~ \arg \min_{\bld x : \norm{\bd A \bld x - \bld y}_{\bd B} \leq \epsilon} \norm{\bd D^\TT \bld x}_1 + \sum_{i=1}^k \beta_i  \langle \nabla g_{\mu} (\bld x^{(i)}) , \bld x - \bld x^{(i)}  \rangle, \\
        {\bld x}^{(k+1)} ~&=~  {\tau}_{k+1} {\bld z}^{(k+1)} + (1- {\tau}_{k+1}) \bld u^{(k+1)}.
    \end{split}
\end{equation}
%------------------------- SFISTA variants ---------------------
\paragraph{SFISTA variants} SFISTA solves the monotone FISTA problem using a smoothing procedure on non-smooth $\ell_1$ penalty. The updates of SFISTA can be found in \cite{tan2014smoothing}. The updates of the TF variant of SFISTA (TF-SFISTA) are given below: 
\begin{equation} \label{eq:TF_SFISTA}
    \begin{split}
        {\bld z}^{(k+1)} ~&=~ {\bld u}^{(k)} - \eta \left( \bd A^\dagger \left( \bd A \bld {\bld u}^{(k)} - \bld y \right) + \frac{1}{\mu} \bd D \left ( \bd D^\TT \bld x^{(k)} - \cl T_{\mu \lambda}\left(\bd D^\TT \bld x^{(k)} \right) \right ) \right ), \\
        {t}^{(k+1)} ~&=~ \frac{1+\sqrt{1+(2{t}^{(k)})^2}}{2},\\
        {\bld x}^{(k+1)} ~&=~ \arg \min_{ \bld x \in \{\bld x^{(k)}, \bld z^{(k+1)}\}} f_{\bd B}\left( \bld x\right) + g_{\mu}\left(\bd D^\TT \bld x\right),\\
        {\bld u}^{(k+1)} ~&=~ {\bld x}^{(k+1)} + \frac{{t}^{(k)}}{{t}^{(k+1)}} \left( {\bld z}^{(k+1)} - {\bld x}^{(k+1)}\right) + \frac{{t}^{(k)} - 1}{{t}^{(k+1)}} \left( {\bld x}^{(k+1)} - {\bld x}^{(k)}\right).
    \end{split}
\end{equation}
Similarly, the RTF-SFISTA updates are given by
\begin{equation} \label{eq:RTF_SFISTA}
    \begin{split}
        {\bld z}^{(k+1)} ~&=~ {\bld u}^{(k)} - \eta \left( \bd C^{-1} \bd A^\dagger \left( \bd A \bld {\bld u}^{(k)} - \bld y \right) + \frac{1}{\mu} \bd D \left ( \bd D^\TT \bld x^{(k)} - \cl T_{\mu \lambda}\left(\bd D^\TT \bld x^{(k)} \right) \right ) \right ), \\
        {t}^{(k+1)} ~&=~ \frac{1+\sqrt{1+(2{t}^{(k)})^2}}{2},\\
        {\bld x}^{(k+1)} ~&=~ \arg \min_{ \bld x \in \{\bld x^{(k)}, \bld z^{(k+1)}\}} f_{\bd B}\left( \bld x\right) + g_{\mu}\left(\bd D^\TT \bld x\right),\\
        {\bld u}^{(k+1)} ~&=~ {\bld x}^{(k+1)} + \frac{{t}^{(k)}}{{t}^{(k+1)}} \left( {\bld z}^{(k+1)} - {\bld x}^{(k+1)}\right) + \frac{{t}^{(k)} - 1}{{t}^{(k+1)}} \left( {\bld x}^{(k+1)} - {\bld x}^{(k)}\right).
    \end{split}
\end{equation}
}
% %----------------------------------------------- Table SNR 30 --------------------------------------------------------
\begin{table*}[t]
	\centering
	\caption{RSNR (in dB) $\pm$ standard deviation for analysis-sparse signal recovery using various competing methods for SNR = 30 dB. TF and RTF variants show performance improvements over benchmark methods for several sparsity scenarios. The best performance is shown in \textbf{boldface} and the second best is shown \underline{underlined}.}
	\label{table:SNR_30}
	\vskip -0.1in
	\begin{center}
		\begin{small}
			\begin{sc}
			\resizebox{0.85\linewidth}{!}{	\begin{tabular}{c||ccccc}
			\hline
					\cellcolor{white} & & & \textbf{Sparsity} & &   \\ 
					\hline
					\textbf{Method} & 1$\%$ & 3$\%$ & 5$\%$ & 7$\%$ & 10$\%$  \\
					\hline
                    ISTA  & $\underline{29.74} \pm \underline{1.25}$ & $23.60 \pm 1.27$ & $14.31 \pm 2.13$ & $9.10 \pm 1.45$ & $5.52 \pm 0.55$  \\
					NESTA  & ${15.70} \pm {1.15}$ & ${9.82} \pm {0.81}$ & ${6.50} \pm {0.71}$ & $4.93 \pm 0.58$ & $3.61 \pm 0.33$   \\
					SFISTA  & $22.66 \pm 1.16$ & $ 20.96 \pm 0.90 $ & $ 15.54 \pm 1.82 $ & $ 10.37 \pm 1.74$ & $ \textbf{5.88} \pm \textbf{0.68} $  \\
					Loris & ${\textbf{30.20}} \pm {\textbf{1.28}}$ & $24.44 \pm 1.18$ & $14.89 \pm 2.17$ & $9.37 \pm 1.52$ & $5.58 \pm 0.57$ \\
					\hline
					TF-ISTA & $27.52 \pm 1.27$ & $23.96 \pm 0.93$ & $17.19 \pm 2.22$ & $10.48 \pm 1.98$ & $5.79 \pm 0.63$  \\
					RTF-ISTA & $29.01 \pm 1.28$ & $\underline{24.50} \pm \underline{0.95}$ & ${17.34} \pm {2.42}$ & $10.45 \pm 1.96$ & ${5.81} \pm {0.65}$  \\
                     TF-NESTA & $ 28.73 \pm  1.27$ & $ 23.86 \pm 0.97$ & $ 17.06 \pm 2.14$ & $ \underline{10.58} \pm \underline{1.94}$ & $  \underline{5.84}\pm \underline{0.68}$   \\
                     RTF-NESTA & $ 28.85 \pm 1.29 $ & $ 24.35 \pm 0.97 $ & $ 17.32 \pm 2.26$ & $ 10.51 \pm 1.95$ & $ 5.80 \pm 0.67 $   \\
                    TF-SFISTA & $27.18 \pm 1.28$ & $24.27 \pm 0.85$ & $\textbf{17.75} \pm \textbf{2.24}$ & ${10.64} \pm {2.07}$ & $5.71 \pm 0.68$   \\
					RTF-SFISTA & $26.71 \pm 1.29$ & $24.16 \pm 0.84$ & $\underline{17.69} \pm \underline{2.16}$ & ${10.67} \pm {2.06}$ & $5.71 \pm 0.68$  \\
                    TF-Loris & $27.61 \pm 1.29$ & $24.25 \pm 0.88$ & $17.31 \pm 2.20$ & $\textbf{10.59} \pm \textbf{2.01}$ & $5.79 \pm 0.63$  \\
                    RTF-Loris & $29.07 \pm 1.28$ & $\textbf{24.60} \pm \textbf{0.94}$ & ${17.67} \pm {2.38}$ & $\underline{10.58} \pm \underline{1.99}$ & ${5.81} \pm {0.65}$  \\
					\hline
				\end{tabular}}
			\end{sc}
		\end{small}
	\end{center}
% 	\vskip -0.1in
\end{table*}

% % --------------------------------------------------------------------------

\begin{figure}[t]
    \begin{center}
    % \hskip -0.4in
        \begin{tabular}[htb]{P{.4\linewidth}P{.4\linewidth}}
			\cellcolor{white} \includegraphics[width=\linewidth]{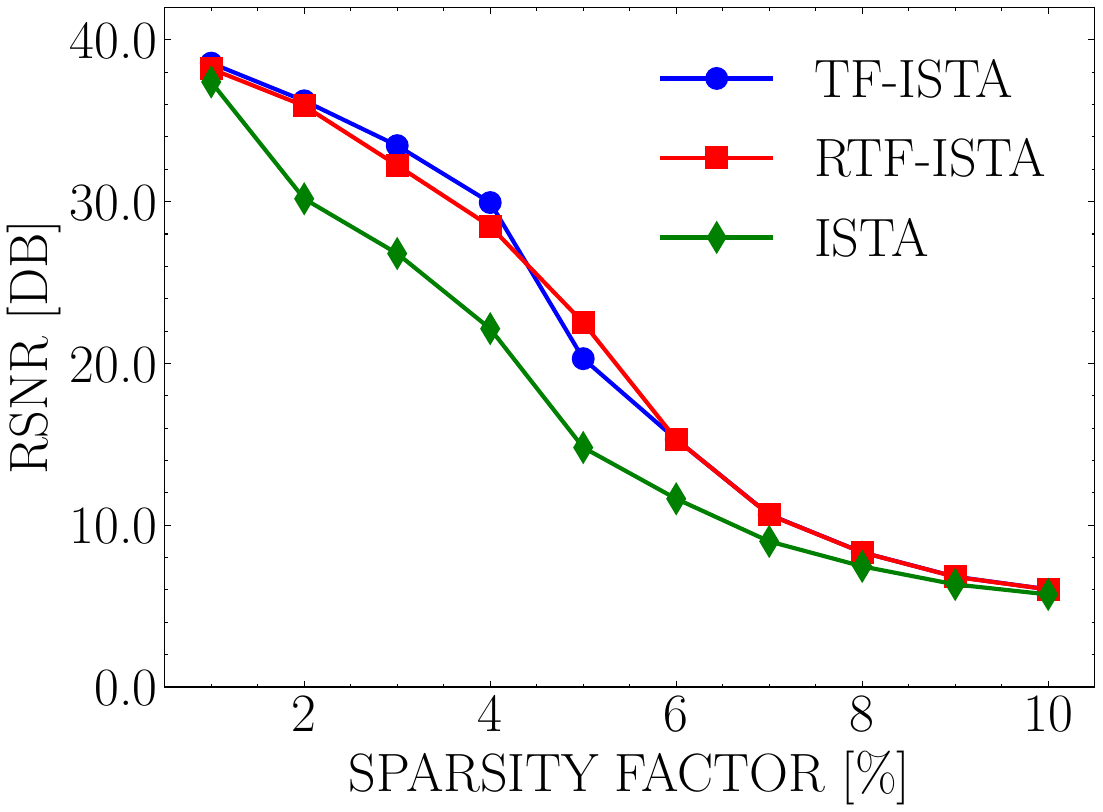}  &
            \cellcolor{white} \includegraphics[width=\linewidth]{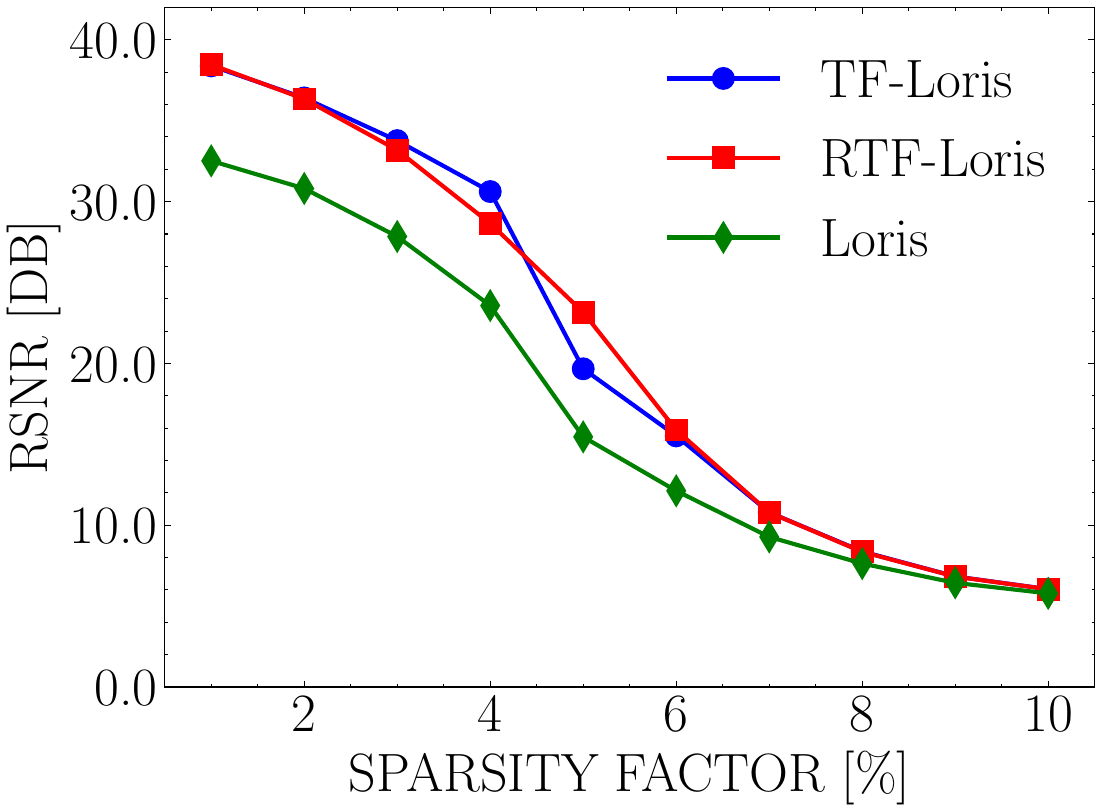}  \\
            \cellcolor{white} \includegraphics[width=\linewidth]{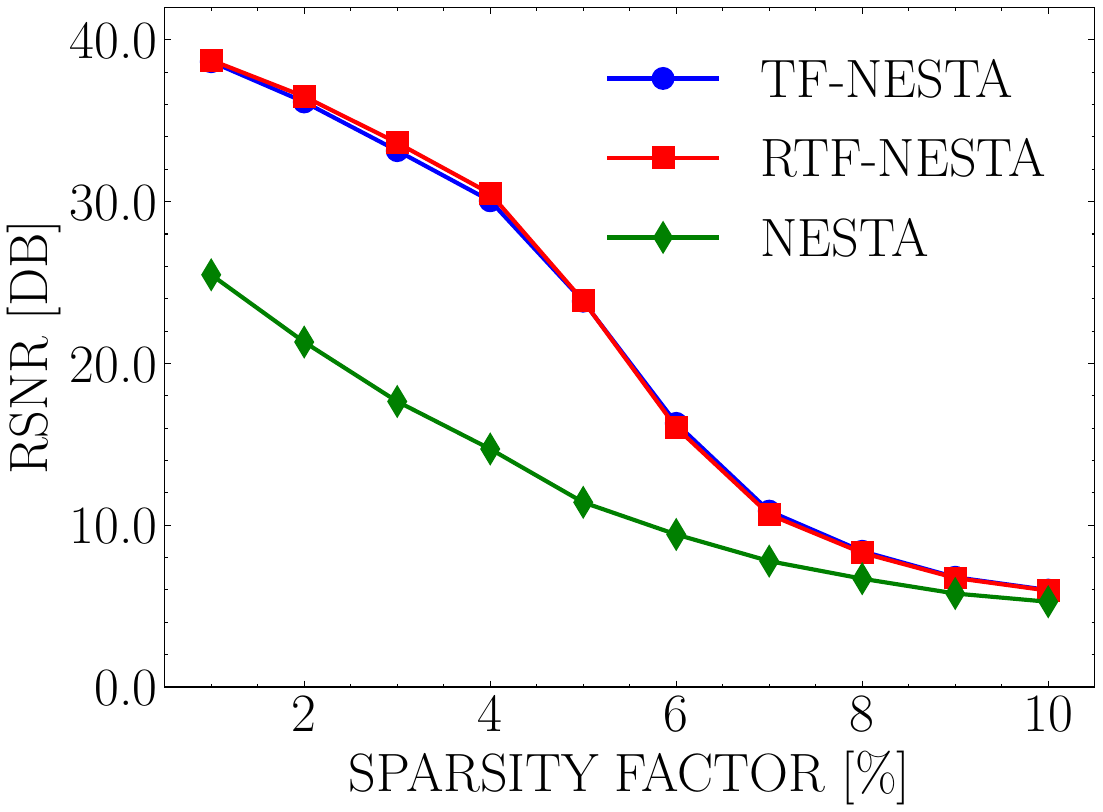}  &
            \cellcolor{white} \includegraphics[width=\linewidth]{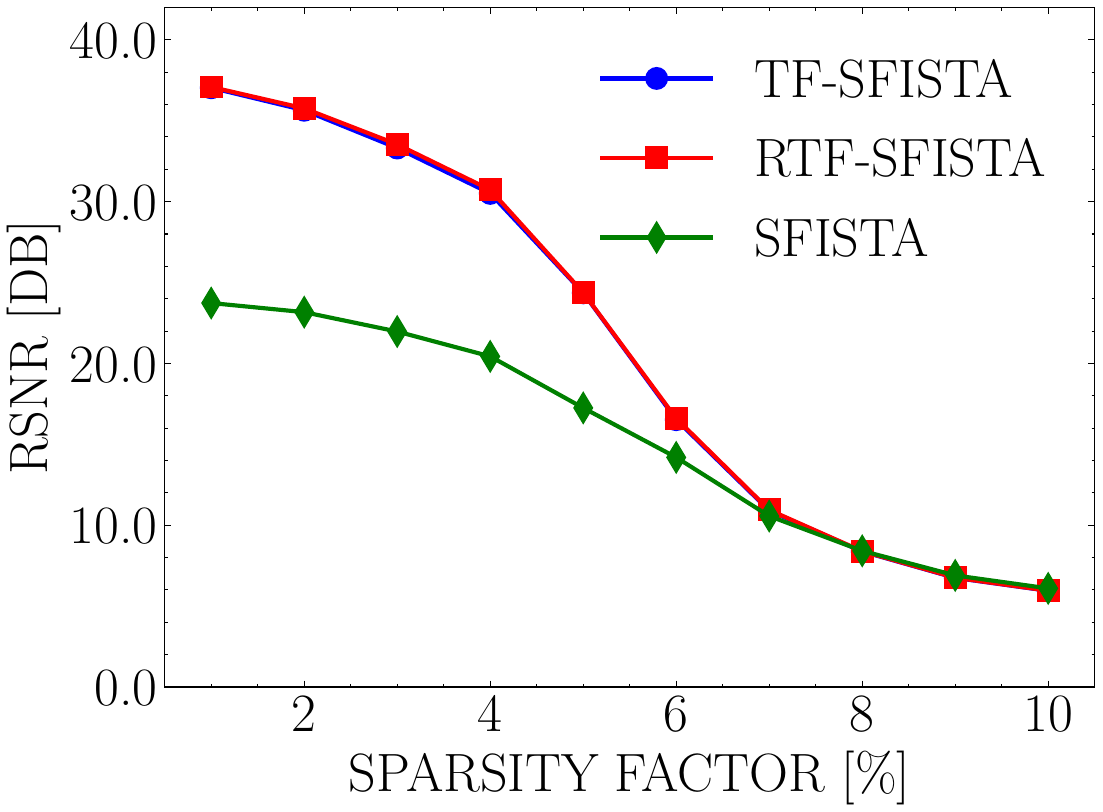}  \\
		\end{tabular} 
	\caption[]{Baseline methods ISTA and Loris in the top row, NESTA and SFISTA in the bottom row, are compared against the TF and RTF variants for SNR = 40 dB for various sparsity level. The sparsity factor (expressed in \%) measures the number of non-zeros in the sparse signal. Lower the sparsity factor, sparser is the signal ($\bld \alpha$). The reconstruction of the analysis-sparse signal $\bld x = \bd D \bld \alpha$ becomes less effective for larger sparsity factors.}
	\label{fig:ISTA_Loris_variants_SNR}
	\end{center}
\end{figure}

% %----------------------------------------------- Table SNR 50 --------------------------------------------------------
\begin{table*}[t]
	\centering
	\caption{RSNR (in dB)  $\pm$ standard deviation for analysis-sparse signal recovery using various competing methods for SNR = 50 dB.  The best performance is shown in \textbf{boldface} and the second best is shown \underline{underlined}.}
	\label{table:SNR_50}
	\vskip 0.05in
	\begin{center}
		\begin{small}
			\begin{sc}
			\resizebox{0.85\linewidth}{!}{	\begin{tabular}{c||ccccc}
			\hline
					\cellcolor{white} & & & \textbf{Sparsity} & &   \\ 
					\hline
					\textbf{Method} & 1$\%$ & 3$\%$ & 5$\%$ & 7$\%$ & 10$\%$  \\
					\hline
                    ISTA  & $41.86 \pm 1.22$ & $27.56 \pm 1.39$ & $15.31 \pm 2.19$ & $9.11 \pm 1.53$ & $5.68 \pm 0.71$\\
					NESTA  & ${35.24} \pm {1.17}$ & ${26.80} \pm {1.16}$ & ${17.25} \pm {2.39}$ & $9.87 \pm 1.74$ & $5.93 \pm 0.79$  \\
                    SFISTA  & $ 23.83 \pm 1.13$ & $22.09 \pm 0.89$ & $17.78 \pm 1.89 $ & $ 10.77 \pm 2.13 $ & $ \textbf{6.16} \pm \textbf{0.87}$  \\
					Loris & $32.61 \pm 1.14$ & $28.62 \pm 1.26$ & $15.84 \pm 3.25$ & $9.41 \pm 1.63$ & $5.83 \pm 0.75$\\
					\hline
					TF-ISTA & $46.20 \pm 1.10$ & $41.69 \pm 2.85$ & $21.50 \pm 7.16$ & $10.89 \pm 2.41$ & $6.06 \pm 0.80$   \\
					RTF-ISTA &${48.62} \pm {1.19}$ & $41.66 \pm 3.36$ & ${24.60} \pm {4.61}$ & $10.89 \pm 2.42$ & $6.09 \pm 0.84$  \\
					TF-NESTA & $ 48.90 \pm 1.16 $ & $ \underline{43.21} \pm \underline{1.01}$ & $\textbf{30.72}  \pm \textbf{7.09}$ & $ 11.02 \pm 2.69$ & $ 6.00 \pm 0.85$   \\
                    RTF-NESTA & $ \underline{49.08} \pm \underline{1.17} $ & $ \textbf{43.67} \pm \textbf{1.02} $ & $ \underline{30.23} \pm \underline{7.44} $ & $ 10.81 \pm 2.63 $ & $ 6.00 \pm 0.84 $   \\
                    TF-SFISTA & $39.72 \pm 1.12$ & $37.64 \pm 0.96$ & $28.31 \pm 5.42$ & $\underline{11.21} \pm \underline{2.78}$ & $6.00 \pm 0.87$  \\
					RTF-SFISTA & $41.23 \pm 1.11$ & $39.07 \pm 0.97$ & $29.05 \pm 5.76$ & ${\textbf{11.30}} \pm {\textbf{2.81}}$ & $6.00 \pm 0.87$  \\
                    TF-Loris & $46.58 \pm 1.13$ & ${42.49} \pm {2.81}$ & $22.75 \pm 7.71$ & ${11.02} \pm {2.47}$ & $6.06 \pm 0.81$  \\
                    RTF-Loris  & ${\textbf{{49.41}}} \pm {{\textbf{1.22}}}$ & ${42.58} \pm {3.23}$ & ${25.19} \pm {4.44}$ & ${11.04} \pm {2.48}$ & $\underline{6.10}$ $ \pm$ $\underline{0.84}$  \\
					\hline
				\end{tabular}}
			\end{sc}
		\end{small}
	\end{center}
% 	\vskip -0.1in
\end{table*}
% %------------------------------------- Different sensing matrices -----------------------------------------------
\begin{table*}[t]
	\centering
	\caption{RSNR (in dB)  $\pm$ standard deviation for analysis-sparse signal recovery using various competing methods and sensing matrices for SNR = 50 dB and sparsity factor = 1\%.  The best performance is shown in \textbf{boldface} and the second best is shown \underline{underlined}.}
	\label{table:sensing}
	\vskip 0.05in
	\begin{center}
		\begin{small}
			\begin{sc}
			\resizebox{0.8\linewidth}{!}{	\begin{tabular}{c||cccc}
			\hline
					\cellcolor{white} & & \multicolumn{2}{c}{\textbf{Sensing Matrix}}  &   \\ 
					\hline
					\textbf{Method} & Gaussian & Bernoulli & uniform & Laplacian  \\
					\hline
                    ISTA  & $41.86 \pm 1.22$ & $ 43.46\pm1.00 $ & $45.06 \pm1.11 $ & $40.97 \pm 2.56$ \\
					NESTA  &  ${35.24} \pm {1.17}$& $35.07 \pm1.08 $ & $ 35.29\pm 1.15$ & $35.27 \pm 1.11$ \\
                    SFISTA  & $ 23.83 \pm 1.13$ & $ 23.81\pm 1.04$ & $ 23.95\pm1.12 $ & $ 23.94\pm 1.10$ \\
					Loris & $32.61 \pm 1.14$ & $ 32.91\pm 1.08$ & $ 33.31\pm 1.18$ & $32.61 \pm1.06 $ \\
					\hline
					TF-ISTA& $46.20 \pm 1.10$ & $46.23 \pm1.05 $ & $46.41\pm1.14 $ & $46.33 \pm1.11 $ \\
					RTF-ISTA & ${48.62} \pm {1.19}$ & $\underline {49.40} \pm \underline {1.19} $ & $48.85 \pm1.21 $ & $ {49.08} \pm  {1.20} $ \\
					TF-NESTA& $ 48.90 \pm 1.16 $ & $48.92 \pm 1.15$ & $ 49.00\pm 1.24$ & $ 49.01\pm1.21 $ \\
                    RTF-NESTA & $ \underline{49.08} \pm \underline{1.17} $ & $49.07 \pm1.16 $ & $\underline {49.15} \pm \underline {1.27} $ & $ \underline {49.17} \pm  \underline {1.22}$ \\
                    TF-SFISTA & $39.72 \pm 1.12$ & $ 39.75\pm 1.06$ & $39.88 \pm 1.11$ & $39.90 \pm1.12 $ \\
					RTF-SFISTA & $41.23 \pm 1.11$ & $42.30 \pm1.07 $ & $42.54\pm1.19 $ & $42.40 \pm 1.07$ \\
                    TF-Loris & $46.58 \pm 1.13$& $46.60 \pm 1.09$ & $46.79 \pm 1.17$ & $46.73 \pm 1.14$ \\
                    RTF-Loris  & ${\textbf{{49.41}}} \pm {{\textbf{1.22}}}$ & $ {\bf50.11}\pm {\bf1.25}$ & ${\bf49.63} \pm {\bf1.31} $ & $ {\bf49.85} \pm {\bf1.27} $ \\
					\hline
				\end{tabular}}
			\end{sc}
		\end{small}
	\end{center}
% 	\vskip -0.1in
\end{table*}
% %------------------------------------------------------------
% %------------------------------------------------------------
\paragraph{Analysis-sparse recovery experiments}
\textcolor{black}{One hundred realizations $\{ \bld x_i^* = \bd D \bld \alpha_i^*\}_{i=1}^{100}$ are generated and the corresponding reconstructions $\{ \hat{\bld x}_i^*\}_{i=1}^{1000}$ are obtained. {\color{black}The performance measure used is reconstruction signal-to-noise ratio (RSNR), which is defined as
\begin{equation}
	\text{RSNR}(\hat{\boldsymbol{x}}, \boldsymbol{x}^*) = 20 \cdot \log_{10} \left( \frac{\| \boldsymbol{x}^* \|_2}{\| \hat{\boldsymbol{x}} - \boldsymbol{x}^* \|_2} \right) \text{ dB},
\end{equation}
where $\boldsymbol{x}^*$ is the ground-truth vector and $\hat{\boldsymbol{x}}$ is its estimate. For a fair comparison, we use the same stopping criterion for all the algorithms: either $\norm{\bld{x}^{(k)}-\bld{x}^{(k-1)}} < 10^{-4}$, where $k$ denotes the iteration index, or when a preset number of iterations is reached, whichever happens earlier.\\
\indent Fig.~\ref{fig:ISTA_Loris_variants_SNR} shows the variation of RSNR with varying levels of sparsity from $1\%$ to $10\%$ for SNR = 40 dB, for TF and RTF variants of benchmark methods. The performance of ISTA decays rapidly as the signals become less sparse. TF-ISTA and RTF-ISTA perform better than ISTA for different levels of sparsity. A similar trend can be observed in the case of Loris, NESTA, and SFISTA and their TF/RTF variants. \\
\indent Tables \ref{table:SNR_30} and \ref{table:SNR_50} show the results of various algorithms for SNR = 30 dB and SNR = 50 dB, respectively. In the noisy case with SNR = 30 dB, Loris performs the best when the sparsity is 1\%. RTF-Loris, TF-SFISTA, and TF-Loris outperform the competing methods for other sparsity factor levels, while SFISTA leads by a small margin of 0.07 dB RSNR for 10\% sparsity. In the less noisy scenario of SNR = 50 dB, TF and RTF variants have a clear advantage over the benchmark methods across all sparsity levels, except when the sparsity factor is 10\%, where SFISTA has a marginally superior performance over RTF-Loris.}\\
\indent Table~\ref{table:sensing} portrays the difference in performance when sensing matrices from different probability distributions are used. The entries of the sensing matrices were drawn from Gaussian, Bernoulli, uniform, and Laplacian distributions. The TF and RTF variants continue to outperform the benchmark methods for various choices of the distribution of the sensing matrix.
}
% ----------------------------------------
% ----------------------------------------
% 	      COMPRESSED IMAGE RECOVERY
% ----------------------------------------
% ----------------------------------------
\section{Compressed Sensing Image Recovery (CS-IR)}
\label{sec:image_recovery}
We now apply the techniques developed thus far for solving the important problem of compressed sensing image recovery (CS-IR). In this imaging modality, an image $\bld x\in\bb R^n$ (vectorized representation) is observed using a sensing matrix $\bd A\in\bb R^{m\times n}$ to obtain compressed measurements $\bld y = \bd A\bld x$. The image is assumed to be sparse under a suitable \emph{sparsifying transform} $\cl F$, such as the DCT \cite{ravi2012}, wavelet transform \cite{wavelet-CSIR}, or a learnt dictionary \cite{KSVD}. Recent learning-based CS-IR methods employ deep neural networks as the analysis operator, for example, ISTA-Net \cite{ISTA-Net}, ADMM-CSNet \cite{yang2018admm},  SCSNet \cite{shi2019scalable}, NL-CSNet \cite{cui2021image}, and NLR-CSNet \cite{sun2020learning}. On the other hand, models that do not explicitly enforce sparsity such as ReconNet \cite{reconnet}, and models that rely on generative image priors such as compressed sensing generative model (CSGM) \cite{CSGM}, deep compressed sensing (DCS) \cite{DCS} have also been successfully developed for solving the CS-IR problem. Learning-based methods have been shown to result in a superior reconstruction performance than their optimization counterparts. We unfold the iterations of TF-ISTA to obtain TF-ISTA-Net, much along the lines of ISTA-Net as shown in Figure~\ref{fig:schematic_model}. In this framework, the analysis operator is a learnable deep sparsifying transform $\cl F$ and the corresponding analysis-sparse recovery problem is posed as:
\begin{equation}
   \underset{\bld x\in \bb R^n}{\text{minimize }} f_{\bd B}(\bld x) + \lambda \norm{\cl F (\bld x)}_1,
\label{eq:form-ISTA-Net}
\end{equation}
where $\lambda >0$ is the regularization parameter. As in the case of ISTA-Net, we obtain the following update steps for solving \eqref{eq:form-ISTA-Net}, where the sparsifying transform $\cl F$ is allowed to vary from one layer to the next.
% --------------------------------------------------------------------------
\begin{figure}[t]
	\centering
	\includegraphics[width=6in]{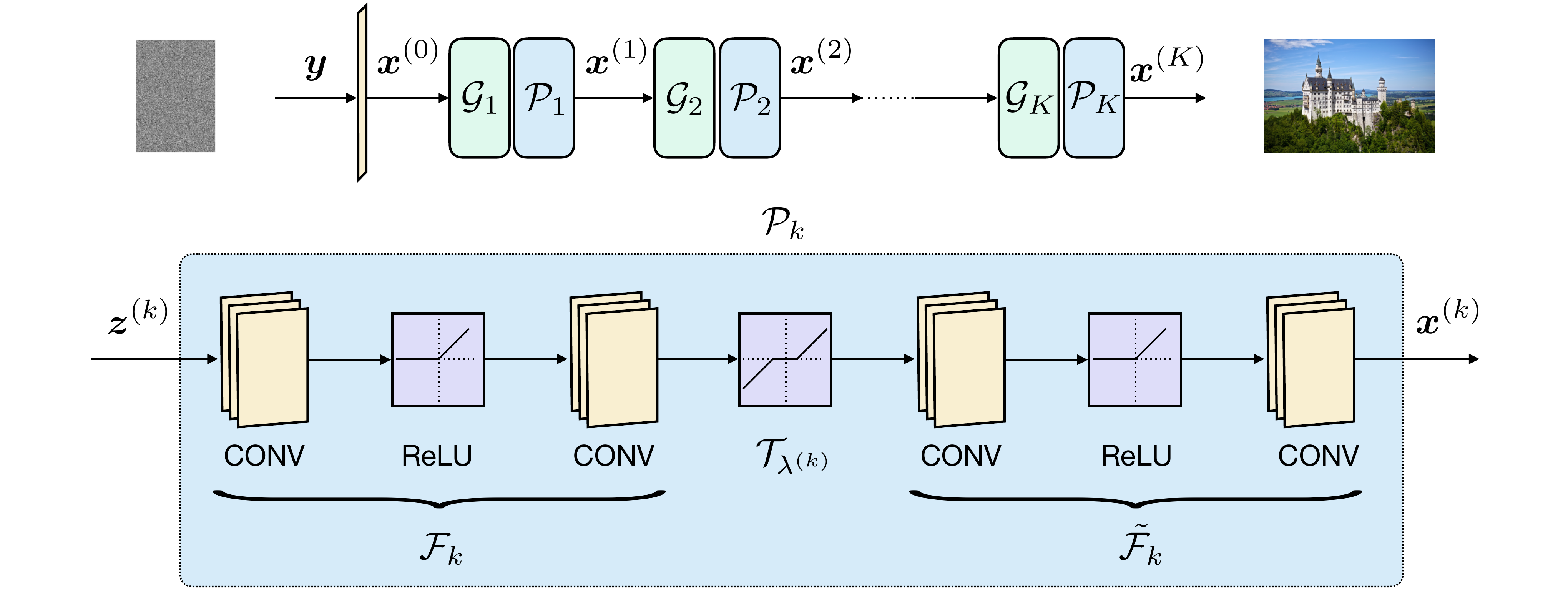}
	\caption{Architecture of the deep neural network (DNN) for compressed sensing image recovery obtained by unrolling TF-ISTA with a learnable sparsifying transform $\{\cl F_k\}$, where $k$ denotes the layer index. The DNN can be trained in an end-to-end fashion.}
	\label{fig:schematic_model}
\end{figure}
% --------------------------------------------------------------------------
Consequently, we have $\cl F_k$ instead of $\cl F$, where $k$ denotes the layer index, in the following updates/layer-wise transformations: 
\begin{align}
	\textbf{Gradient step ($\cl G$): }& \bld z^{(k)} = \cl G\left(\bld x^{(k-1)}\right) \label{eq:istanet_gradient} 
	= \bld x^{(k-1)} - \eta^{(k)}\nabla f_{\bd B}(\bld x^{(k-1)}),  \\ 
	\textbf{Proximal step ($\cl P_k$): }& \bld x^{(k)} = \cl P_{k}\left(\bld z^{(k)}\right) \label{eq:istanet_iterations}
	=  \tilde{\cl {F}}_{k}\left(\cl T_{\lambda^{(k)}}\left(\mathcal{F}_{k}\left({\bld z}^{(k)}\right)\right)\right), 
\end{align}
{\color{black} where $\eta^{(k)}>0$ is the learnt step-size parameter in the $k$\textsuperscript{th} layer}. Effectively, $\tilde{\cl F}_{k}$ is trained to invert $\cl F_{k}$. The learnable quantities are $\left\{\eta^{(k)}, \lambda^{(k)}, \cl F_k, \tilde{\cl F}_k\right\}_{k=1}^K$, where $K$ is the total number of layers. The network is trained end-to-end on image patches with the loss function:
\begin{equation} \begin{split}
	\cl L\left(\bld x^{(K)}, \bld x^*\right) = \norm{\bld x^{(K)} - \bld x^{*}}_2^2 + 
	\sum_{k=1}^K \norm{\bld x^{*} - \tilde{\cl F}_k\left(\cl F_k \left(\bld x^{*} \right)\right)}_2^2,  \end{split}
\label{eq:network_function}
\end{equation}
where the first term in the loss is aimed at reducing the $\ell_2$ loss between the estimate obtained at the output of the $K$\textsuperscript{th} layer and the ground-truth, thereby bringing the estimate closer to the ground-truth, while the second term enforces $\tilde{\cl F}_{k}$ to act as the inverse of the operator $\cl F_k$.\\
\indent In a similar vein, one could also unfold the TF and RTF variants of ISTA and Fast-ISTA (FISTA), namely, TF-FISTA, RTF-ISTA, RTF-FISTA to obtain the corresponding ``Net'' counterparts. 
Overall, we have four network counterparts of the algorithms developed in this paper.\\
\indent {\color{black} The back-projection loss was used for solving the inverse problem in computerized tomography (CT) and the resulting algorithm was unfolded to obtain the corresponding network counterpart \cite{genzel2022near}. The authors use a post-processing U-Net \cite{ronneberger2015convolutional} for iterative reconstruction. On the contrary, we use a sparsifying transform consisting of two linear convolutional blocks separated by ReLU activations inspired by ISTA-Net \cite{ISTA-Net}. Genzel {\it et al.} use a pre-trained U-Net and a computational backbone, but in our case, we train the sparsifying transform from scratch. The applications considered are also different -- while we consider CS-IR with random sensing matrices, Genzel {\it et al.} consider the reconstruction problem in CT imaging.}
% -------------------------------------------------------------------
\paragraph{Experiments on Natural Images}
We perform CS-IR on natural images taken from Set11 \cite{reconnet}, DIV2K \cite{Agustsson_2017_CVPR_Workshops}, Urban100 \cite{huang2015single} and BSD68 \cite{BSD68} datasets. The model shown in Figure~\ref{fig:schematic_model}, with $K=8$ unfoldings, is trained using $88,912$ image patches of size $33 \times 33$ extracted from $91$ images, similar to that reported in \cite{reconnet}. The sensing matrices are in $\bb R^{m \times 1089}$ and their entries are drawn from a Gaussian distribution with zero mean and variance $m^{-1}$ and the columns are normalized to have unit $\ell_2$-norm. The compressed sensing (CS) ratio is $m/{(33 \times 33)} = m/1089.$ The training was performed for $100$ epochs using Adam \cite{kingma2014adam} optimizer with a constant learning rate. The computing platform deployed is a Dell 5820 workstation equipped with Intel Xeon processor and a single Nvidia RTX 2080 Ti GPU. The training took approximately three hours.
\begin{figure}[]
	\begin{center}
	\begin{tabular}[htb]{P{.23\linewidth}P{.22\linewidth}P{.22\linewidth}P{.22\linewidth}}
		\cellcolor{white} {\includegraphics[width=1.0\linewidth]{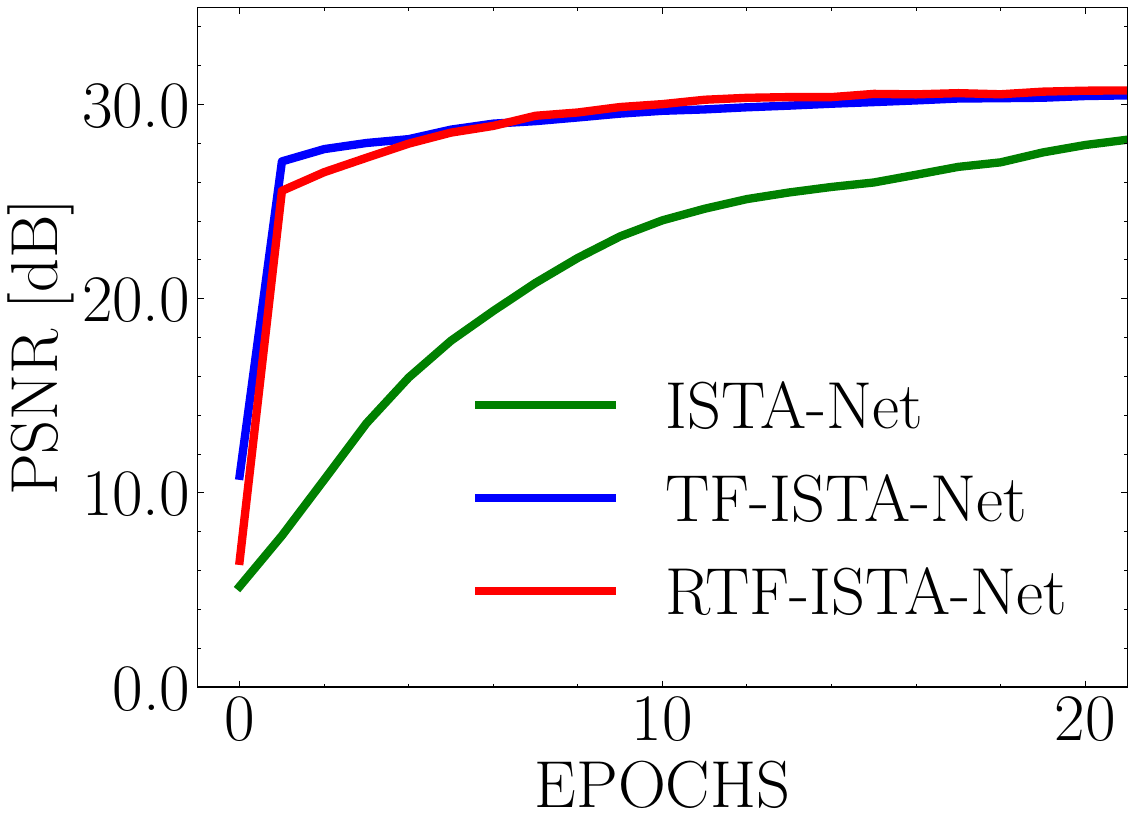}} &
		\cellcolor{white} {\includegraphics[width=1.0\linewidth]{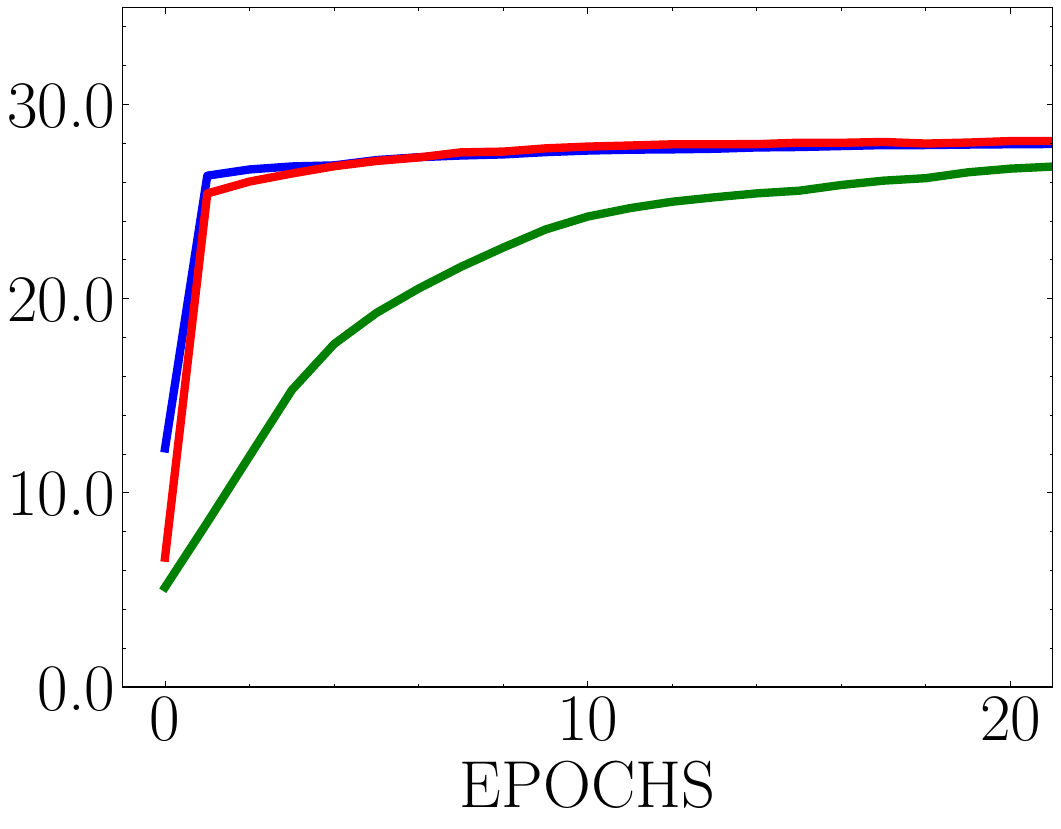}} &
		\cellcolor{white} {\includegraphics[width=1.0\linewidth]{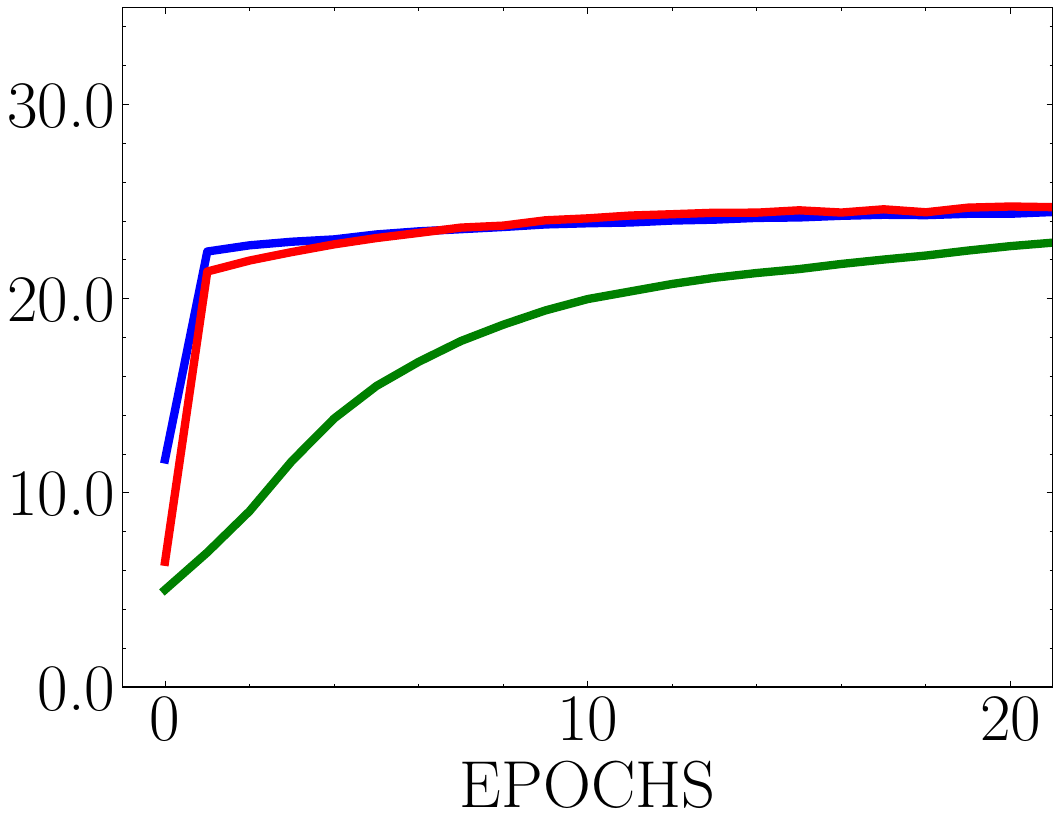}} &
		\cellcolor{white} {\includegraphics[width=1.0\linewidth]{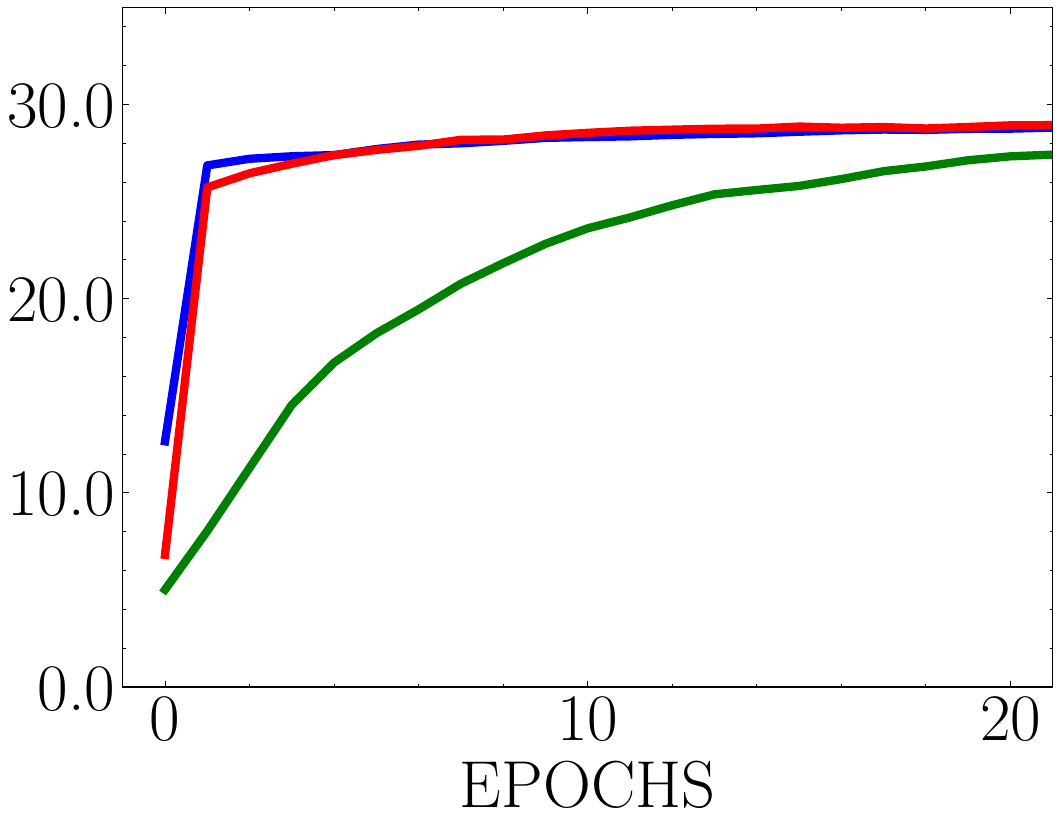}} \\
		\cellcolor{white} {\includegraphics[width=1.0\linewidth]{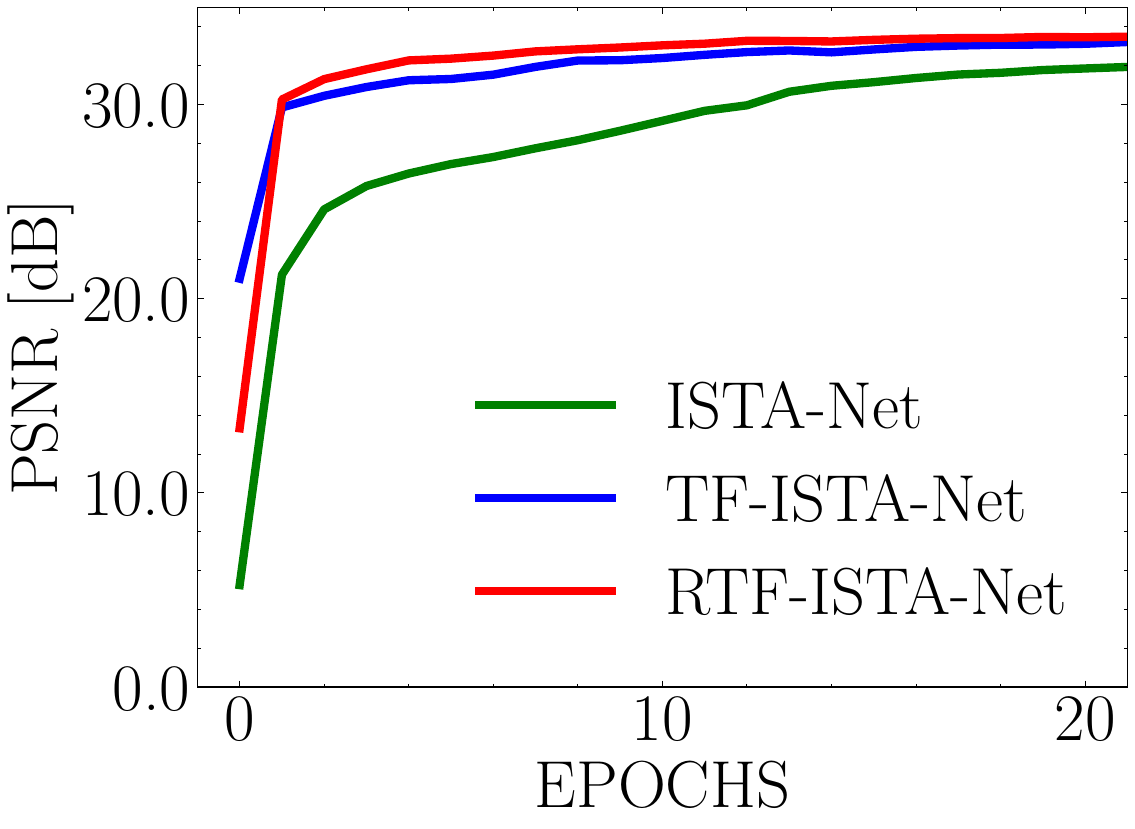}} &
		\cellcolor{white} {\includegraphics[width=1.0\linewidth]{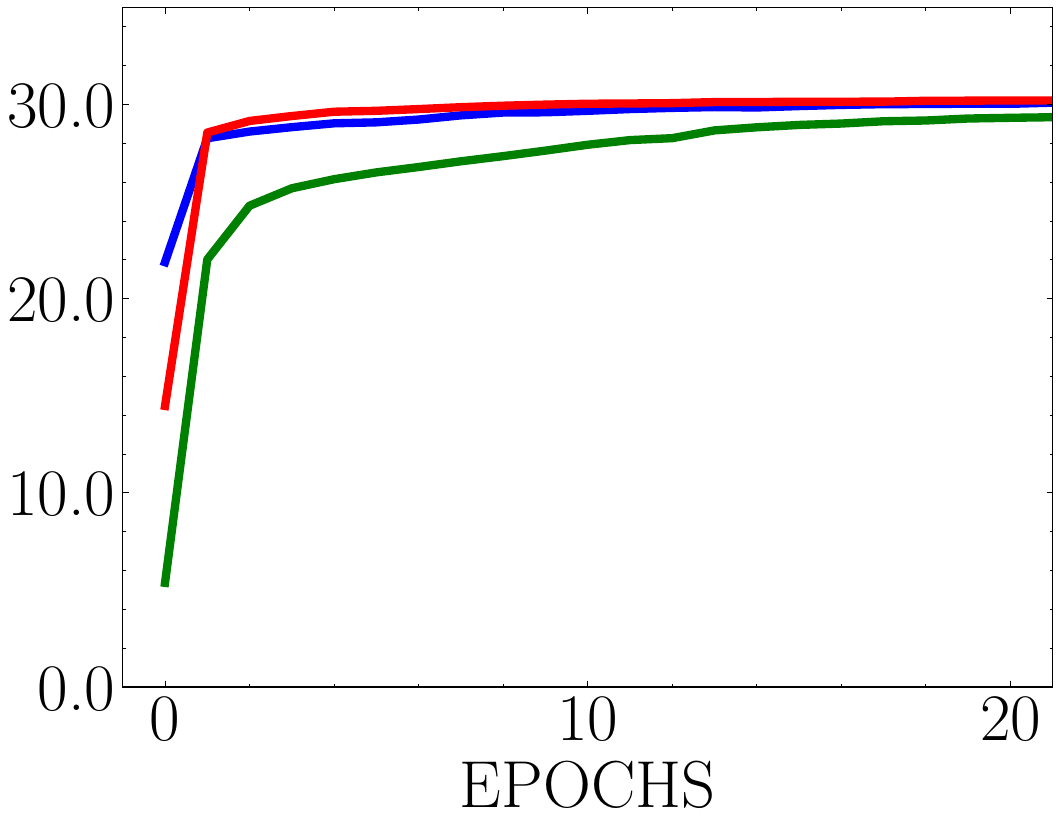}} &
		\cellcolor{white} {\includegraphics[width=1.0\linewidth]{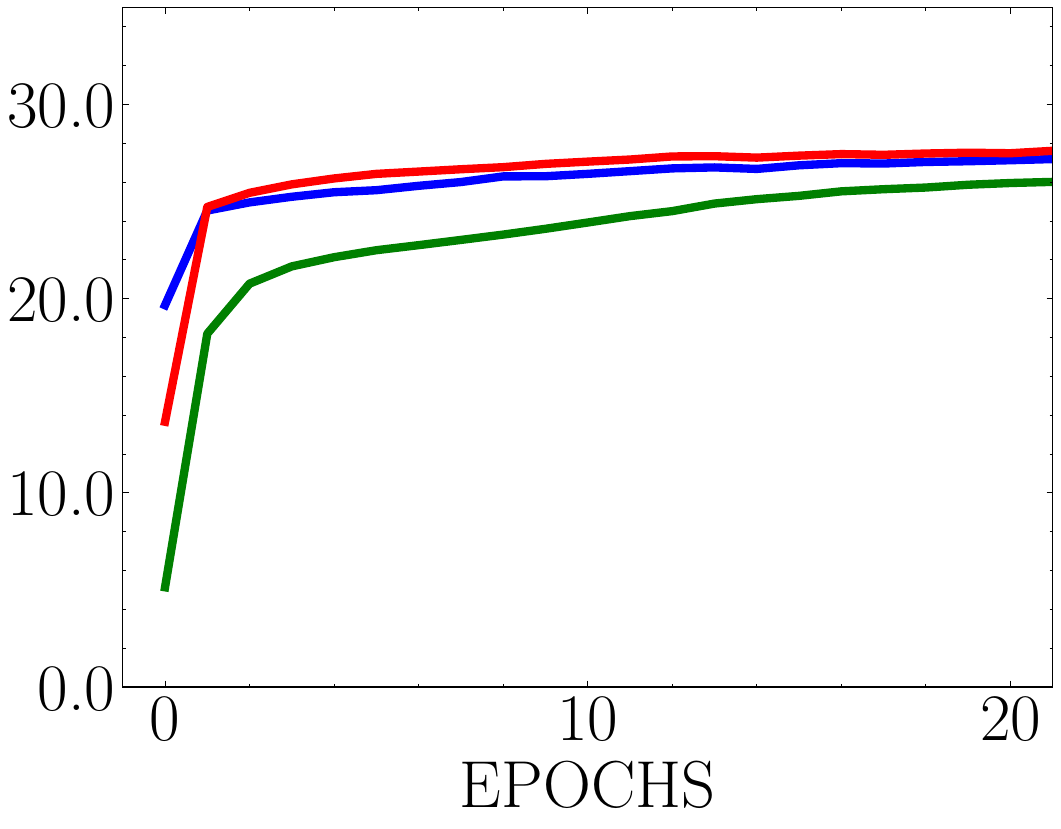}} &
		\cellcolor{white} {\includegraphics[width=1.0\linewidth]{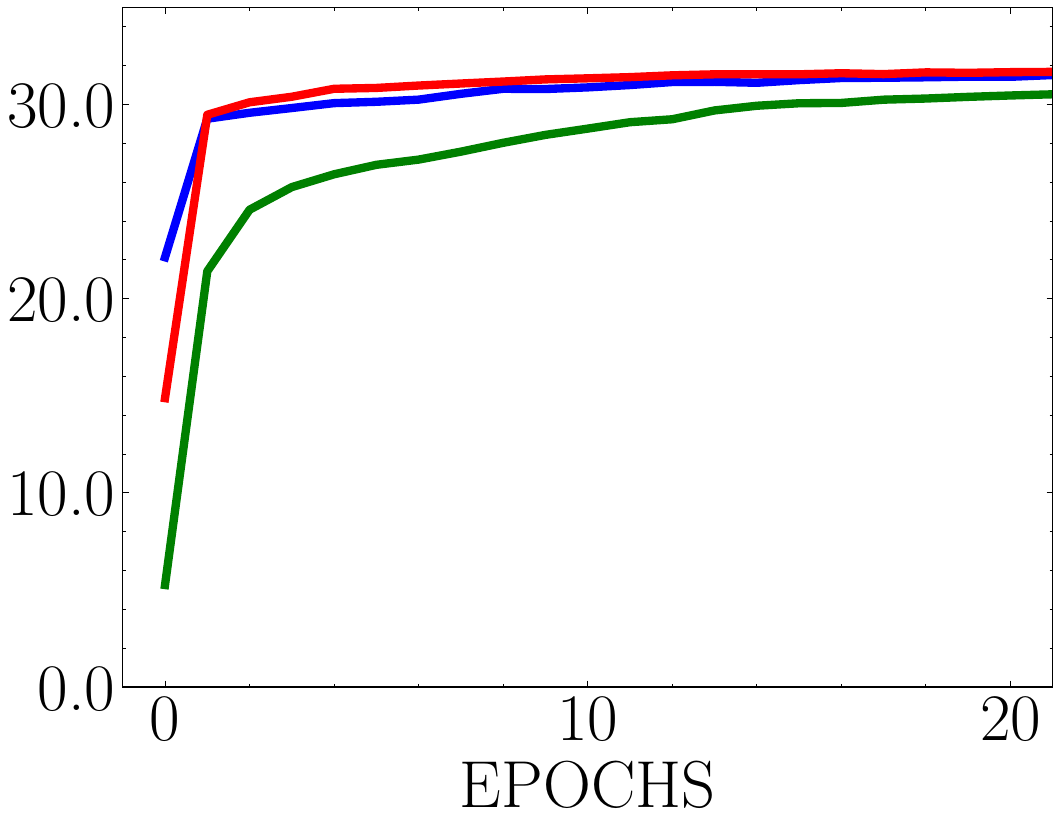}} \\
		$\quad \ \, $ Set11 & $\quad $ BSD68 & $\quad $ Urban100 & $\quad  $ DIV2K \\
	\end{tabular}
	\caption{Training convergence behavior for 20\% (top row) and 30\% (bottom row) CS ratio on various test datasets with ISTA-Net, TF-ISTA-Net, and RTF-ISTA-Net models. The figures show that the best convergence performance is achieved by the TF variants.}
	\label{fig:CSIR_convergence_ISTA} 
	\end{center} 
\end{figure}
% --------------------------------------------------------------------------
\begin{figure}[htb]
    \begin{center}
    % \hskip -0.4in
	\begin{scriptsize}
	\begin{tabular}[b]{P{.20\linewidth}P{.20\linewidth}P{.20\linewidth}P{.20\linewidth}}
	% \captionsetup[subfigure]{labelformat=empty}
	\cellcolor{white} \includegraphics[width=\linewidth]{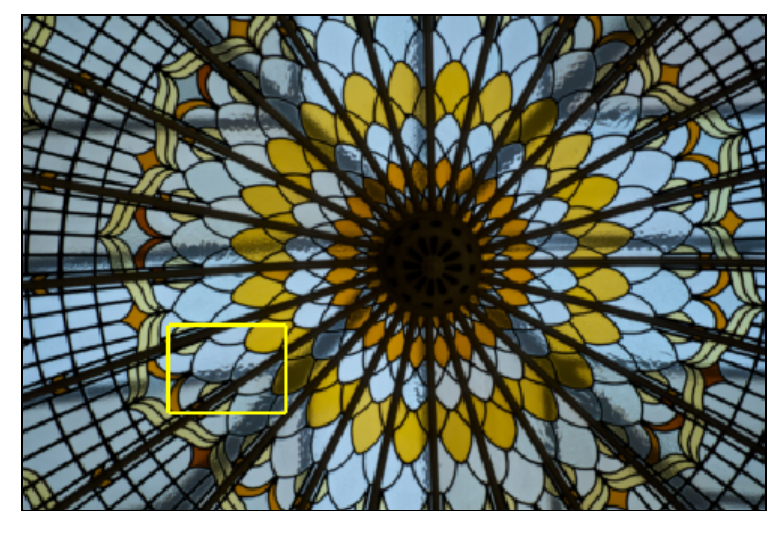} &
	\includegraphics[width=\linewidth]{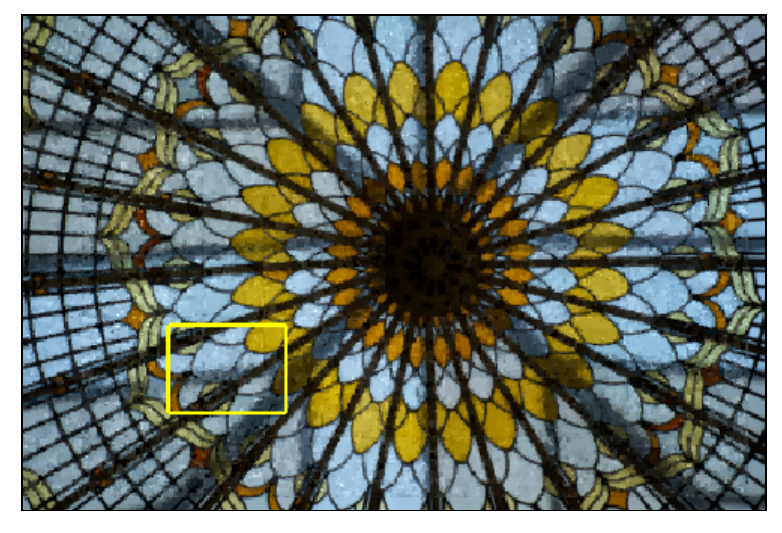}&
	\includegraphics[width=\linewidth]{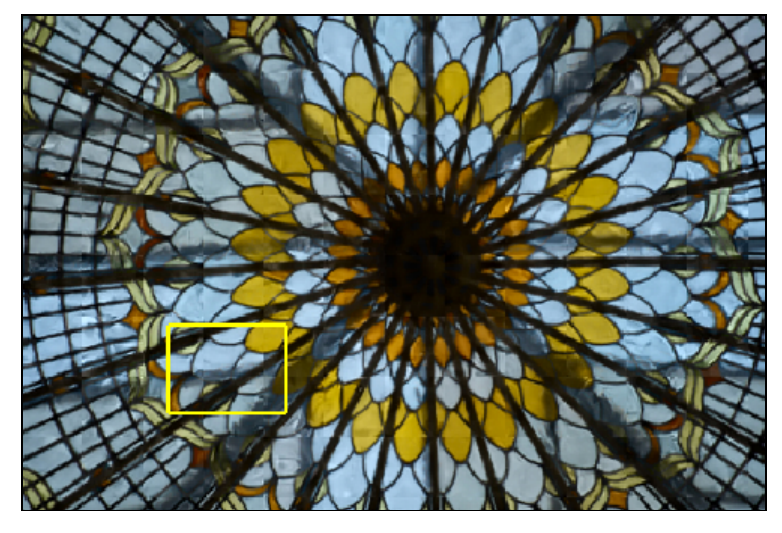}&
	\includegraphics[width=\linewidth]{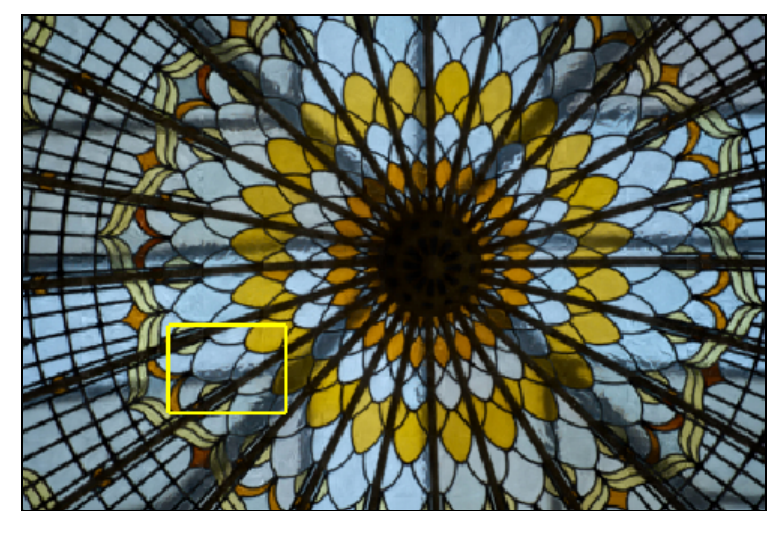} \\[-1pt]
	\includegraphics[width=\linewidth]{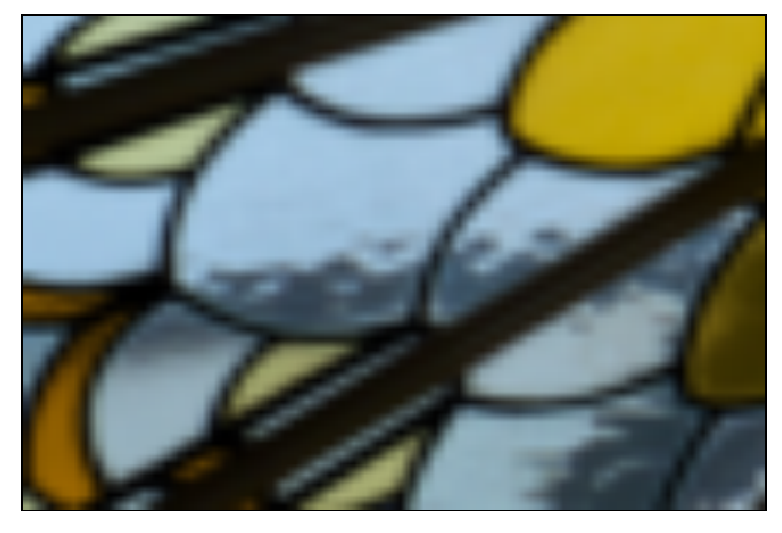}&
	\includegraphics[width=\linewidth]{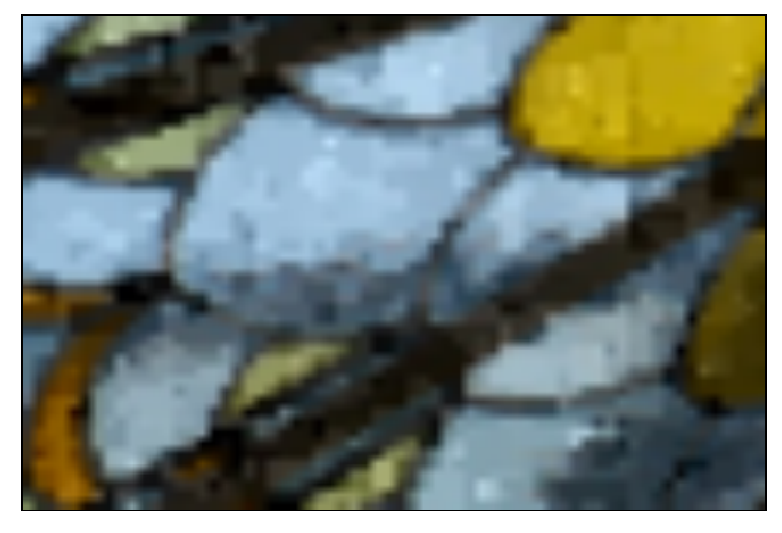}&
	\includegraphics[width=\linewidth]{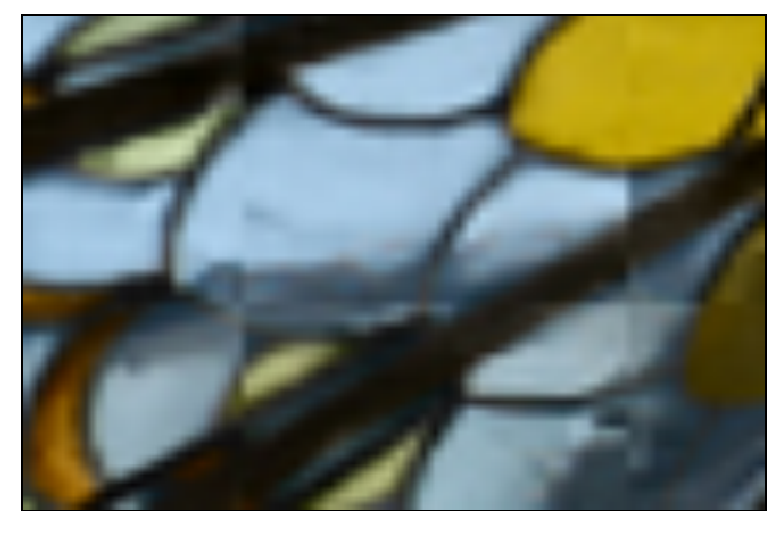}&
	\includegraphics[width=\linewidth]{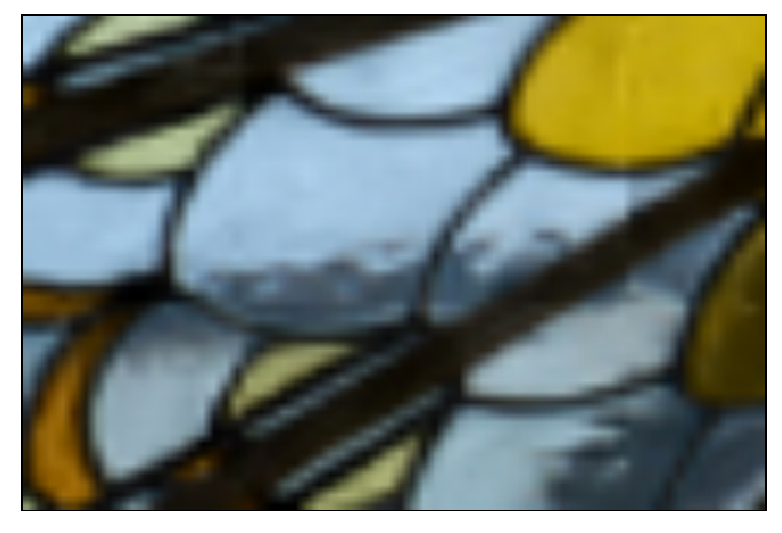} \\[-1pt]
	Original & TVAL3 & ReconNet & ISTA-Net \\
	& 22.76/0.8569 & 23.19/0.8674 & 29.71/0.9643 \\
	\includegraphics[width=\linewidth]{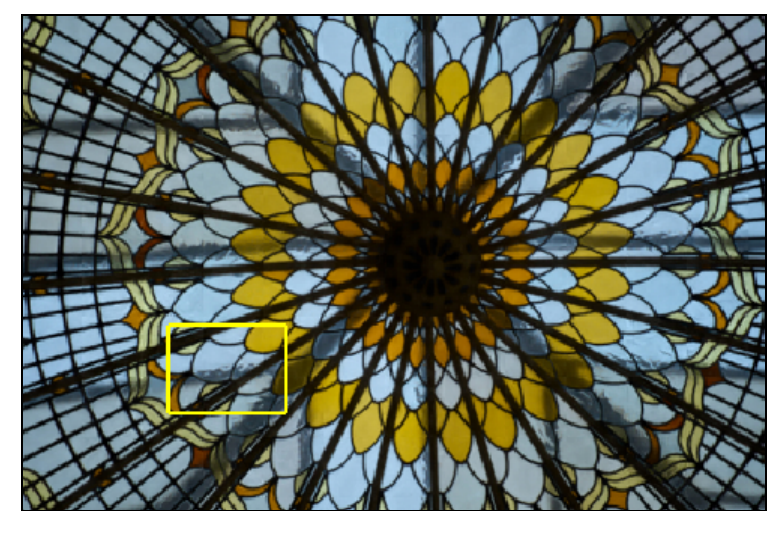}&
	\includegraphics[width=\linewidth]{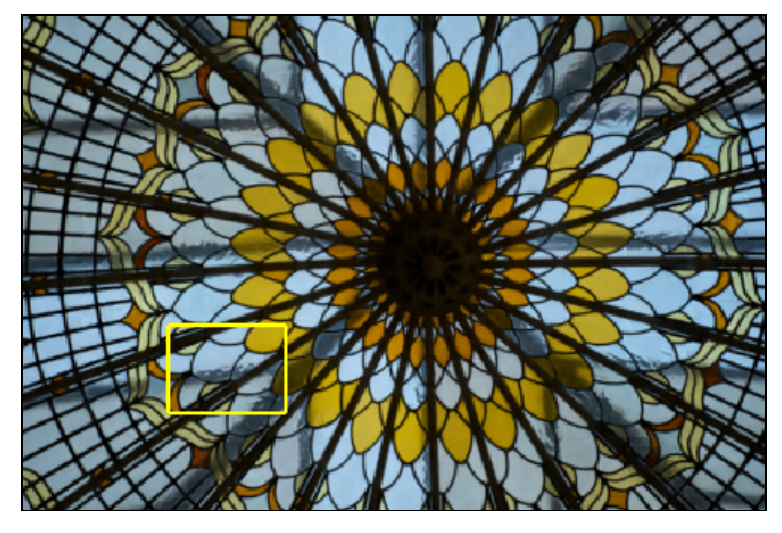}&
	\includegraphics[width=\linewidth]{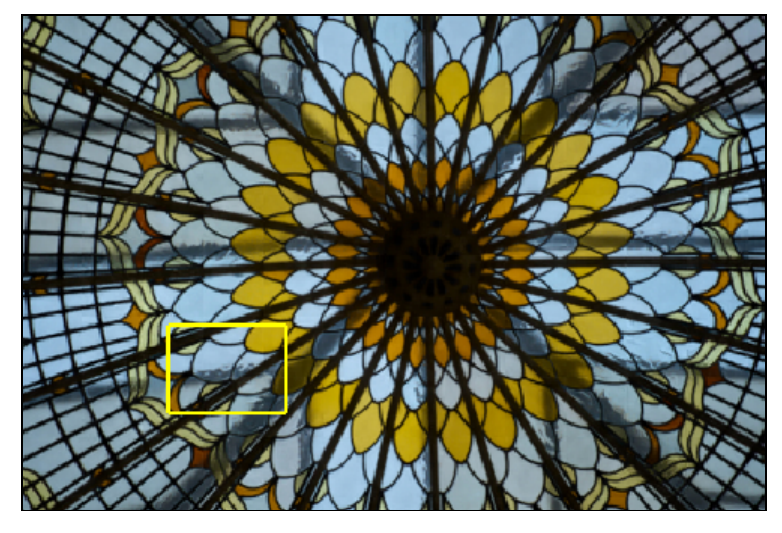}&
	\includegraphics[width=\linewidth]{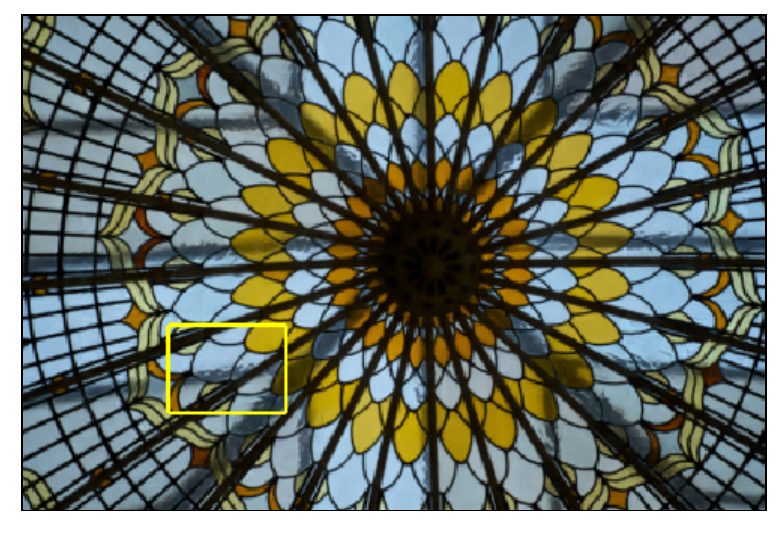} \\[-1pt]
	\includegraphics[width=\linewidth]{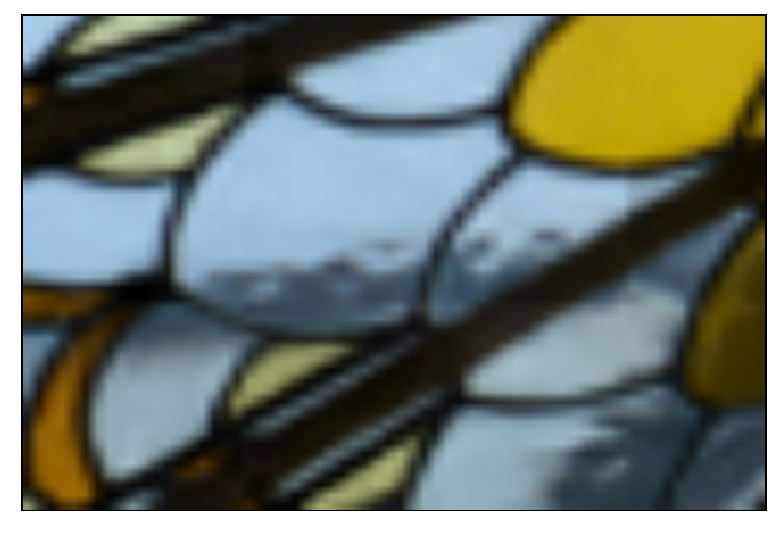}&
	\includegraphics[width=\linewidth]{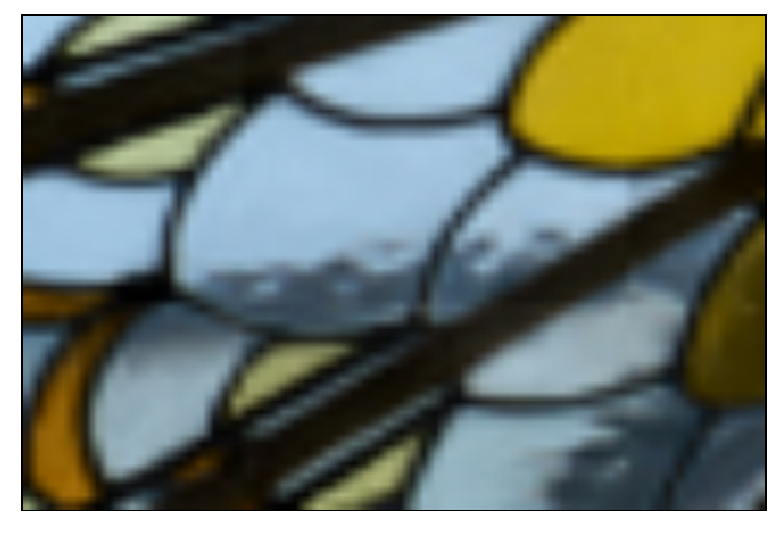}&
	\includegraphics[width=\linewidth]{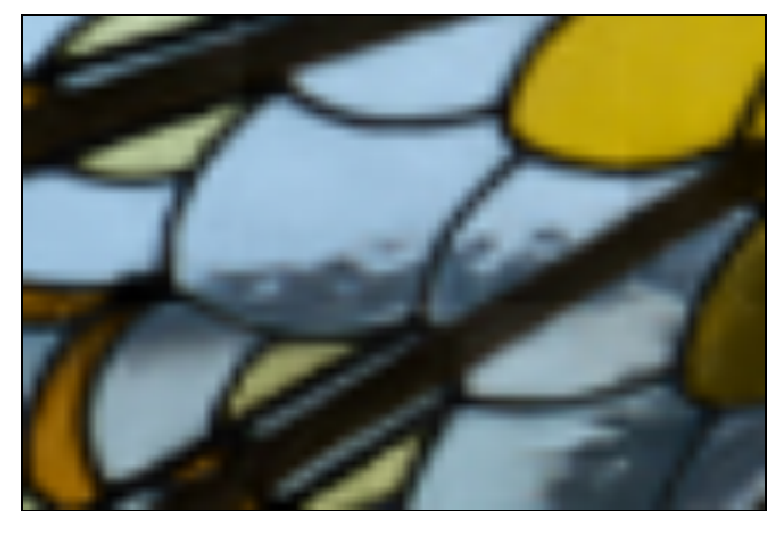}&
	\includegraphics[width=\linewidth]{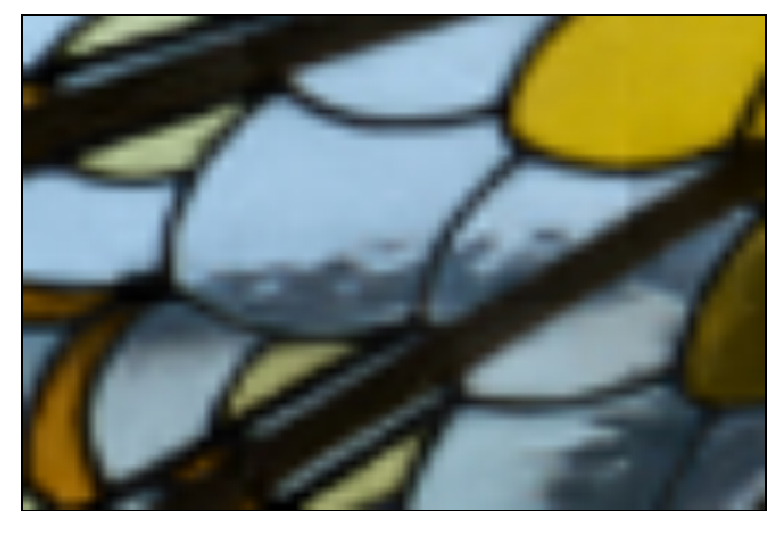} \\[-1pt]
	TF-ISTA-Net & RTF-ISTA-Net & \bf{TF-FISTA-Net}& RTF-FISTA-Net \\
	32.04/0.9783 & 32.22/0.9791 & \bf{32.54/0.9806}& 31.84/0.9774 \\
	\end{tabular} 
	\end{scriptsize}
	\caption{Reconstructions achieved using various methods under comparison. The numbers indicate PSNR/SSIM. Unlike the benchmark methods (first row), the proposed methods (second row) do not have strong artifacts. Spatial continuity is the best with the proposed methods, whereas the benchmark methods have distortions (cf. the distorted blue background in the zoomed-in area). TF-FISTA-Net has the best performance.}
	\label{fig:example_recon}   
\end{center}
\end{figure}
%----------------------------------------------- Table Set11 --------------------------------------------------------
\begin{table*}
	\centering
	\caption{CS-IR on \textit{Set11}: PSNR (in dB) and SSIM for compressed sensing image recovery using various competing methods. TF-FISTA-Net has a superior performance among all techniques under comparison for all CS ratios. We observe the following order of performance: TF-FISTA-Net $>$ TF-ISTA-Net $>$ RTF-FISTA-Net $\approx$ RTF-ISTA-Net $>$ ISTA-Net $>$ TVAL3 $>$ ReconNet. The symbol $>$ must be read as ``better than.'' The best performance is shown in \textbf{boldface} and the second best is shown \underline{underlined}.}
	\label{table:Set11}
	\vskip 0.05in
	\begin{center}
		\begin{small}
			\begin{sc}
			\resizebox{0.95\linewidth}{!}{	\begin{tabular}{c||ccccc}
			\hline
					%\toprule
					\cellcolor{white} & & & CS Ratio & &   \\ 
					%\midrule
					\hline
					Model & 10$\%$ & 20$\%$ & 30$\%$ & 40$\%$ & 50$\%$  \\
					%\midrule 
					\hline
				
                    TVAL3  & 22.56, 0.6589 & 26.48, 0.8057 & 29.07, 0.8717 & 31.25, 0.9112  &  33.23, 0.9370  \\
					ReconNet  & 23.59, 0.6694 & 26.34, 0.7988 & 27.64, 0.8526 & 30.38, 0.9005 & 31.03, 0.9104   \\
					ISTA-Net  & 26.33, 0.8003 & 30.43, 0.8924 & 33.08, 0.9318 & 35.21, 0.9522 & 37.13, 0.9652  \\
					%\midrule
					\hline
					TF-ISTA-Net & \underline{26.52, 0.8060} & \underline{30.98, 0.9021} & \underline{33.69, 0.9383} & \underline{35.97, 0.9578} & \underline{37.99, 0.9702}   \\
					RTF-ISTA-Net & 26.12, 0.7941 & 30.82, 0.8996 & 33.65, 0.9371 & 35.93, 0.9569 & 37.90, 0.9696   \\
					TF-FISTA-Net & \textbf{26.80, 0.8137} & \textbf{31.04, 0.9025} & \textbf{33.74, 0.9388} & \textbf{36.01, 0.9580} & \textbf{38.06, 0.9706}  \\
					RTF-FISTA-Net & 26.28, 0.7980 & 30.92, 0.9009 & 33.58, 0.9370 & 35.78, 0.9566 &  37.75, 0.9690  \\
					%\bottomrule 
					\hline
				\end{tabular}}
			\end{sc}
		\end{small}
	\end{center}
% 	\vskip -0.1in
\end{table*}

%----------------------------------------------- Table BSD68 --------------------------------------------------------
\begin{table*}
	\centering
	\caption{CS-IR on \textit{BSD68}: PSNR (in dB) and SSIM for compressed sensing image recovery using various competing methods. TF-FISTA-Net and TF-ISTA-Net have superior performance for all CS ratios. We observe the following order of performance: TF-FISTA-Net $\approx$ TF-ISTA-Net $>$ RTF-FISTA-Net $\approx$ RTF-ISTA-Net $>$ ISTA-Net $>$ TVAL3 $>$ ReconNet. The best performance is shown in \textbf{boldface} and the second best is shown \underline{underlined}.}
	\label{table:BSD68}
	\vskip 0.05in
	\begin{center}
		\begin{small}
			\begin{sc}
			\resizebox{0.95\linewidth}{!}{	\begin{tabular}{c||ccccc}
					\hline
					\cellcolor{white} & & & CS Ratio & &   \\ 
					\hline
					Model & 10$\%$ & 20$\%$ & 30$\%$ & 40$\%$ & 50$\%$  \\
					\hline
				
                    TVAL3  & 23.39, 0.6170 & 26.00, 0.7442 & 27.85, 0.8173 & 29.56, 0.8684 &  31.23, 0.9053  \\
					ReconNet  & 23.38, 0.6186 & 25.37, 0.7385 & 26.37, 0.8102 & 28.64, 0.8563  & 29.69, 0.8846   \\
					ISTA-Net  & 25.27, 0.6977 & 27.95, 0.8079 & 29.98, 0.8700 & 31.77, 0.9092 & 33.50, 0.9361  \\
					\hline
					TF-ISTA-Net & \underline{25.37, 0.7030} & \textbf{28.24, 0.8185} & \textbf{30.32, 0.8780} & \textbf{32.20, 0.9163} &  \underline{34.04, 0.9426}  \\
					RTF-ISTA-Net & 25.16, 0.6949 & 28.19, 0.8159 & 30.29, 0.8770 & 32.17, 0.9154 &  \underline{34.04, 0.9424}  \\
					TF-FISTA-Net & \textbf{25.44, 0.7059} & \underline{28.22, 0.8166} & \textbf{30.32, 0.8779} & \textbf{32.19, 0.9159} &  \textbf{34.07, 0.9432}  \\
					RTF-FISTA-Net & 25.20, 0.6938 & 28.19, 0.8160 & 30.26, 0.8747 & 32.13, 0.9151 & 33.95, 0.9416   \\
					\hline
				\end{tabular}}
			\end{sc}
		\end{small}
	\end{center}
% 	\vskip -0.1in
\end{table*}
%----------------------------------- Table Ur100 -----------------------------------------
\begin{table}[]
	\centering
	\caption{Compressed sensing image recovery (CS-IR) on \textit{Urban100} dataset: PSNR (in dB) and SSIM for recovery using various techniques. TF-FISTA-Net has a superior performance among all the techniques for all CS ratios. We observe the following order of performance: TF-FISTA-Net $>$ TF-ISTA-Net $\approx$ RTF-FISTA-Net $\approx$ RTF-ISTA-Net $>$ ISTA-Net $>$ TVAL3 $>$ ReconNet. The best figures are in \textbf{boldface} and the second best is \underline{underlined}.}
	\label{table:ur100}
	\vskip 0.05in
	\begin{center}
		\begin{small}
			\begin{sc}
			\resizebox{0.95\linewidth}{!}{	\begin{tabular}{c||ccccc}
					\hline
					 \cellcolor{white} & & & CS Ratio & &   \\ 
					\hline
					Model & 10$\%$ & 20$\%$ & 30$\%$ & 40$\%$ & 50$\%$  \\
					\hline
                    TVAL3  & 19.20, 0.5075 & 21.69, 0.6773 & 23.70, 0.7793 & 25.67, 0.8495 & 27.73, 0.8989   \\
					ReconNet  & 20.10, 0.5186 & 20.86, 0.6682 & 21.42, 0.7541 & 24.75, 0.8404 & 26.10, 0.8752   \\
					ISTA-Net  & 21.18, 0.6352 & 24.45, 0.7990 & 27.01, 0.8783 & 29.33, 0.9238 & 31.56, 0.9507  \\
					\hline
					TF-ISTA-Net & \underline{21.28, 0.6448} & \underline{24.89, 0.8157} & 27.68, 0.8920 & 30.17, 0.9341 &  \underline{32.77, 0.9602}  \\
					RTF-ISTA-Net & 21.10, 0.6343 & 24.84, 0.8115 & 27.69, 0.8917 & \underline{30.26, 0.9344} & 32.70, 0.9597   \\
					TF-FISTA-Net & \textbf{21.38, 0.6512} & \textbf{24.97, 0.8176} & \textbf{27.76, 0.8943} & \textbf{30.37, 0.9364} & \textbf{32.80, 0.9607}   \\
					RTF-FISTA-Net & 21.20, 0.6353 & \underline{24.90, 0.8157} & \underline{27.70, 0.8924} & 30.21, 0.9343 & 32.58, 0.9590   \\
					\hline
				\end{tabular}}
			\end{sc}
		\end{small}
	\end{center}
% 	\vskip -0.1in
\end{table}
%----------------------------------- Table DIV2K -----------------------------------
\begin{table}[]
	\centering
	\caption{Compressed sensing image recovery on \textit{DIV2K} dataset: PSNR (in dB) and SSIM for various techniques. TF-FISTA-Net has a superior performance among all the techniques for most CS ratios. We observe the following order of performance: TF-FISTA-Net $>$ TF-ISTA-Net $>$ RTF-FISTA-Net $\approx$ RTF-ISTA-Net $>$ ISTA-Net $>$ TVAL3 $>$ ReconNet. The best performance is shown in \textbf{boldface} and the second best is shown \underline{underlined}.}
	\label{table:DIV2K}
	\vskip 0.05in
	\begin{center}
		\begin{small}
			\begin{sc}
			\resizebox{0.95\linewidth}{!}{	\begin{tabular}{c||ccccc}
					\hline
					 & & & CS Ratio & &   \\ 
					\hline
					Model & 10$\%$ & 20$\%$ & 30$\%$ & 40$\%$ & 50$\%$  \\
					\hline
                    TVAL3  & 22.79, 0.6328 & 25.81, 0.7760 & 28.02, 0.8520 & 30.04, 0.9003 & 31.98, 0.9326   \\
					ReconNet  & 23.02, 0.6351 & 25.23, 0.7653 & 26.42, 0.8417 & 29.21, 0.8952 & 31.03, 0.9153   \\
					ISTA-Net  & 25.29, 0.7439 & 28.76, 0.8627 & 31.41, 0.9186 & 33.69, 0.9488 &  35.83, 0.9666 \\
					\hline
					TF-ISTA-Net & \underline{25.38, 0.7476} & \textbf{29.16, 0.8735} & \underline{31.92, 0.9259} & \underline{34.35, 0.9543} & \underline{36.68, 0.9714}   \\
					RTF-ISTA-Net & 25.18, 0.7398 & 29.05, 0.8705 & 31.80, 0.9244 & 34.23, 0.9532 &  36.53, 0.9707  \\
					TF-FISTA-Net & \textbf{25.49, 0.7525} & \underline{29.16, 0.8723} & \textbf{31.98, 0.9269} & \textbf{34.39, 0.9547} & \textbf{36.73, 0.9719}   \\
					RTF-FISTA-Net & 25.22, 0.7394 & 29.08, 0.8715 & 31.79, 0.9244 & 34.13, 0.9526 &  36.35, 0.9699  \\
					\hline
				\end{tabular}}
			\end{sc}
		\end{small}
	\end{center}
% 	\vskip -0.1in
\end{table}
% % -------------------------------------------------------------------------
\paragraph{Initialization and Training Performance} 
The parameters of the deep sparsifying transform $\cl F$ are initialized using the Xavier weight initialization scheme \cite{glorot2010understanding}, {\color{black}which prevents gradients from exploding or vanishing}. The step-size $\eta^{(k)}$ is initialized to 0.5 for every unfolding, and the shrinkage-thresholding parameter $\lambda^{(k)}$ is initialized to 0.01. The default learning rate is $10^{-3}$. The difference between  ISTA-Net/FISTA-Net and the proposed variants lies in the gradient computation step $\mathcal{G}$ (cf. Figure~\ref{fig:schematic_model}):
The gradient computation strategy is fixed and the learnable quantities are $\{ \eta^{(k)}, \lambda^{(k)}, \cl F_k, \tilde{\cl F}_k \}_{k=1}^K$. Figure~\ref{fig:CSIR_convergence_ISTA} shows a comparison of the PSNR versus epochs. Right at initialization, the techniques have different PSNRs with the PSNR of TF-ISTA-Net being superior to that of the others, which is  attributed to the new gradient computation. As the training proceeds, the learnt sparsifying transform gives performance improvements over the initialization. The convergence behavior is the best for TF-ISTA-Net closely followed by RTF-ISTA-Net, both of which being clearly superior to ISTA-Net.\\
\indent Overall, the inference from these results is that the TF variants have a headstart in performance owing to a superior gradient computation. { \color{black} Further, the intermediate reconstructions $\bld z^{(k)}$ in \eqref{eq:istanet_gradient} are also consistent with the measurements}. The results also show that the TF variants have a superior convergence behavior and require fewer epochs (shorter training time). These findings are consistent across various datasets. 
% ----------------------------------------------------------------------------------
\paragraph{Results} We test ISTA-Net \cite{ISTA-Net}, TF-ISTA-Net, TF-FISTA-Net, RTF-ISTA-Net, and RTF-FISTA-Net with $8$ unfoldings for CS ratios of $10\%$, $20\%$, $30\%$, $40\%$, and $50\%$. We also compare the performance of these models with TVAL3, which is a standard iterative CS-IR method \cite{li2013efficient}, and ReconNet, which is a deep learning model that is not obtained by unfolding an iterative algorithm \cite{reconnet}. The results are presented in Tables~\ref{table:Set11}, \ref{table:BSD68}, \ref{table:ur100} and \ref{table:DIV2K} for the datasets under consideration. Across all the datasets, we observe the following order of performance (read $>$ as ``better than''): TF-FISTA-Net $>$ TF-ISTA-Net $>$ RTF-FISTA-Net $>$ RTF-ISTA-Net $>$ ISTA-Net $>$ TVAL3 $>$ ReconNet, in terms of both PSNR and SSIM, for varying CS ratios. In the tables, the best performance is highlighted in boldface and the next best is underlined. Figure~\ref{fig:example_recon} shows a visual comparison across methods with 40\% CS ratio on an example image taken from the DIV2K dataset. It can be observed that RTF-ISTA-Net has a higher PSNR than other methods and is able to recover the details more accurately. From the magnified portions of the reconstructed images, we observe that the proposed methods have fewer artifacts compared with the benchmark techniques.
\paragraph{Effect of Number of Unfolding Steps} The optimal number of unfolding steps is determined experimentally by testing across Set11, BSD68, Urban100, and DIV2K datasets. Table~\ref{tableS:comp-image-recovery-unfoldings} shows the average PSNR and SSIM obtained on the test datasets. The PSNR and SSIM improve with an increase in the number of unfolding steps. However, the increase in performance beyond $8$ unfolding steps is marginal (limited to the second decimal place). Hence, we set the number of unfolding steps to $8$, which also keeps the model complexity small without compromising on the performance.\\
\indent Additional results on compressed sensing image recovery are presented in the Supplementary Material.
\begin{table*}[!t]
\centering
\caption{Effect of the number of unfolding steps on reconstruction performance: PSNR (in dB) and SSIM with varying number of unfolding steps of TF-ISTA-Net on various datasets.}
\label{tableS:comp-image-recovery-unfoldings}
\vskip 0.05in
\begin{center}
\begin{small}
\begin{sc}
\resizebox{0.9\linewidth}{!}{
	\begin{tabular}{c||ccccc}
		\toprule
		Unfoldings & 6 & 8 &  10 & 12 & 14 \\
		\midrule
		Dataset & (PSNR, SSIM) & (PSNR, SSIM) & (PSNR, SSIM) & (PSNR, SSIM) & (PSNR, SSIM) \\
		\midrule
		Set11 & 30.81, 0.8997 & 31.07, 0.9028 & 31.10, 0.9042 & 31.12, 0.9050 & 31.21, 0.9059 \\
		BSD68 & 28.16, 0.8158 & 28.25, 0.8190 & 28.28, 0.8198 & 28.30, 0.8209 & 28.32, 0.8213 \\
		Urban100 & 24.75, 0.8105 & 24.96, 0.8182 & 24.94, 0.8184 & 24.94, 0.8191 & 25.00, 0.8209  \\
		DIV2K & 29.06, 0.8711 & 29.21, 0.8746 & 29.23, 0.8752 & 29.27, 0.8761 & 29.30, 0.8765 \\
		\bottomrule 
	\end{tabular}}
\end{sc}
\end{small}
\end{center}
\end{table*}

% % --------------------------------------------
% % 	Discussions and Conclusion
% % --------------------------------------------
\section{Discussion and Conclusion}
\label{sec:conclusion}
We considered the analysis-sparse $\ell_1$-minimization problem with a generalized $\ell_2$-norm-based data-fidelity as the constraint. The proposed formulation results in the effective sensing matrix constituting a tight frame, thereby combining the best of both worlds: the ease of constructing random Gaussian sensing matrices; and the performance advantages accruing from the tight-frame property. We derived performance bounds on the $\ell_2$-error in reconstruction in terms of the noise level and the $\ell_1$-error in compressibility. The analysis on the performance bounds showed that the proposed formulation offers superior performance compared to the standard least-squares data-fidelity constraint when fewer measurements are available or the signal is less compressible.\\
\indent The constrained analysis-sparse $\ell_1$-minimization problem is solved in the unconstrained setting using proximal methods, effectively, offering a tight-frame counterpart to existing techniques. Further, we rescaled the gradients of the data-fidelity loss in the iterative updates to improve the accuracy of analysis-sparse recovery. We validated the claims via experiments on synthetic analysis-sparse signals, and showed that the proposed techniques outperformed the benchmarks in terms of reconstruction accuracy for varying sparsity, noise levels, and sensing matrix distributions. We demonstrated the application of the proposed techniques on compressed sensing image recovery with a deep sparsifying transform as the analysis operator, which enhanced the reconstruction quality compared  with the benchmarks, in terms of PSNR and SSIM, as validated on Set11, BSD68, Urban100 and DIV2K datasets. The improvements come with a modest one-time computational overhead, that of precomputing the inverse of a smaller Gram matrix of size $m \times m$ for an $m \times n$ sensing matrix ($m \ll n$).\\
\indent One could leverage the superior image reconstruction capabilities of the proposed techniques to related computational imaging modalities such as magnetic resonance imaging \cite{vaswani2009}, single-pixel camera \cite{duarte2008single}, seismic imaging \cite{yuan2013spectral, mache2023introducing}, etc., which is a promising direction for further research. 
%------------------------------------------------------------

% ----------------------------------------------------------------------------------------------------------------------------------------------------------
% ----------------------------------------------------------------------------------------------------------------------------------------------------------
% 	      									APPENDICES BEGIN
% ----------------------------------------------------------------------------------------------------------------------------------------------------------
% ----------------------------------------------------------------------------------------------------------------------------------------------------------

\appendix

% % --------------------------------------------
% % 	Statistical Significance
% % --------------------------------------------

\section{Statistical Significance}
Sparse signal recovery from noisy measurements indirectly also results in an estimate of the noise. Let $\hat{\bld x}$ be the solution to the optimization problem in \eqref{eq:unconstrained-P1-analysis} with $\bd D = \bd I$. The corresponding noise estimate is $\hat{\bld w} = \bld y - \bd A \hat{\bld x}$. Ideally, we expect $\hat{\bld w}$ to have the same statistics as ${\bld w}$, i.e., zero mean and uncorrelated. Consider
\begin{equation*}
\begin{split}
	\hat{\bld w} = \bld y - \bd A \hat{\bld x}
	= \bd A \bld x^* + \bld w - \bd A \hat{\bld x}
	= \bd A (\bld x^* - \hat{\bld x}) + \bld w.
\end{split}
\end{equation*}
The estimation error in noise is given by $\Delta \bld w \overset{\text{def.}}{=} \bld w-\hat{\bld w} = \bd A \Delta \bld x$, where $\Delta \bld x = \bld x^* - \hat{\bld x}$ is the estimation error in the signal $\bld x$. The expected estimation error in $\bld w$ is given by $\bb E[\Delta \bld w] = \bd A\bb E[\Delta \bld x] \neq \bld 0$, since the $\ell_1$ penalty is known to result in biased estimates \cite{zhang2010nearly}. The bias can be corrected by an appropriate choice of the sparsity enforcing penalty, for example, the minimax-concave penalty (MCP) \cite{zhang2010nearly,pokala2021iteratively} instead of the $\ell_1$ penalty. Let us now examine the correlation matrix of $\Delta \bld w$:
\begin{equation*}
    \begin{split}
        R_{\hat{\bld w}} \overset{\text{def.}}{=} \bb E (\Delta \bld w \Delta \bld w^\TT)
        = \bb E (\bd A \Delta \bld x \Delta \bld x^\TT \bd A^\TT ) 
        = \bd A \bb E (\Delta \bld x \Delta \bld x^\TT) \bd A^\TT
        = \bd A R_{\hat{\bld x}} \bd A^\TT,
    \end{split}
\end{equation*}
where $R_{\hat{\bld x}}$ denotes the signal correlation matrix, which is positive semi-definite. Considering its eigen-decomposition $R_{\hat{\bld x}}=\bd U_{\Delta \bld x} \bd \Lambda_{\Delta \bld x}  \bd U_{\Delta \bld x}^\TT$, we get $R_{\hat{\bld w}} = \left(\bd A \bd U_{\Delta \bld x}\right)  \bd \Lambda_{\Delta \bld x}  \left(\bd A \bd U_{\Delta \bld x} \right)^\TT$.
The following result relates the statistics of the signal and noise estimation errors. 
\begin{proposition}
Suppose $\bd A$ is a tight-frame matrix. If the signal estimation error $\Delta \bld x$  has i.i.d. entries, then the noise estimation error $\Delta \bld w$ also has i.i.d. entries, and vice versa.
\end{proposition}
\begin{proof}
If $\bd \Lambda_{\Delta \bld x} = \bd I$, then $R_{\hat{\bld w}} = \bd A \bd A^\TT = \bd I$, by virtue of the tightness of the sensing matrix. Conversely, if $R_{\hat{\bld w}} = \bd I$, then 
$\left(\bd A \bd U_{\Delta \bld x}\right) \bd \Lambda_{\Delta \bld x}  \left( \bd A \bd U_{\Delta \bld x}\right)^\TT = \bd I = \bd A \bd A^\TT.$
Noting that  $\bd U_{\Delta \bld x} \bd U_{\Delta \bld x}^\TT = \bd I$ gives
$\left(\bd A \bd U_{\Delta \bld x}\right)  \bd \Lambda_{\Delta \bld x}  \left( \bd A \bd U_{\Delta \bld x}\right)^\TT = \bd A \bd U_{\Delta \bld x} \bd U_{\Delta \bld x}^\TT \bd A^\TT,$
which leads to $\Lambda_{\Delta \bld x} = \bd I$, since this equality must hold for any tight-frame sensing matrix $\bd A$.
\end{proof}
Thus, a tight-frame sensing matrix ensures that the noise estimate is consistent with the assumptions on the noise made in the observation model.
 
% ---------------------------------------------------------------------------------
% 	Relationship Between $\hat{\delta}_{2s}$ and $\hat{\delta}_{ps}$}
% ---------------------------------------------------------------------------------
%\clearpage
\section{Relationship Between RICs Corresponding to Different Orders of Sparsity}
% % ----------------------- Consequence of Corollary 3.4 from CoSAMP paper ----------
The following lemma establishes that $\norm{\bd A}^2_2 > 1$.
\begin{lemma}[$\norm{\bd A}^2_2 > 1$]\label{lem:A_2_greater_1} 
Let $\bd A \in \bb R^{m \times n}, m < n$ satisfy RIP. Then, $\norm{\bd A}_2^2 > 1$.
\end{lemma}
\begin{proof}
    Recall the RIP condition of $\bd A$ with RIP constant $0 < \delta_s < 1$:
    \begin{equation} \label{eq:delta_s_equal} \nonumber
       \left(1-\delta_s \right)\norm{\bld x}_2^2 \leq  \norm{\bd A \bld x}_2^2 \leq \left(1+\delta_s \right)\norm{\bld x}_2^2,
    \end{equation}
for all $s$-sparse $\bld x \in \mathbb{R}^n$.\\
By definition of the matrix 2-norm \cite{horn2012matrix}, we have
$\norm{\bd A \bld x}_2^2 \leq \norm{\bd A}_2^2 \norm{\bld x}^2_2, \forall \bld x \in \bb R^n.$
Since $s$-sparse vectors in $\mathbb{R}^n$ are a subset of $\mathbb{R}^n$, we have the inequality
$\norm{\bd A}_2^2 \geq 1+\delta_s > 1$,
    since $\delta_s > 0$.
\end{proof}
% ---------------------------------------------------------------------------------
\begin{lemma}[$\bd B$-norm RIP relations]\label{lem:delta_ps_2s}
Let $p$ and $s$ be positive integers and let $\bd A$ satisfy RIP with RIC $\delta_s$, and $\bd B$-norm RIP with $\bd B$-norm RIC $\hat{\delta}_s$. Then, $\hat{\delta}_{ps} \leq p \cdot \hat{\delta}_{2s}$.
\end{lemma}
\begin{proof}
From \eqref{eq:delta_relation}, we have
$\delta_s = 1- \norm{\bd A}_2^2 \left(1-\hat{\delta}_s\right)\text{ and } \delta_{ps} = 1- \norm{\bd A}_2^2 \left(1-\hat{\delta}_{ps}\right)$.
We use Corollary 3.4 from \cite{needell2009cosamp}, which states that ${\delta}_{ps} \leq p \cdot {\delta}_{2s}$ for positive integers $p$ and $s$. Combining the preceding equations gives
\begin{equation} \nonumber
\begin{split}
{\delta}_{ps} = 1- \norm{\bd A}_2^2 \left(1-\hat{\delta}_{ps}\right) \leq p \left(1- \norm{\bd A}_2^2 \left(1-\hat{\delta}_{2s}\right) \right) &\implies \hat{\delta}_{ps} \leq \left( p - 1 \right) \left( \frac{1}{\norm{\bd A}_2^2} - 1 \right) + p \cdot \hat{\delta}_{2s} \\ 
&\implies \hat{\delta}_{ps} \leq  p \cdot \hat{\delta}_{2s}, \text{ since } \norm{\bd A}_2^2 > 1.
\end{split}
\end{equation}
\end{proof}
% ---------------------------------------------------------------------------------
% 	Proof of Theorem~\ref{thm:MSE_D}
% ---------------------------------------------------------------------------------

\section{Proof of Theorem~\ref{thm:MSE_D}}
\label{appendix:mse_bounds}
The mechanism of the proof of Theorem~\ref{thm:MSE_D} follows that provided by Cand\`es \cite{candes_analysis} with the $\ell_2$ norm adapted to $\bd B$-norm. We recall relevant lemmas from \cite{candes_analysis}, which are cited in context.\\
\indent We begin with the proof of Theorem~\ref{thm:MSE_D}. Let $\bld h = \bld x^* - \hat{\bld x}$, where $\bld x^*$ is the ground-truth signal. We seek to bound $\|\bld h \|_2$. Assume $\bd T_0$ be the set of largest $s$ coefficients of $\bd D^\TT \bld x^*$ in magnitude. Denote by $\bd D_{\bd T}$ the matrix $\bd D$ restricted to the columns indexed by $\bd T$.
\begin{lemma}[Cone constraint \cite{candes_analysis}]\label{lem:cone}
The vector $\bd D^{\TT} \bld h$ obeys the following cone constraint, $$\| \bd D_{\bd T_0^c}^{\TT} \bld h \|_1 \leq 2 \| \bd D_{\bd T_0^c}^{\TT} \bld x^* \|_1 + \| \bd D_{\bd T_0}^{\TT} \bld h \|_1.$$
\end{lemma}
The set complement of $\bd T_0$ is denoted by $\bd T^c_0$. We divide $\bd T^c_0$ into subsets of size $M$ in the order of decreasing magnitude of $\bd D_{\bd T^c_0}^{\TT}\bld h$. We denote these sets as $\bd T_1, \bd T_2, \cdots$, and $\bd T_{01} = \bd T_0 \cup \bd T_1$. 
\begin{lemma}[Bounding the tail \cite{candes_analysis}]\label{lem:bound_tail}
Let $\displaystyle\rho = \frac{s}{M}$ and $\displaystyle\eta = \frac{2 \left \|  \bd D_{\bd T_0^c}^{\TT} \bld x^* \right\|_1}{\sqrt{s}} $. Then, $$\displaystyle\sum_{j \geq 2} \| \bd D_{\bd T_j}^{\TT} \bld h \|_2 \leq \sqrt{\rho} \left( \| \bd D_{\bd T_0}^{\TT} \bld h \|_2 + \eta \right).$$
\end{lemma}

% ---------------------------------------------------------------------------------
\begin{lemma}[Tube constraint]\label{lem:tube}
The vector $\bd A \bld h$ satisfies following bound: $\left \| \bd A \bld h \right \|_{\bd B} \leq 2 \epsilon$.
\end{lemma}
\begin{proof}
Since $\hat{\bld x}$ is a minimizer of \eqref{eq:P1-TF},  it satisfies the constraint 
$\norm{\bd A \hat{\bld x} - \bld y}_{\bd B} \leq \epsilon.$
We also know that $\bld x^*$ satisfies the feasibility conditions of \eqref{eq:P1-TF}, which leads to
$\norm{\bd A {\bld x}^* - \bld y}_{\bd B} \leq \epsilon.$
The error term $\bld h = \bld x^* - \hat{\bld x}$ satisfies, 
$\norm{\bd A \bld h}_{\bd B} = \norm{\bd A (\bld x^* - \hat{\bld x})}_{\bd B} \leq \norm{\bd A \bld x^*}_{\bd B} + \norm{\bd A \hat{\bld x}}_{\bd B} \leq 2\epsilon.$
\end{proof}
% ---------------------------------------------------------------------------------
Next, we make use of the $\bd B$-norm D-RIP condition.
\begin{lemma}[Consequence of $\bd B$-norm D-RIP]\label{lem:consequence}
The following inequality holds
\begin{equation*}   
    \sqrt{1-\hat{\delta}_{s+M}} \left \| \bd D_{\bd T_{01}} \bd D_{\bd T_{01}}^{\TT} \bld h \right \|_2 - \sqrt{\rho} \left( \| \bld h \|_2 + \eta \right)  \leq 2 \epsilon.
\end{equation*}
\end{lemma}
\begin{proof}
We start with Lemma~\ref{lem:tube} and use the property that $\bd D$ is tight i.e., $\bd D \bd D^\TT = \bd I$.
\begin{equation} \nonumber
\begin{split}
2\epsilon &\geq \norm{\bd A \bld h}_{\bd B} = \norm{\bd A \bd D \bd D^\TT \bld h}_{\bd B}
= \norm{\sum_{j\geq0} \bd A \bd D_{T_j} \bd D_{T_j}^\TT \bld h}_{\bd B}
=\norm{  \bd A \bd D_{T_01} \bd D_{T_01}^\TT \bld h +  \sum_{j\geq2} \bd A \bd D_{T_j} \bd D_{T_j}^\TT \bld h}_{\bd B}, \\
&\geq \norm{  \bd A \bd D_{T_{01}} \bd D_{T_{01}}^\TT \bld h}_{\bd B} -  \norm{\sum_{j\geq2} \bd A \bd D_{T_j} \bd D_{T_j}^\TT \bld h}_{\bd B}
\geq \norm{  \bd A \bd D_{T_{01}} \bd D_{T_{01}}^\TT \bld h}_{\bd B} -  \sum_{j\geq2} \norm{ \bd A \bd D_{T_j} \bd D_{T_j}^\TT \bld h}_{\bd B}, \\
&\geq \sqrt{1-\hat{\delta}_{s+M}} \norm{ \bd D_{T_{01}} \bd D_{T_{01}}^\TT \bld h}_2 -  \sum_{j\geq2} \norm{\bd D_{T_j} \bd D_{T_j}^\TT \bld h}_2,\\
&\geq \sqrt{1-\hat{\delta}_{s+M}} \norm{ \bd D_{T_{01}} \bd D_{T_{01}}^\TT \bld h}_2 - \sqrt{\rho} \left(  \norm{ \bd D_{\bd T_0}^{\TT} \bld h }_2 + \eta \right) \text{(invoking Lemma~\ref{lem:cone}),}\\
&\geq \sqrt{1-\hat{\delta}_{s+M}} \norm{ \bd D_{T_{01}} \bd D_{T_{01}}^\TT \bld h}_2 - \sqrt{\rho} \left(  \norm{ \bld h }_2 + \eta \right). (\text{since } \bd D_{T_j}^\TT \text{ is tight, } \norm{\bd D_{T_j} \bld x}_2 = \norm{\bld x}_2 \forall j)
\end{split}
\end{equation}
\end{proof}
% ---------------------------------------------------------------------------------
We now bound the norm of the error vector $\bld h$.
\begin{lemma}[Bounding the error \cite{candes_analysis}]\label{lem:bound_error}
The error vector $\bld h$ has $\ell_2$ norm that satisfies
\begin{equation*}   
     \left \| \bld h \right \|_2^2  \leq \left \| \bld h \right \|_2  \left \| \bd D_{\bd T_{01}} \bd D_{\bd T_{01}}^{\TT} \bld h \right \|_2 + \rho \left(  \| \bd D^\TT_{\bd T_0} \bld h \|_2 + \eta  \right)^2.
\end{equation*}
\end{lemma}
Recall the following result based on arithmetic mean (AM) -- geometric mean (GM) inequality. 
\begin{lemma}
\label{lem:algebra}
For real values $u$, $v$, and $c > 0$, we have the inequality:
\begin{equation*}   
	| u v | \leq \frac{c u^2}{2} + \frac{v^2}{2c}.
\end{equation*}
\end{lemma} 
% ---------------------------------------------------------------------------------
Now, we proceed with the proof of the main result. Apply Lemma \ref{lem:algebra} to the tail error $\left \| \bld h \right \|_2  \left \| \bd D_{\bd T_{01}} \bd D_{\bd T_{01}}^{\TT} \bld h \right \|_2$ to obtain 
\begin{equation}\label{eq:h_Dh}
\left \| \bld h \right \|_2  \left \| \bd D_{\bd T_{01}} \bd D_{\bd T_{01}}^{\TT} \bld h \right \|_2 \leq \frac{ c_1 \| \bld h\|_2^2}{2} + \frac{\left \| \bd D_{\bd T_{01}} \bd D_{\bd T_{01}}^{\TT} \bld h \right \|_2^2}{2c_1}.
\end{equation}
Similarly,
\begin{align}\label{eq:h_eta}
\left(  \| \bld h \|_2 + \eta  \right)^2 &=  \| \bld h \|_2^2 + \eta^2 + 2  \| \bld h \|_2 \eta \leq   \| \bld h \|_2^2 + \eta^2 + 2 \left(  \frac{c_2 \| \bld h \|_2^2}{2} + \frac{\eta^2}{2c_2} \right).
\end{align}
Using \eqref{eq:h_Dh} and \eqref{eq:h_eta} in Lemma~\ref{lem:bound_error} leads to
\begin{align}
\left \| \bld h \right \|_2^2 & \leq \frac{ c_1 \| \bld h\|_2^2}{2} + \frac{\left \| \bd D_{\bd T_{01}} \bd D_{\bd T_{01}}^{\TT} \bld h \right \|_2^2}{2c_1} + \rho \left( \| \bld h \|_2^2 + \eta^2 + 2 \left(  \frac{c_2 \| \bld h \|_2^2}{2} + \frac{\eta^2}{2c_2} \right)\right). \nonumber
\end{align}
Simplifying the above leads to
\begin{align}
\left( 1 - \frac{c_1}{2} - \rho - \rho c_2 \right) \left \| \bld h \right \|_2^2 &\leq \frac{1}{2c_2} \left \| \bd D_{\bd T_{01}} \bd D_{\bd T_{01}}^{\TT} \bld h \right \|_2^2 + \left( \rho + \frac{\rho}{c_2} \right) \eta^2. \nonumber
\end{align}
Using the fact that $\sqrt{u^2 + v^2} \leq u + v $ for $u, v \geq 0$, we can simplify the above expression further to obtain the desired lower bound
\begin{equation} \label{eq:Dh_greater}
\left \| \bd D_{\bd T_{01}} \bd D_{\bd T_{01}}^{\TT} \bld h \right \|_2 \geq \left \| \bld h \right \|_2 \sqrt{2c_1\left( 1 - \frac{c_1}{2} - \rho -\rho c_2\right)} - \eta \sqrt{2c_1\left( \frac{\rho}{c_2} + \rho \right).}
\end{equation}
From Lemma~\ref{lem:consequence}, we know
\begin{equation} \label{eq:Dh_lower}
\left \| \bd D_{\bd T_{01}} \bd D_{\bd T_{01}}^{\TT} \bld h \right \|_2 \leq \frac{2 \epsilon}{\sqrt{1 - \hat{\delta}_{s+M}} } + \sqrt{\frac{\rho}{1- \hat{\delta}_{s+M}}} \left( \| \bld h\|_2 + \eta \right). 
\end{equation}
Combining inequalities in eqs.~\eqref{eq:Dh_greater} and \eqref{eq:Dh_lower} leads to 
\begin{align}
\frac{2 \epsilon}{\sqrt{1 - \hat{\delta}_{s+M}} } 
+ \sqrt{\frac{\rho}{1- \hat{\delta}_{s+M}}} \left( \| \bld h\|_2 + \eta \right) \geq 
 \left \| \bld h \right \|_2 \sqrt{2c_1\left( 1 - \frac{c_1}{2} - \rho -\rho c_2\right)} 
 - \eta \sqrt{2c_1\left( \frac{\rho}{c_2} + \rho \right)}. \nonumber
\end{align}
Therefore, we have
$\displaystyle\| \bld h \|_2 \leq \frac{2\epsilon}{K_1} + \frac{K_2}{K_1} \eta$, where
\begin{equation} \label{eq:final_expression_for_perf_bounds}
K_1 = \sqrt{2c_1\left( 1 - \frac{c_1}{2} -\rho -\rho c_2 \right)\left( 1 - \hat{\delta}_{s+M} \right)} - \sqrt{\rho},
K_2 = \sqrt{2c_1\left(\frac{\rho}{c_2} + \rho \right) \left(1-\hat{\delta}_{s+M} \right) } - \sqrt{\rho}.
\end{equation}
The parameters $c_1, c_2,$ and $M$ must be chosen such that $K_1$ and $K_2$ are both positive and the upper-bound is as tight as possible. \qed\\
\indent The selection of $M$ smaller than $4s$ leads to a scenario where there exist no $c_1$ and $c_2$ such that $K_1 > 0$. We chose $M = 6s$ akin to \cite{candes_analysis}. Consider $\hat{\delta}_{2s} < 0.08$, which leads to $\hat{\delta}_{7s} < 0.56$ by virtue of Lemma~\ref{lem:delta_ps_2s}. Choose $c_1 = 0.75, M = 6s, c_2 = 0.234$ such that $K_1$ is positive and the upper bound is the tightest, in which case the $\ell_2$-norm bound takes the form $\|\bld h\|_2 \leq 16.97 \epsilon + 3.00 \eta$.\\
\indent For canonical sparse coding, i.e., the dictionary $\bd D = \bd I$ (identity), the formulation reduces to
\begin{equation}
   \underset{\boldsymbol{x}\in\bb R^n}{\text{minimize }} \|  \boldsymbol{x}\|_1 \ \text{subject to }\ \| \mathbf{A}\boldsymbol{x} - \boldsymbol{y} \|_{\bd B} \leq \epsilon.
 \label{eq:P1-C-TF}
\end{equation}
In this scenario, we have the following corollary of Theorem~\ref{thm:MSE_D}:
\begin{corollary}
Let $\bd A$ satisfy the $\bd B$-norm RIP with $\hat{\delta}_{2s} < 0.08$. The minimizer $\bld x^*$ of \eqref{eq:P1-C-TF} satisfies
\begin{equation}
	\| \hat{\bld x} - \bld x^* \|_2 \leq C_0^{\textsc{TF}} \epsilon + C_1^{\textsc{TF}} \frac{\| \bld x^* - \left(  \bld x^* \right)_s \|_1}{\sqrt{s}},
\end{equation}
where the constants $C_0^{\textsc{TF}}$ and $C_1^{\textsc{TF}}$ depend on $\hat{\delta}_{2s}$.
\label{corollary:MSE_canonical}
\end{corollary}

% ----------------------------------------
% ----------------------------------------
% 	      ADDITIONAL RESULTS
% ----------------------------------------
% ----------------------------------------
\section{Additional Results on Compressed Sensing Image Recovery}
Tables~\ref{tableS:ur100_highres}~and \ref{tableS:div2k_highres} show the performance metrics for the reconstructions obtained on Urban100 and DIV2K high-resolution datasets, respectively. The proposed TF variants result in a superior performance over the benchmark techniques. Among the proposed ones, TF-FISTA-Net is the best. 

\begin{table}[htb]
\centering
\caption{\noindent Compressed sensing image recovery on \textit{Urban100} (High-Res): PSNR (in dB) and SSIM for recovery using various competitive methods. TF-FISTA-Net and RTF-ISTA-Net have superior performance among all methods for all CS ratios. We observe the following order of performance: TF-FISTA-Net $\approx$ RTF-ISTA-Net $>$ TF-ISTA-Net $>$ RTF-FISTA-Net $>$ ISTA-Net $>$ TVAL3 $>$ ReconNet. The best performance is shown in \textbf{boldface}, and the second best is shown \underline{underlined}.}
\label{tableS:ur100_highres}
\vskip 0.05in
\begin{center}
\begin{small}
\begin{sc}
\resizebox{.95\linewidth}{!}{
    \begin{tabular}{c||ccccc}
    \toprule
    & & & CS Ratio & &   \\ 
    \midrule
    Method & 10$\%$ & 20$\%$ & 30$\%$ & 40$\%$ & 50$\%$  \\
    \midrule
    TVAL3  & 20.72, 0.5942 & 23.67, 0.7398 & 25.90, 0.8206 & 27.94, 0.8744 & 29.92, 0.9118  \\
    ReconNet  & 21.54, 0.6124 & 22.93, 0.7021 & 23.56, 0.7987 & 26.85, 0.8656 & 28.42, 0.8934   \\
    ISTA-Net  & 23.40, 0.7142 & 26.90, 0.8375 & 29.43, 0.8954 & 31.52, 0.9282 & 33.46, 0.9496  \\
    \midrule
    TF-ISTA-Net & 23.57, \underline{0.7242} & 27.52, 0.8528 & 30.17, 0.9069 & 32.33, 0.9372 &  \underline{34.49, 0.9578}  \\
    RTF-ISTA-Net & \underline{23.64}, 0.7239 & \underline{27.61}, \textbf{0.8557} & \textbf{30.24, 0.9086} & \textbf{32.58, 0.9398} & \textbf{34.62, 0.9588}  \\
    TF-FISTA-Net & \textbf{23.76, 0.7314} & \textbf{27.63}, \underline{0.8545} & \underline{30.21, 0.9082} & \underline{32.46, 0.9390} & {34.48, 0.9580}   \\
    RTF-FISTA-Net & 23.36, 0.7114 & 27.54, 0.8534 & 30.17, 0.9071 & 32.40, 0.9526 & 34.44, 0.9576   \\
    \bottomrule 
    \end{tabular}}
\end{sc}
\end{small}
\end{center}
\end{table}

\begin{table}[htb]
    \centering
    \caption{\noindent Compressed sensing image recovery on \textit{DIV2K} (High-Res): PSNR (in dB) and SSIM for recovery using various competitive methods. TF-FISTA-Net and RTF-ISTA-Net have superior performance among all methods for all CS ratios. We observe the following order of performance: TF-FISTA-Net $>$ RTF-ISTA-Net $\approx$ TF-ISTA-Net $>$ RTF-FISTA-Net $>$ ISTA-Net $>$ TVAL3 $>$ ReconNet. The best performance is shown in \textbf{boldface}, and the second best is shown \underline{underlined}.}
    \label{tableS:div2k_highres}
    \vskip 0.05in
    \begin{center}
    \begin{small}
    \begin{sc}
    \resizebox{.95\linewidth}{!}{
        \begin{tabular}{c||ccccc}
        \toprule
        & & & CS Ratio & &   \\ 
        \midrule
        Method & 10$\%$ & 20$\%$ & 30$\%$ & 40$\%$ & 50$\%$  \\
        \midrule
        TVAL3  & 24.94, 0.6970 & 27.98, 0.8075 & 30.05, 0.8659 & 31.89, 0.9042 & 33.64, 0.9309  \\
        ReconNet  & 25.12, 0.6916 & 27.41, 0.7957 & 28.54, 0.8564 & 31.03, 0.8978 & 32.86, 0.9104   \\
        ISTA-Net  & 27.65, 0.7833 & 30.73, 0.8691 & 33.00, 0.9136 & 34.91, 0.9394 & 36.71, 0.9569  \\
        \midrule
        TF-ISTA-Net & \underline{27.74}, \underline{0.7871} & \textbf{31.08, 0.8775} & \underline{33.41, 0.9196} & 35.43, 0.9450 &  \underline{37.39, 0.9621}  \\
        RTF-ISTA-Net & 27.71, 0.7861 & \underline{31.06}, \underline{0.8769} & 33.39, 0.9196 & \textbf{35.45, 0.9450} & 37.36, 0.9619  \\
        TF-FISTA-Net & \textbf{27.80, 0.7885} & 31.06, 0.8763 & \textbf{33.43, 0.9199} & \textbf{35.45, 0.9452} & \textbf{37.41, 0.9625}   \\
        RTF-FISTA-Net & 27.53, 0.7767 & 31.02, 0.8757 & 33.31, 0.9184 & 35.31, 0.9438 & 37.21, 0.9610   \\
        \bottomrule 
        \end{tabular}}
    \end{sc}
    \end{small}
    \end{center}
\end{table}

% ----------------------------------------
% ----------------------------------------
% 	      EXAMPLE FIGURES
% ----------------------------------------
% ----------------------------------------

Figures~\ref{fig:example_recon_BSD68}-\ref{fig:example_recon_Set11} show results of compressed sensing image recovery with data drawn from BSD68, Urban100, DIV2K, and Set11 datasets. These results show that the proposed methods have superior reconstruction capabilities. The artifacts in the reconstructed images are significantly lower in the case of the proposed methods vis-\`a-vis the benchmark methods. The proposed TF variants result in a superior performance over the benchmark techniques. Among the proposed ones, TF-FISTA-Net performs the best. 
\begin{figure*}[htb]
\begin{center}
	\begin{tabular}[b]{P{.21\linewidth}P{.21\linewidth}P{.21\linewidth}P{.21\linewidth}}
		\includegraphics[width=1.0\linewidth]{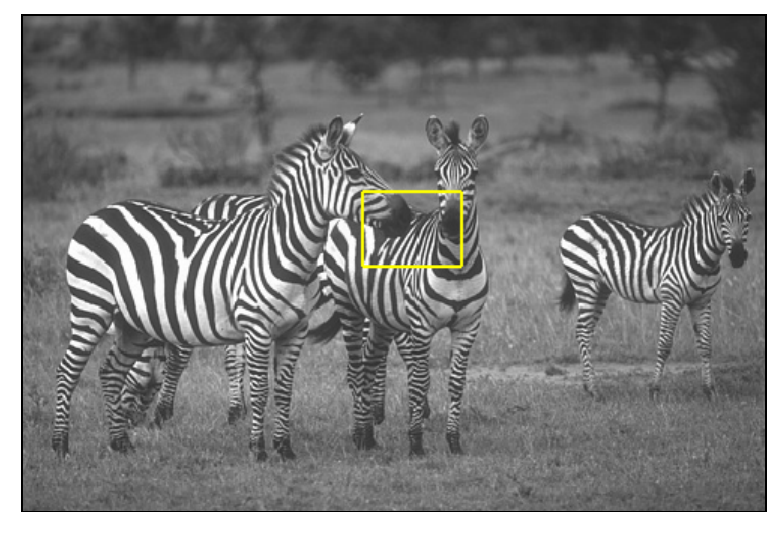}  &
		\includegraphics[width=1.0\linewidth]{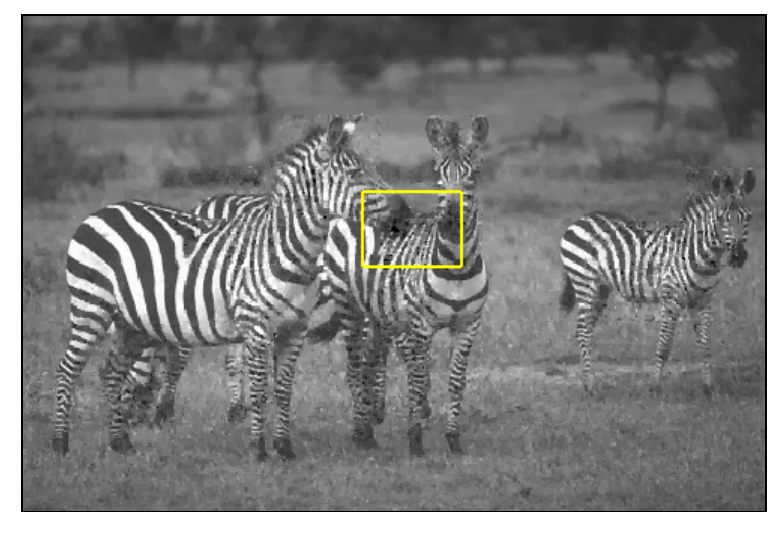}  &
		\includegraphics[width=1.0\linewidth]{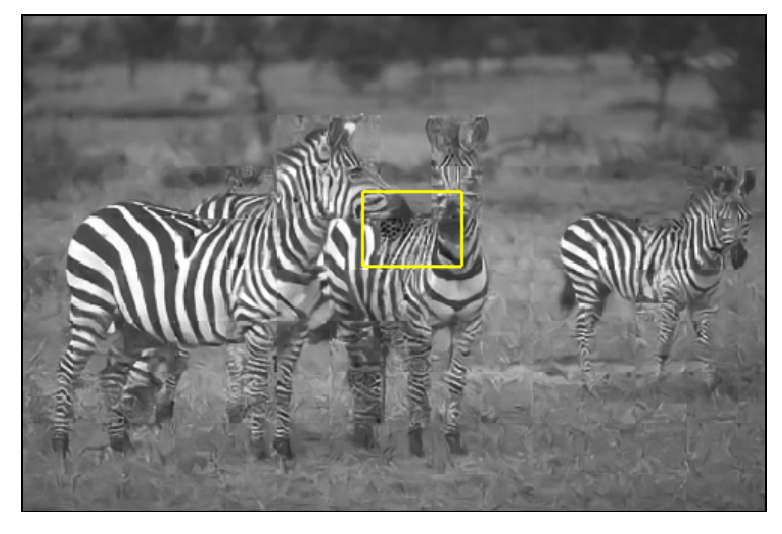} & 
		\includegraphics[width=1.0\linewidth]{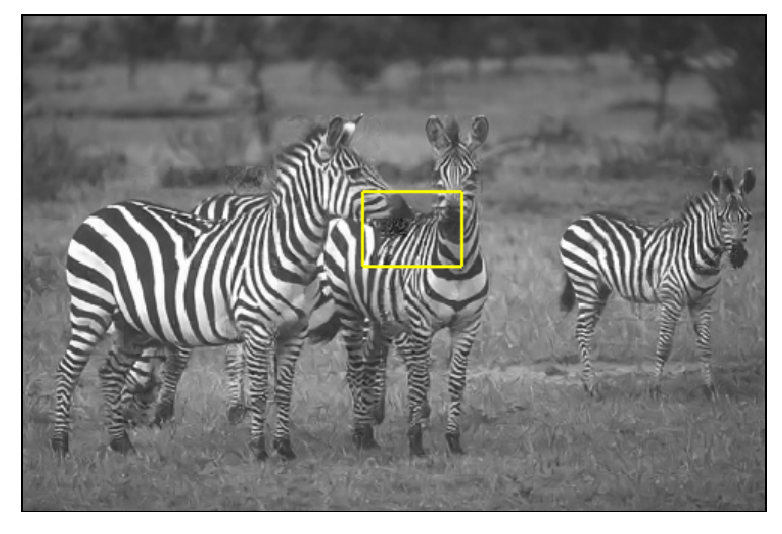} \\[-1pt]
		\includegraphics[width=1.0\linewidth]{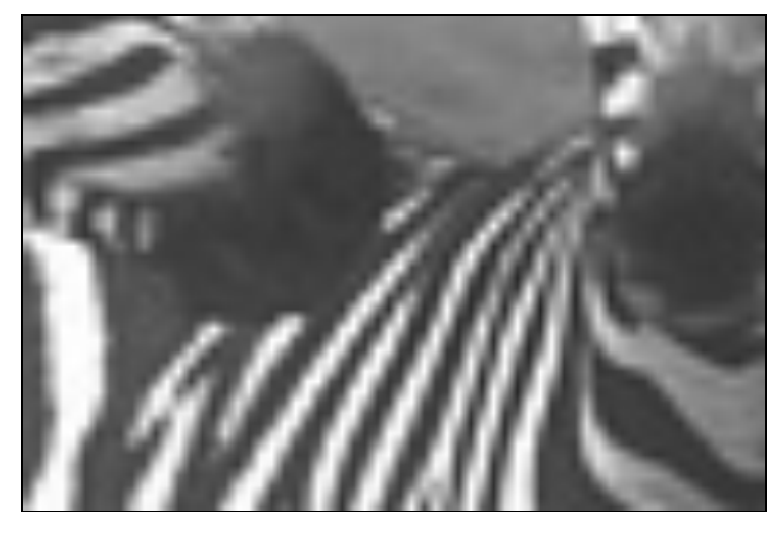}  &
		\includegraphics[width=1.0\linewidth]{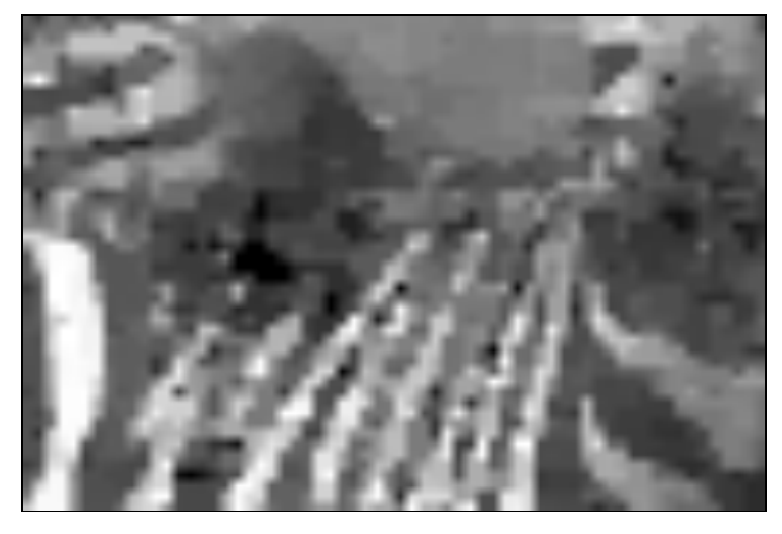}  &
		\includegraphics[width=1.0\linewidth]{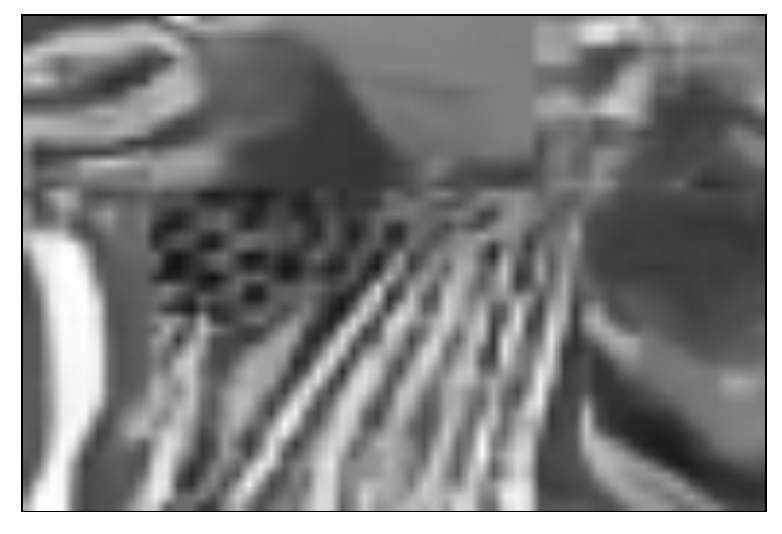} & 
		\includegraphics[width=1.0\linewidth]{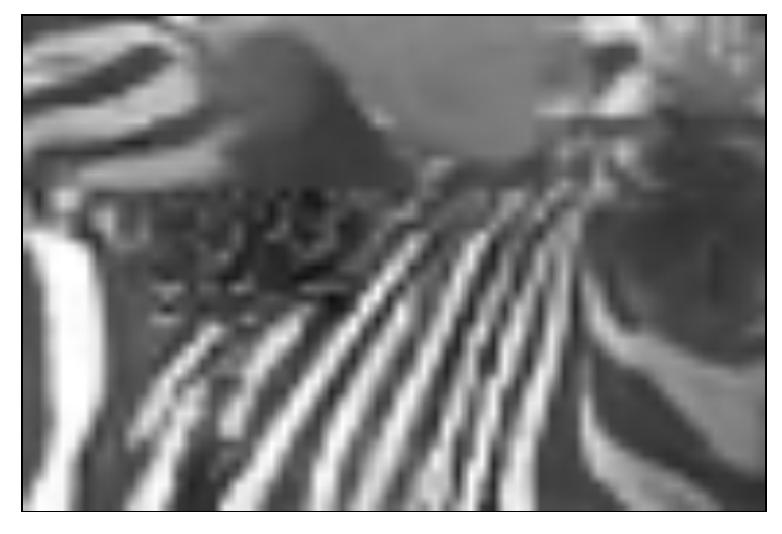} \\[-1pt]
		Original & TVAL3 & ReconNet & ISTA-Net \\ 
		& (25.03, 0.8445) & (24.51, 0.7938) & (28.98, 0.8976)\\ 
		\includegraphics[width=1.0\linewidth]{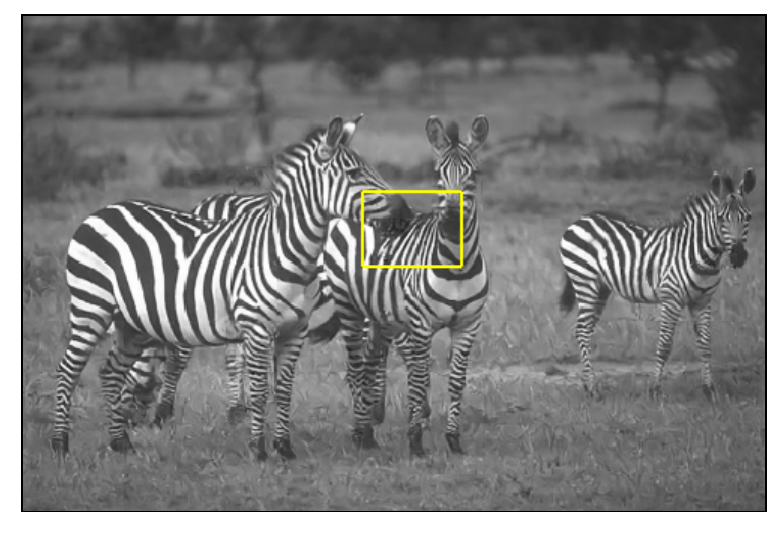}  &
		\includegraphics[width=1.0\linewidth]{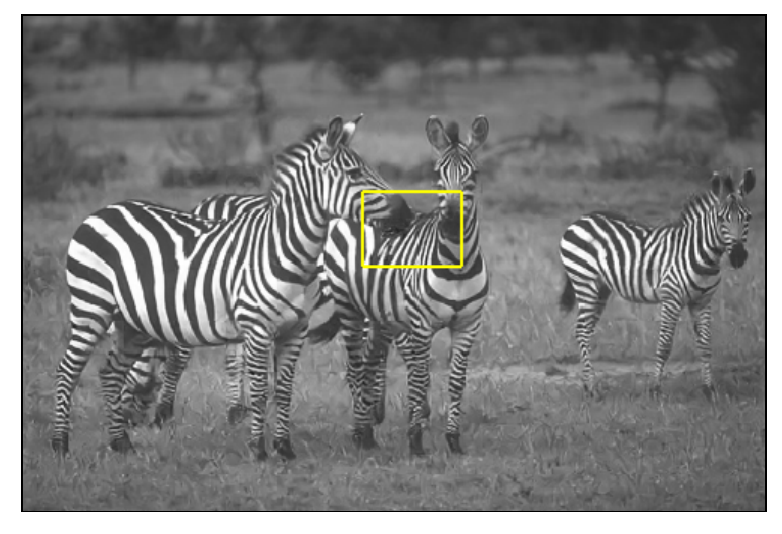}  &
		\includegraphics[width=1.0\linewidth]{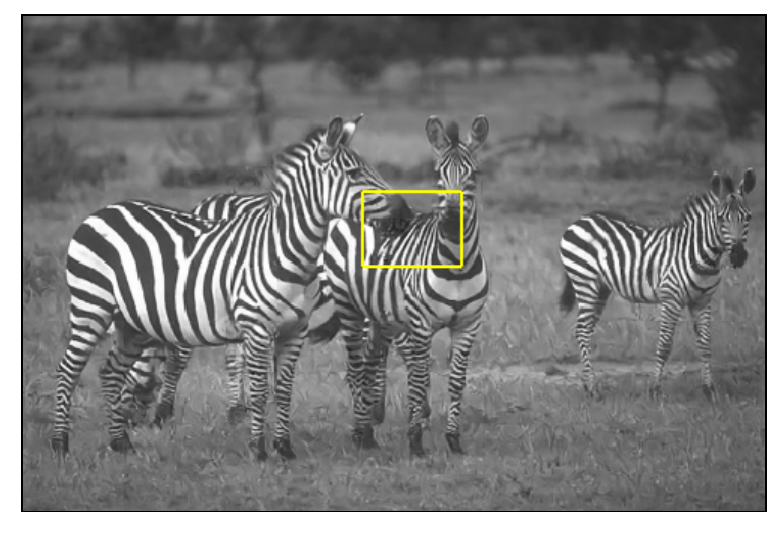} & 
		\includegraphics[width=1.0\linewidth]{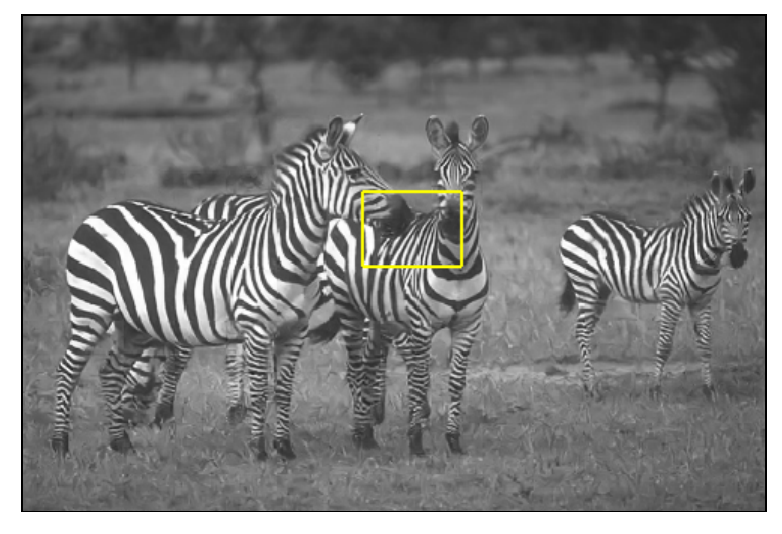} \\[-1pt]
		\includegraphics[width=1.0\linewidth]{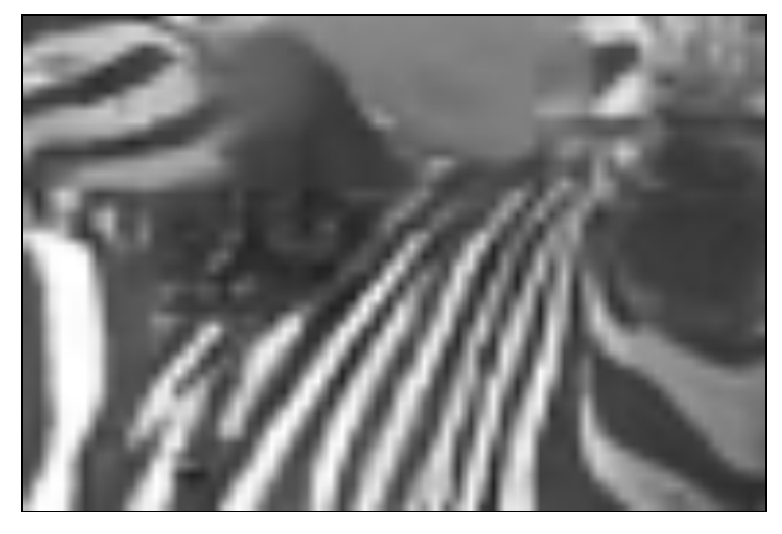}  &
		\includegraphics[width=1.0\linewidth]{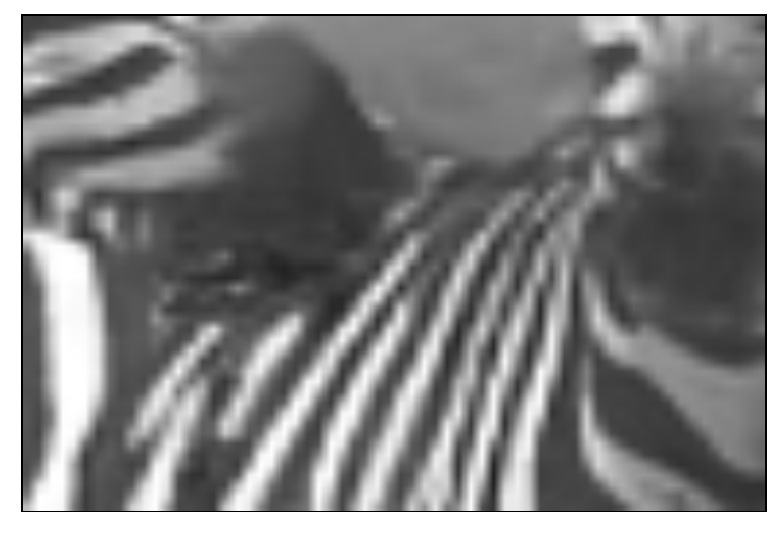}  &
		\includegraphics[width=1.0\linewidth]{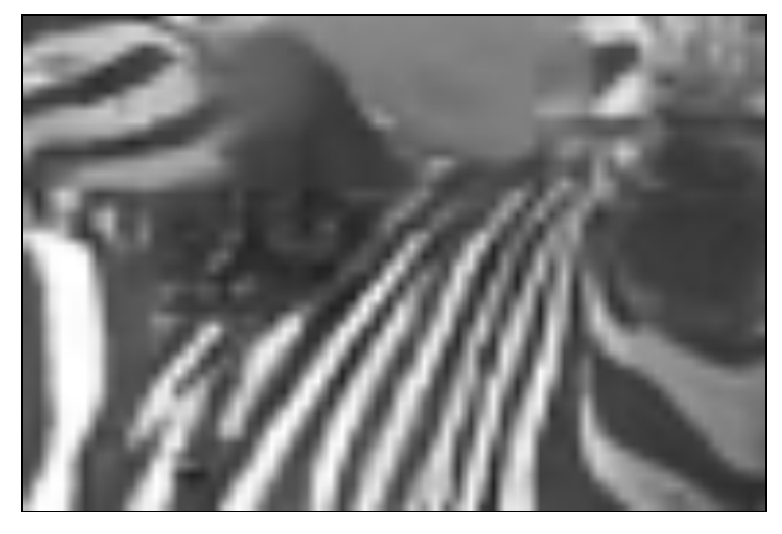} & 
		\includegraphics[width=1.0\linewidth]{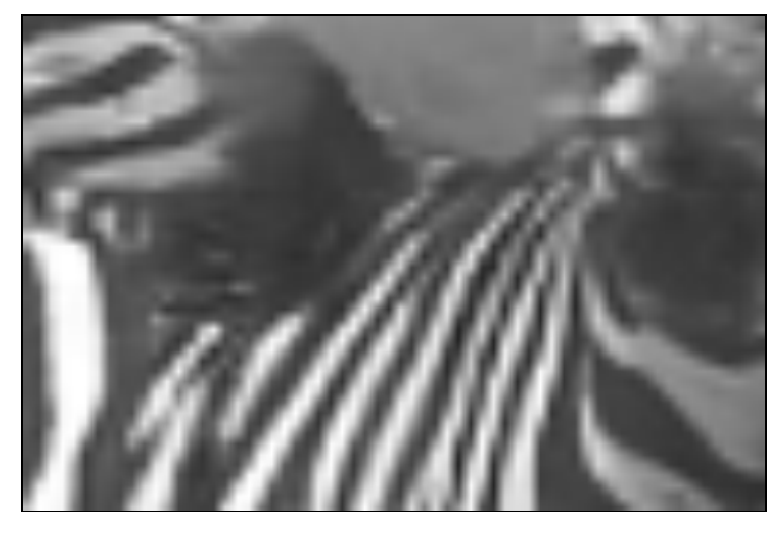} \\[-1pt]
		TF-ISTA-Net & \bf{RTF-ISTA-Net} & TF-FISTA-Net & RTF-FISTA-Net \\
		(29.76, 0.9081) & \textbf{(30.20, 0.9118)} & (29.77, 0.9079) & (29.98, 0.9086)
	\end{tabular} 
\caption[]{ Examples of reconstructions achieved using various methods under comparison. The numbers indicate (PSNR in dB, SSIM). Unlike the benchmark methods (first row), the proposed methods (second row) do not have prominent artifacts. The best performance is given by RTF-ISTA-Net.}
\label{fig:example_recon_BSD68}  
\end{center}
\end{figure*}

%------------------------------------------------------------------------

\begin{figure*}[htb]
\begin{center}
    \begin{tabular}[b]{P{.20\linewidth}P{.20\linewidth}P{.20\linewidth}P{.20\linewidth}}
        \includegraphics[width=1.0\linewidth]{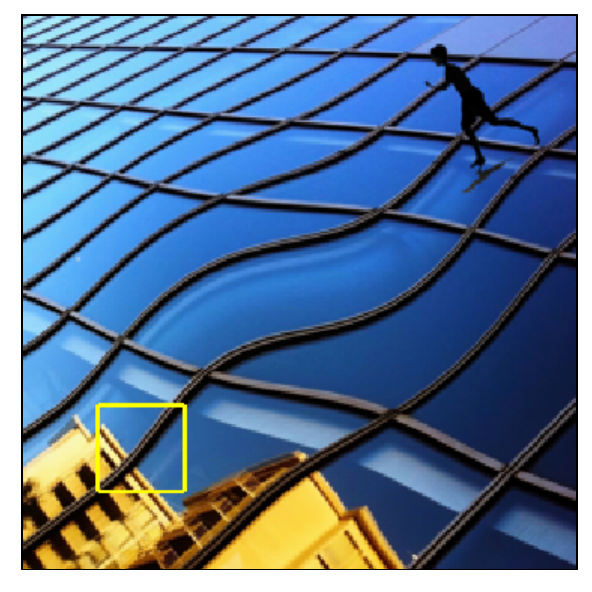}  &
        \includegraphics[width=1.0\linewidth]{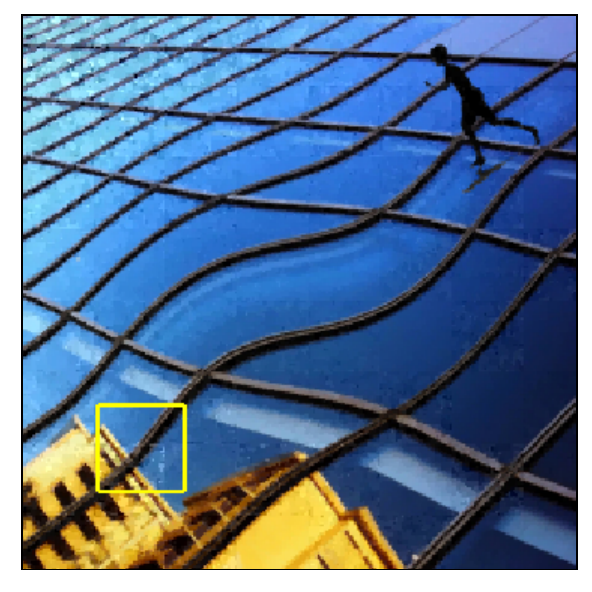}  &
        \includegraphics[width=1.0\linewidth]{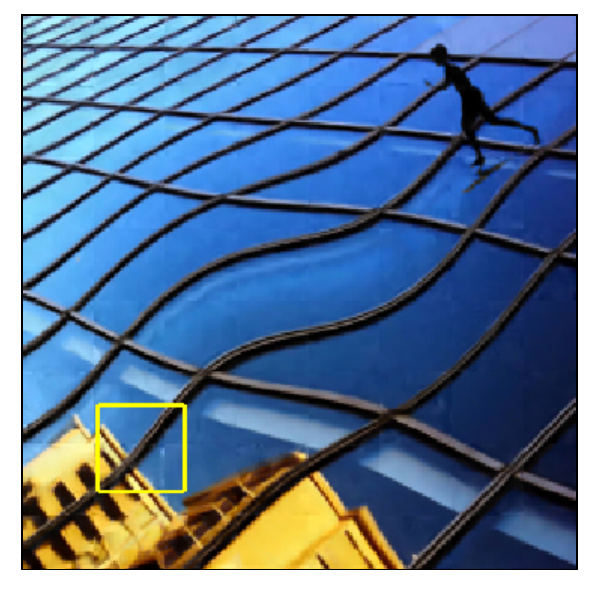} & 
        \includegraphics[width=1.0\linewidth]{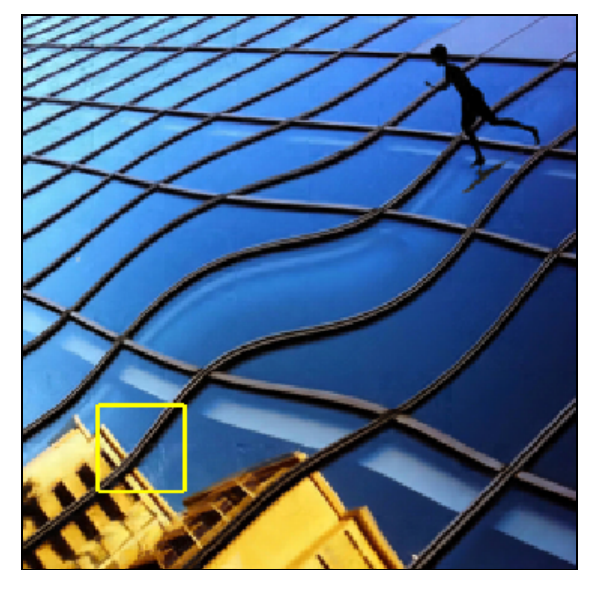} \\[-1pt]
        \includegraphics[width=1.0\linewidth]{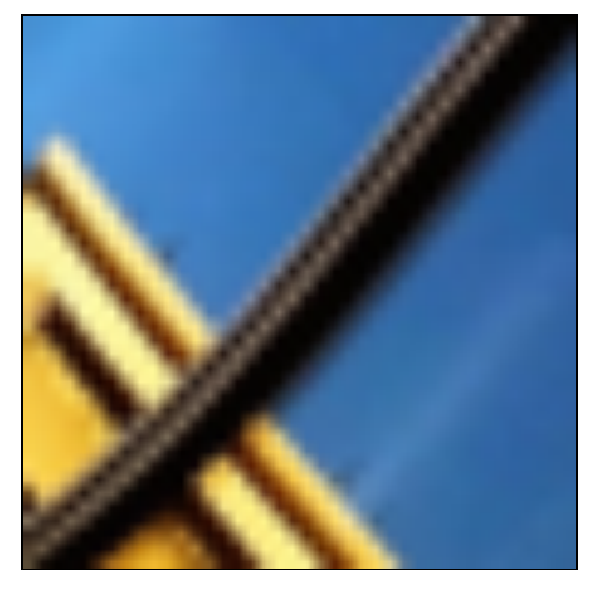}  &
        \includegraphics[width=1.0\linewidth]{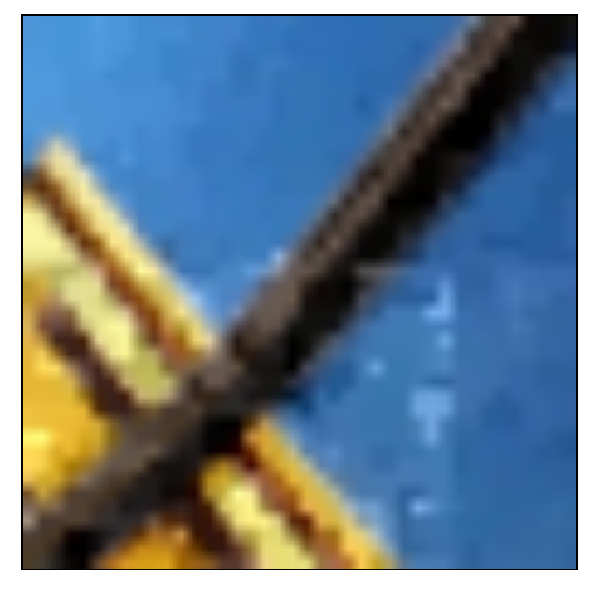}  &
        \includegraphics[width=1.0\linewidth]{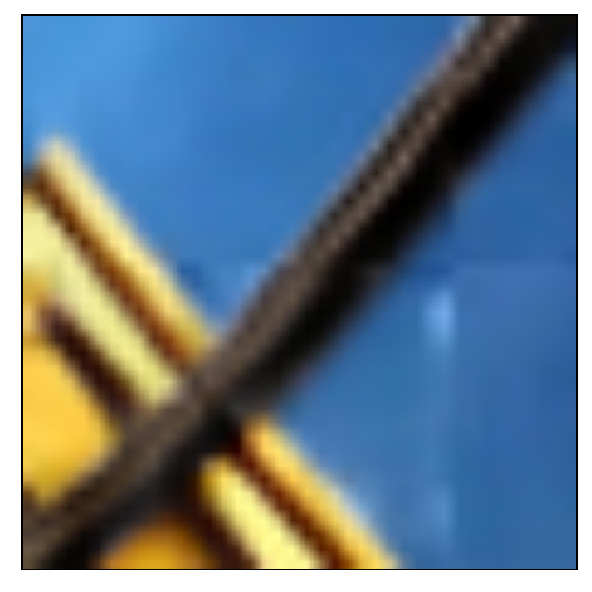} & 
        \includegraphics[width=1.0\linewidth]{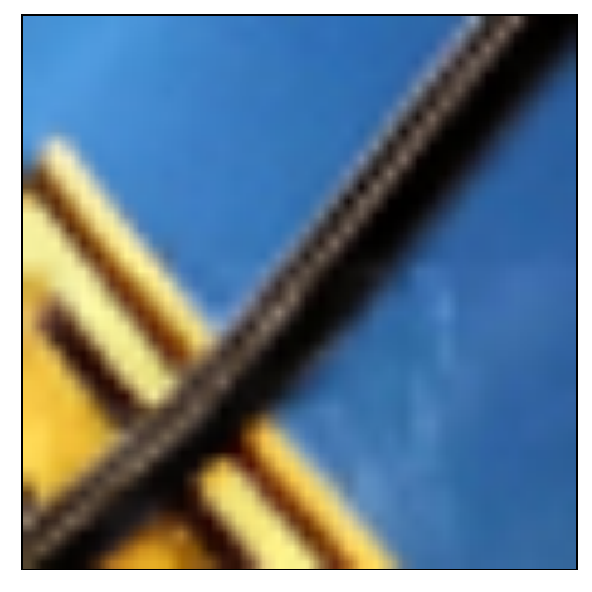} \\[-1pt]
        Original & TVAL3 & ReconNet & ISTA-Net \\ 
        & (26.52, 0.9053) & (29.67, 0.9442) & (34.55, 0.9796)\\ 
        \includegraphics[width=1.0\linewidth]{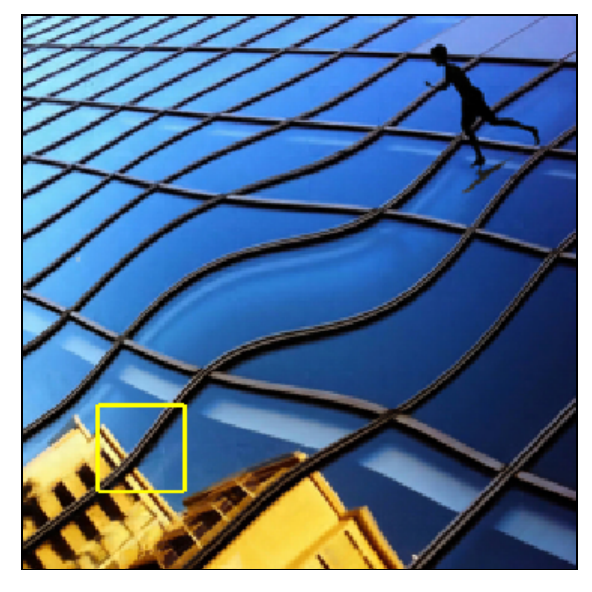}  &
        \includegraphics[width=1.0\linewidth]{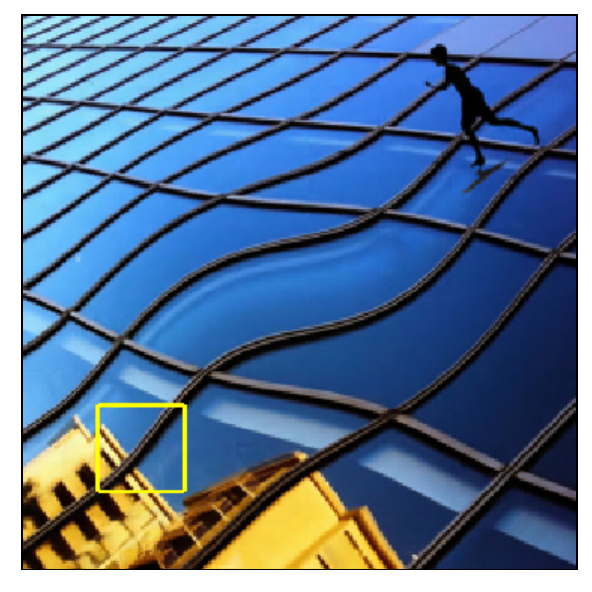}  &
        \includegraphics[width=1.0\linewidth]{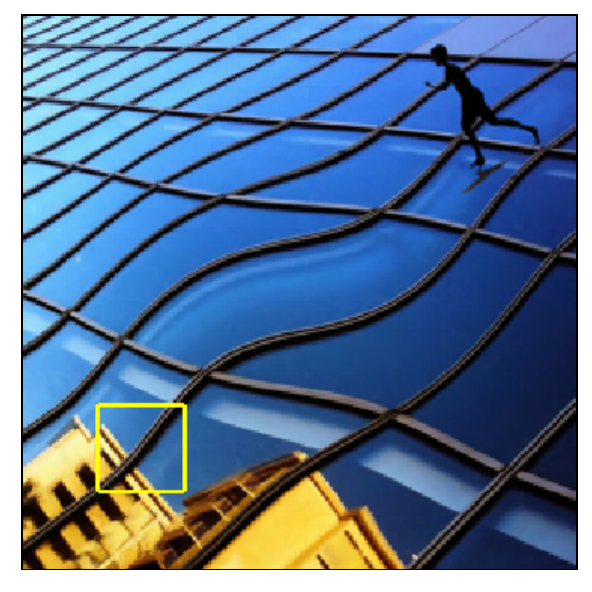} & 
        \includegraphics[width=1.0\linewidth]{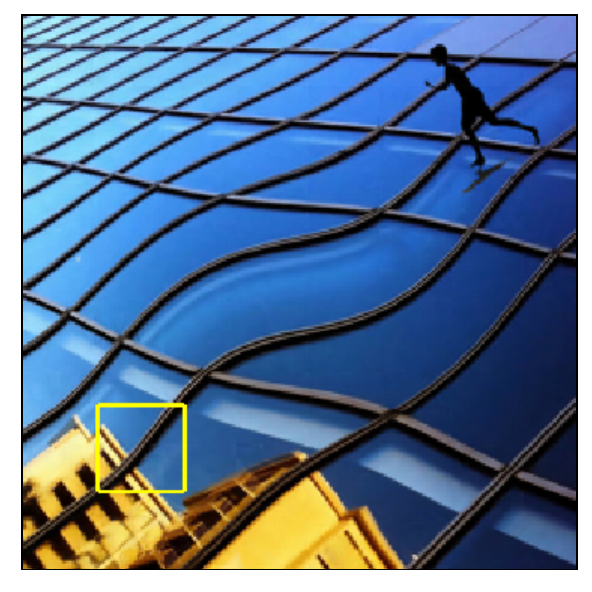} \\[-1pt]
        \includegraphics[width=1.0\linewidth]{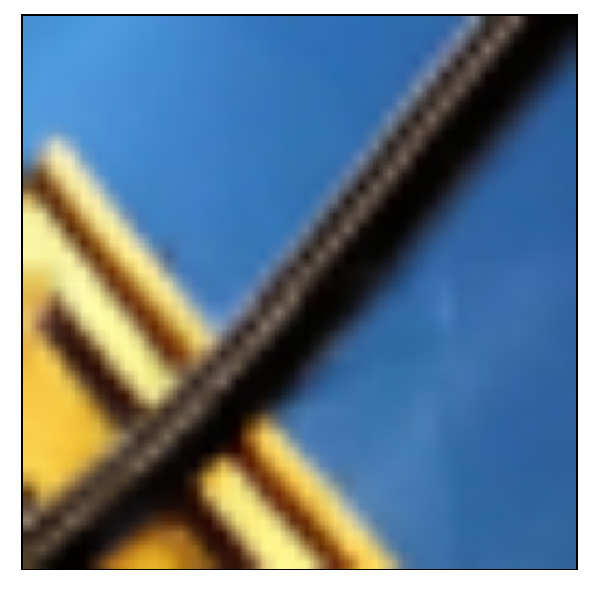}  &
        \includegraphics[width=1.0\linewidth]{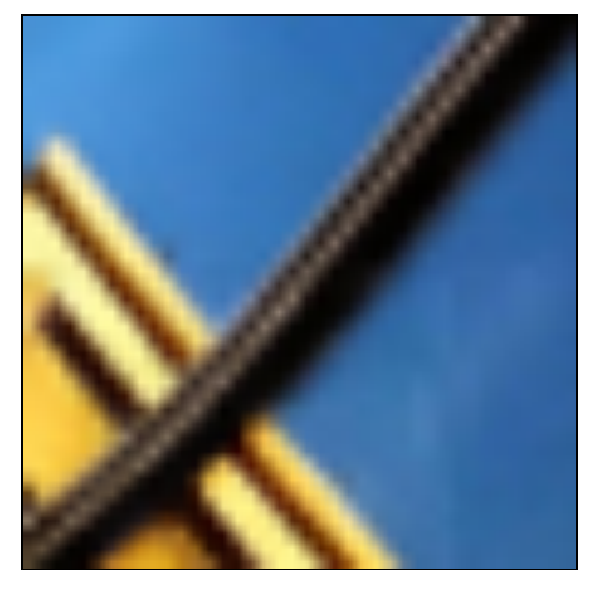}  &
        \includegraphics[width=1.0\linewidth]{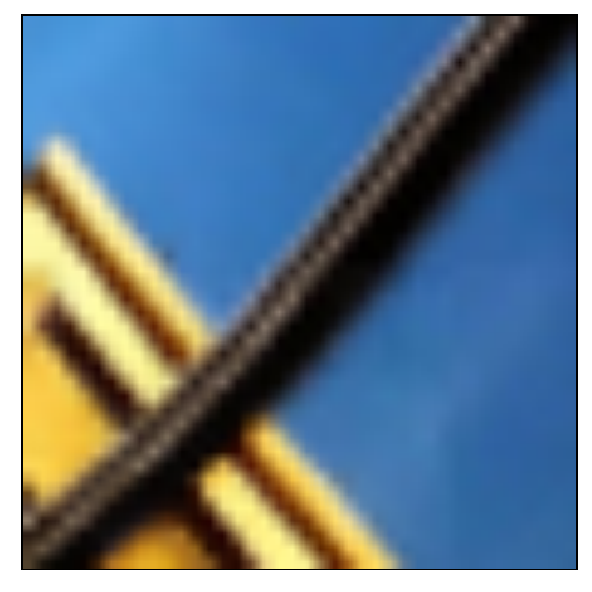} & 
        \includegraphics[width=1.0\linewidth]{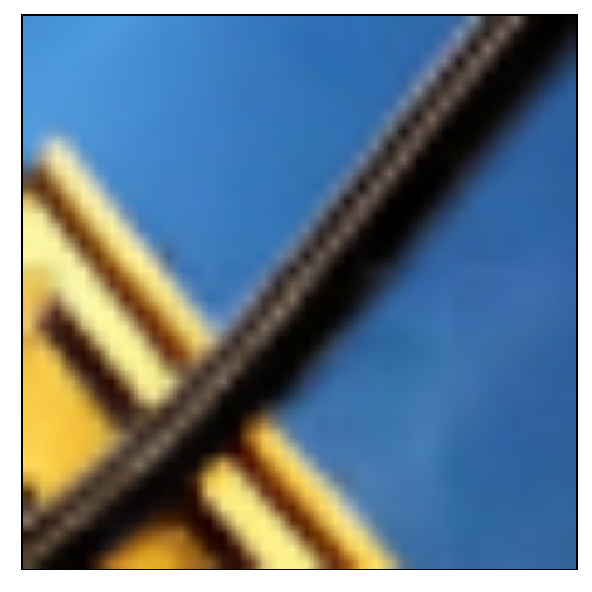} \\[-1pt]
        TF-ISTA-Net & \bf{RTF-ISTA-Net} & TF-FISTA-Net & RTF-FISTA-Net \\
        (37.20, 0.9876) & \textbf{(37.59, 0.9875)} & (37.46, 0.9876) & (37.00, 0.9859)
    \end{tabular} 
\caption[]{ Examples of reconstructions achieved using various methods under comparison. The numbers indicate (PSNR in dB, SSIM). Unlike the benchmark methods (first row), the proposed methods (second row) do not have prominent artifacts. The best performance is given by RTF-ISTA-Net.}
\label{fig:example_recon_ur100}  
\end{center}
\end{figure*}

%------------------------------------------------------------------------

\begin{figure*}[htb]
\begin{center}
    \begin{tabular}[b]{P{.21\linewidth}P{.21\linewidth}P{.21\linewidth}P{.21\linewidth}}
        \includegraphics[width=1.0\linewidth]{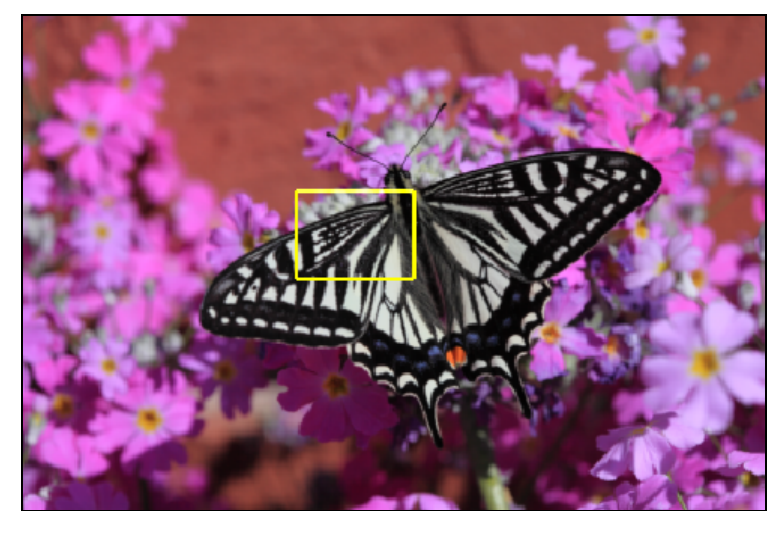}  &
        \includegraphics[width=1.0\linewidth]{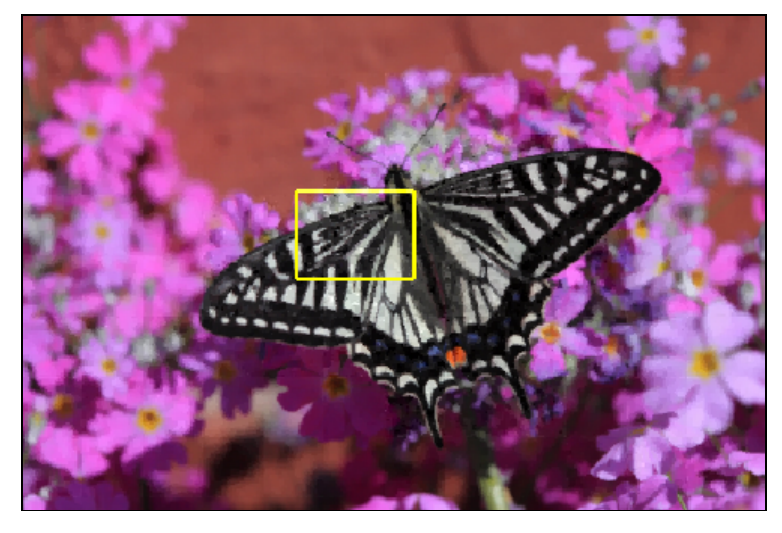}  &
        \includegraphics[width=1.0\linewidth]{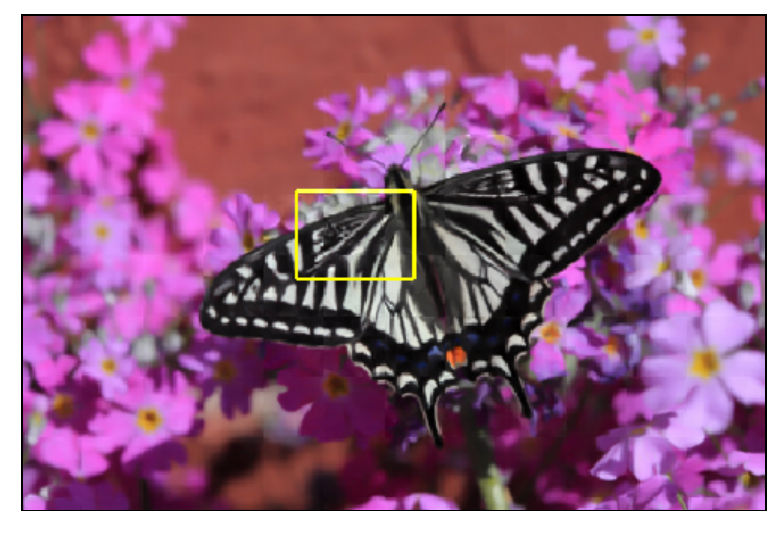} & 
        \includegraphics[width=1.0\linewidth]{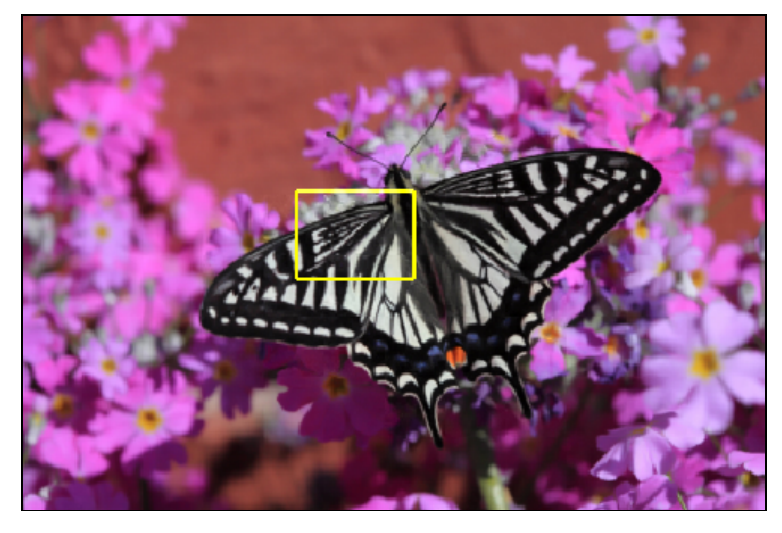} \\[-1pt]
        \includegraphics[width=1.0\linewidth]{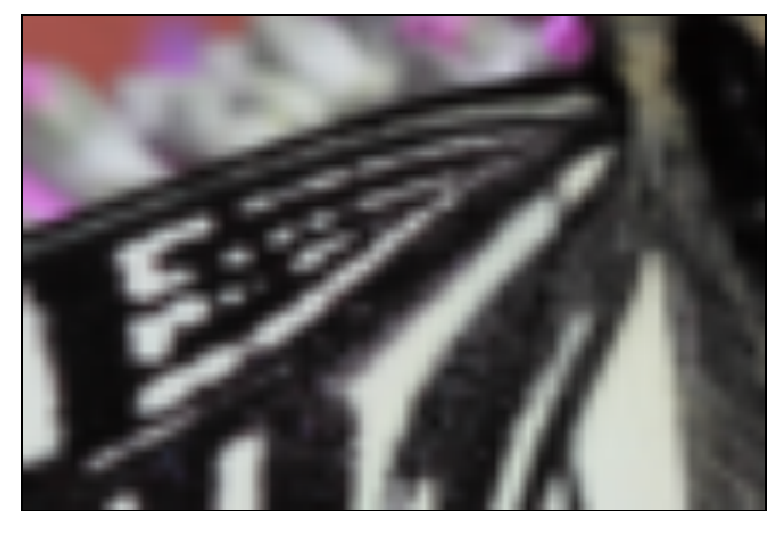}  &
        \includegraphics[width=1.0\linewidth]{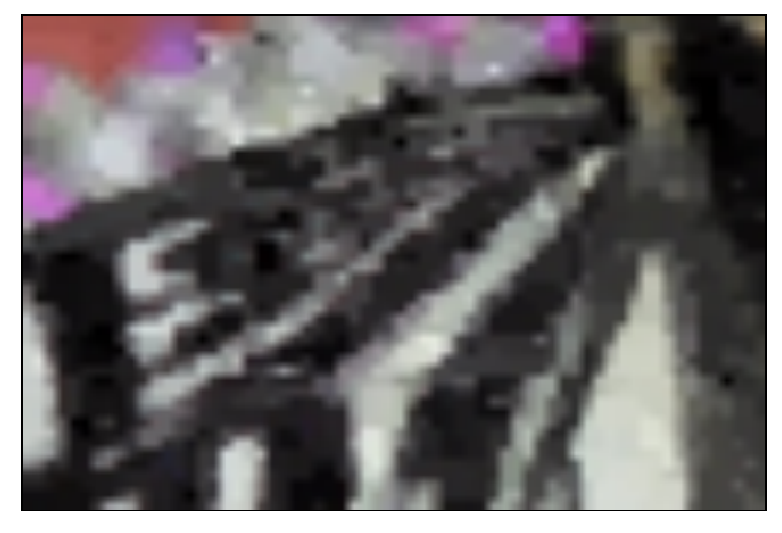}  &
        \includegraphics[width=1.0\linewidth]{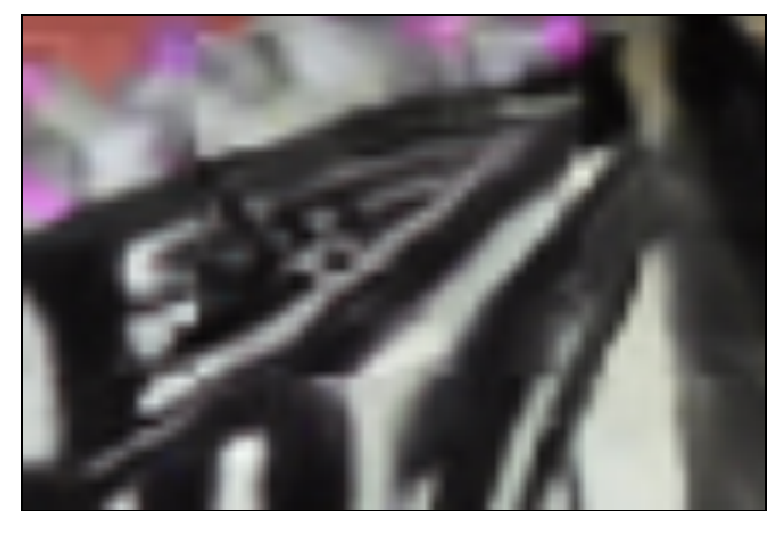} & 
        \includegraphics[width=1.0\linewidth]{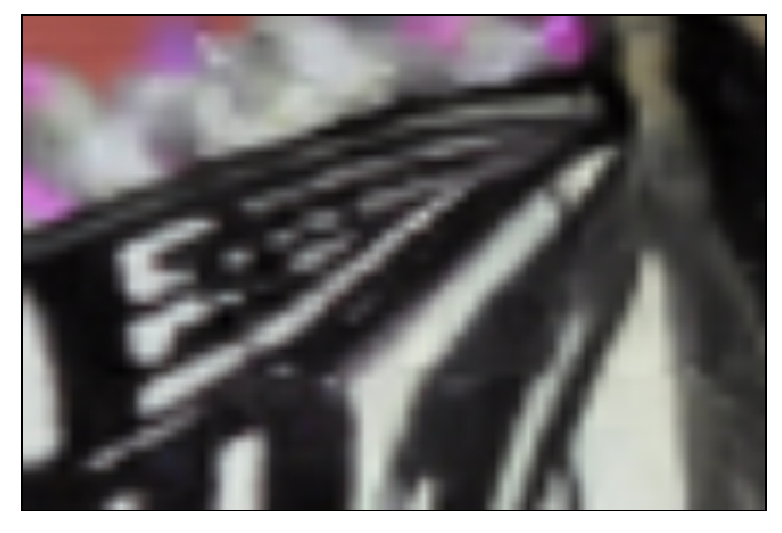} \\[-1pt]
        Original & TVAL3 & ReconNet & ISTA-Net \\ 
        & (29.70, 0.407) & (30.65, 0.9510) & (36.77, 0.9864)\\ 
        \includegraphics[width=1.0\linewidth]{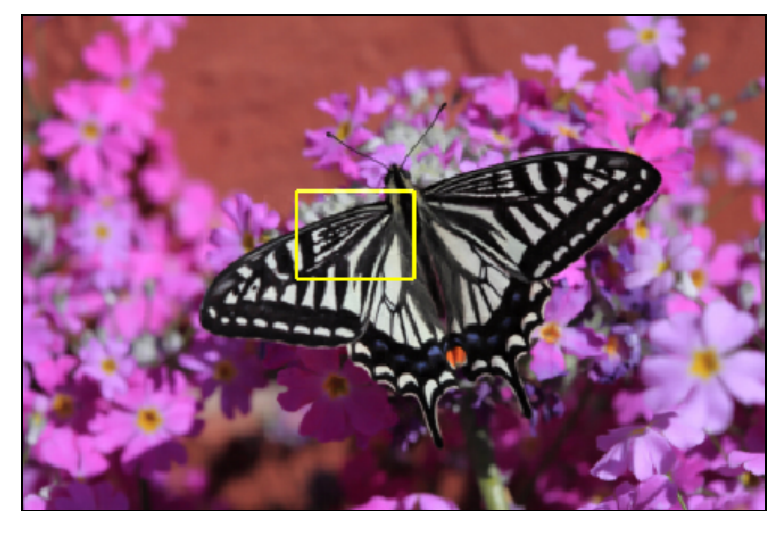}  &
        \includegraphics[width=1.0\linewidth]{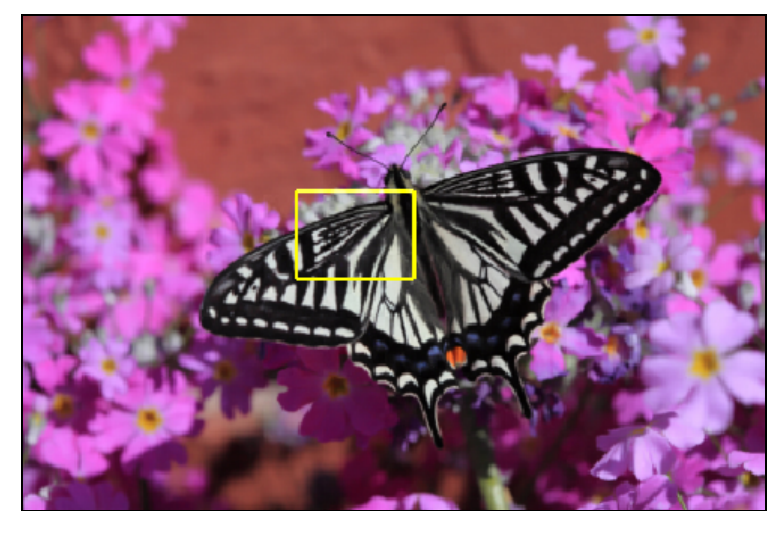}  &
        \includegraphics[width=1.0\linewidth]{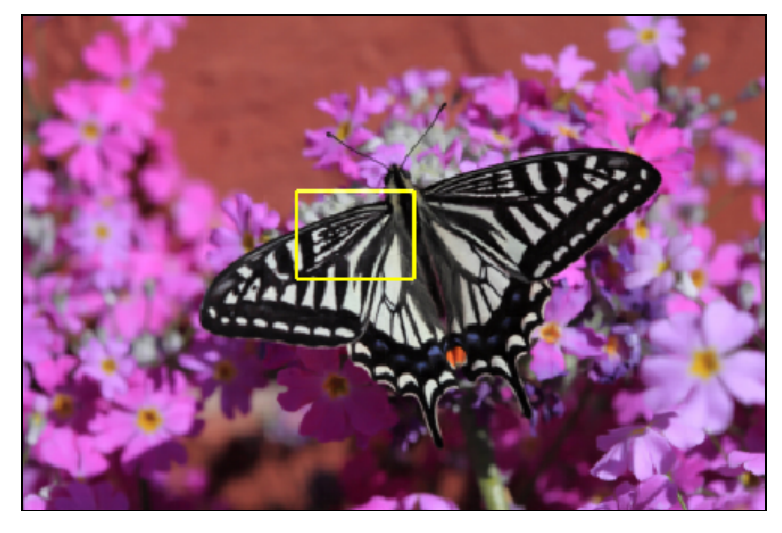} & 
        \includegraphics[width=1.0\linewidth]{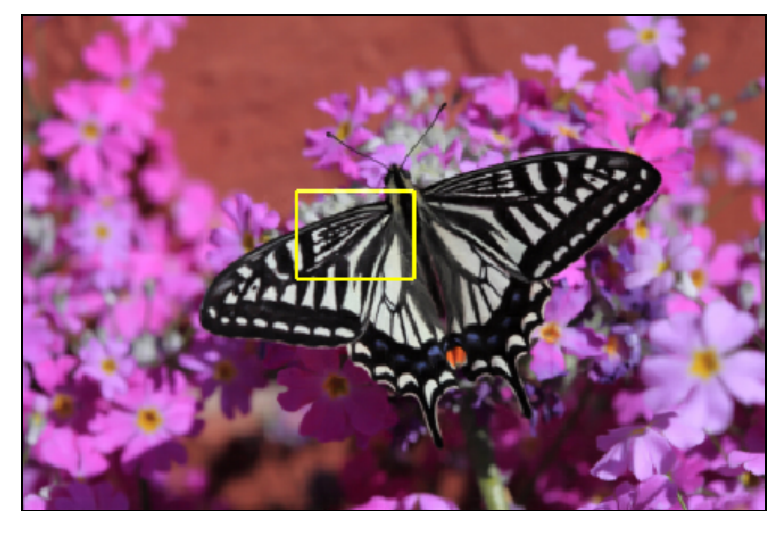} \\[-1pt]
        \includegraphics[width=1.0\linewidth]{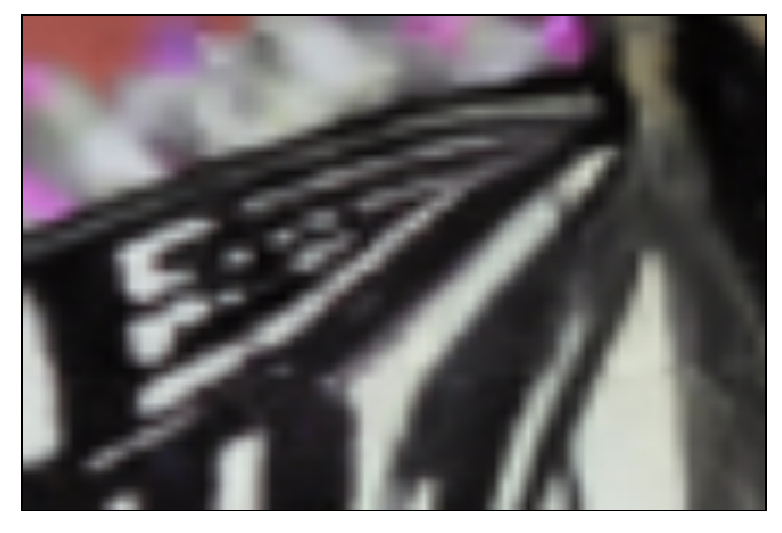}  &
        \includegraphics[width=1.0\linewidth]{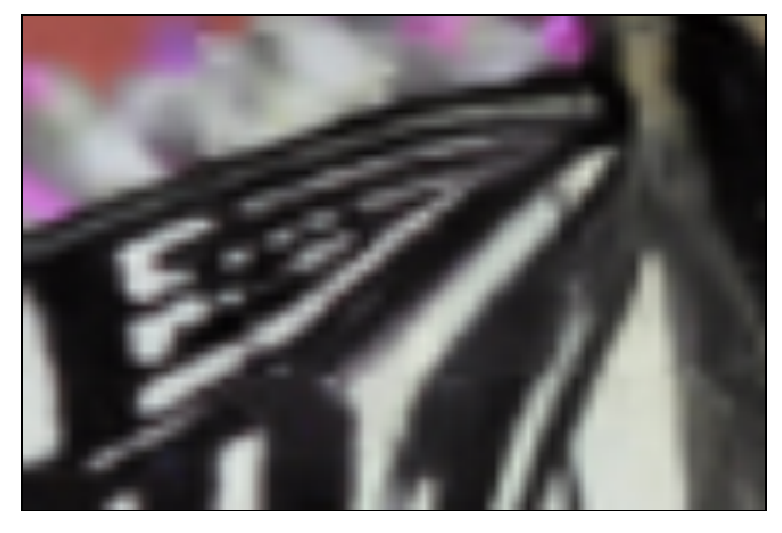}  &
        \includegraphics[width=1.0\linewidth]{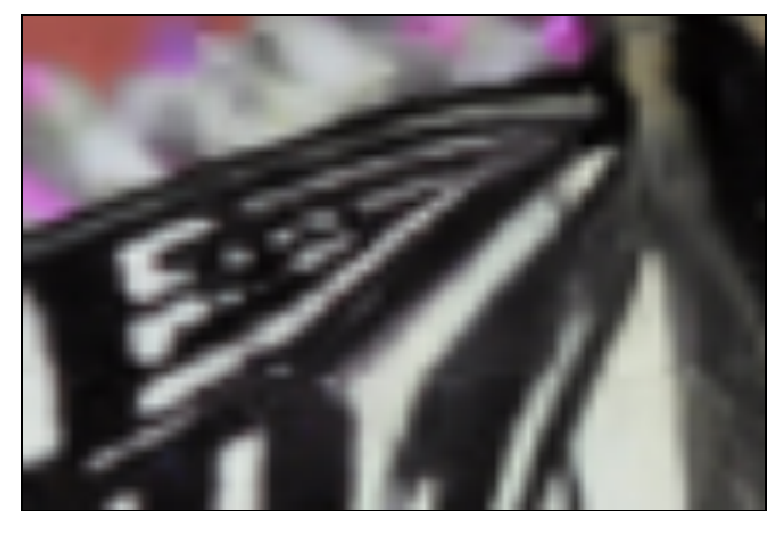} & 
        \includegraphics[width=1.0\linewidth]{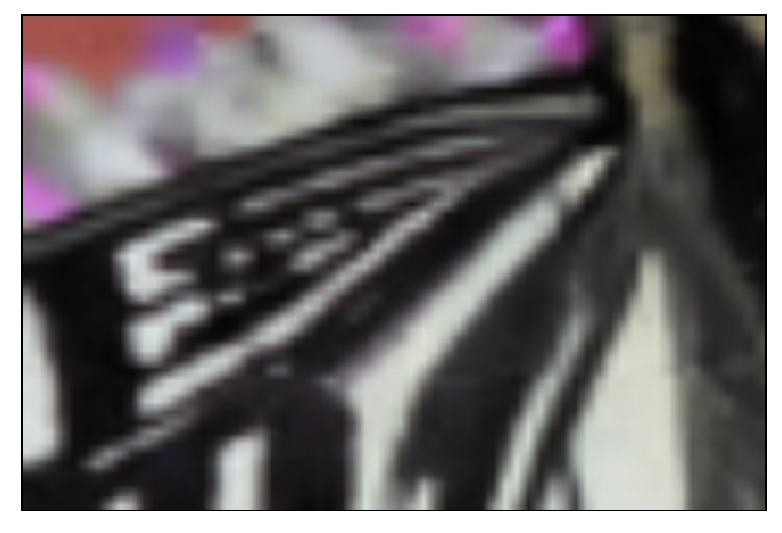} \\[-1pt]
        TF-ISTA-Net & RTF-ISTA-Net & \bf{TF-FISTA-Net} & RTF-FISTA-Net \\
        (38.33, 0.9897) & (38.28, 0.9896) & \textbf{(38.46, 0.9901)} & (37.93, 0.9885)
    \end{tabular} 
\caption[]{ Examples of reconstructions achieved using various methods under comparison. The numbers indicate (PSNR in dB, SSIM). Unlike the benchmark methods (first row), the proposed methods (second row) do not have prominent artifacts. TVAL3 and ReconNet have smudged wings and ISTA-Net has patch stitching artifacts (cf. zoomed in area). One can observe that the structure of the flowers is retained in all the proposed methods, but not in the benchmark methods (see zoomed-in area). The best reconstruction performance for this example is given by TF-FISTA-Net.}
\label{fig:example_recon_DIV2K}  
\end{center}
\end{figure*}

%------------------------------------------------------------------------

\begin{figure*}[htb]
    \begin{center}
        \begin{tabular}[b]{P{.21\linewidth}P{.21\linewidth}P{.21\linewidth}P{.21\linewidth}}
            \includegraphics[width=1.0\linewidth]{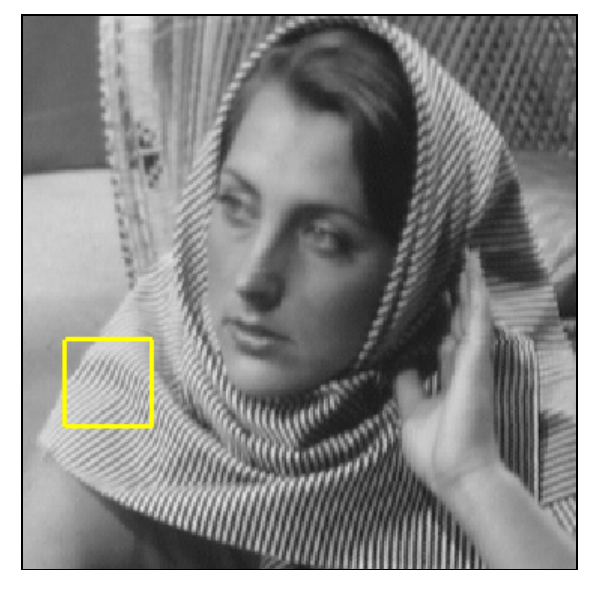}  &
            \includegraphics[width=1.0\linewidth]{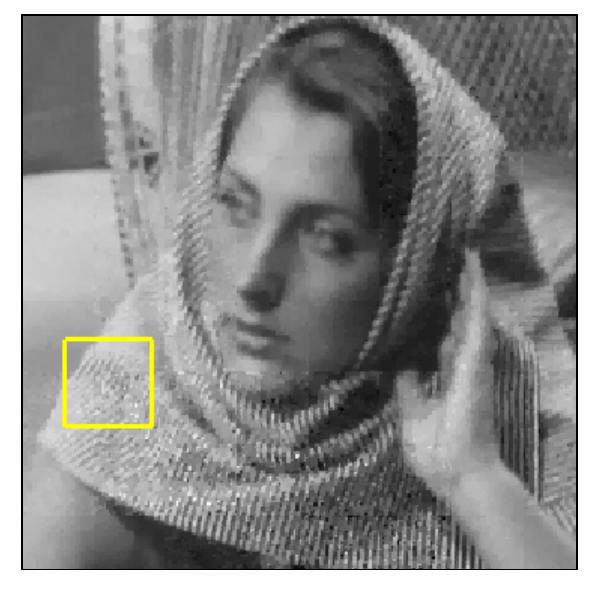}  &
            \includegraphics[width=1.0\linewidth]{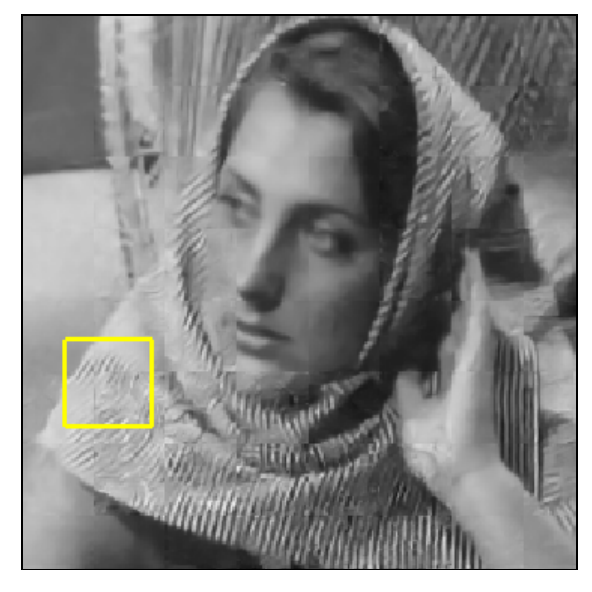} & 
            \includegraphics[width=1.0\linewidth]{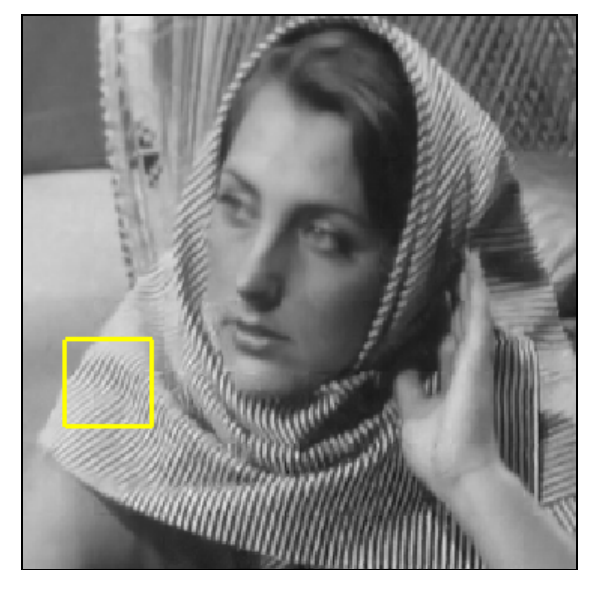} \\[-1pt]
            \includegraphics[width=1.0\linewidth]{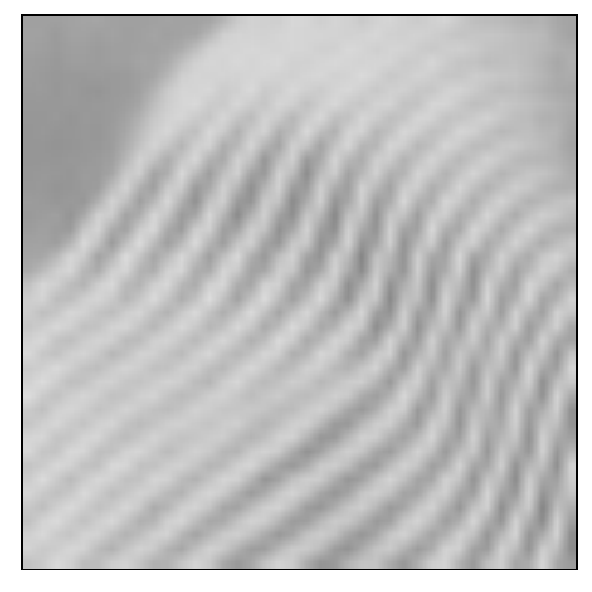}  &
            \includegraphics[width=1.0\linewidth]{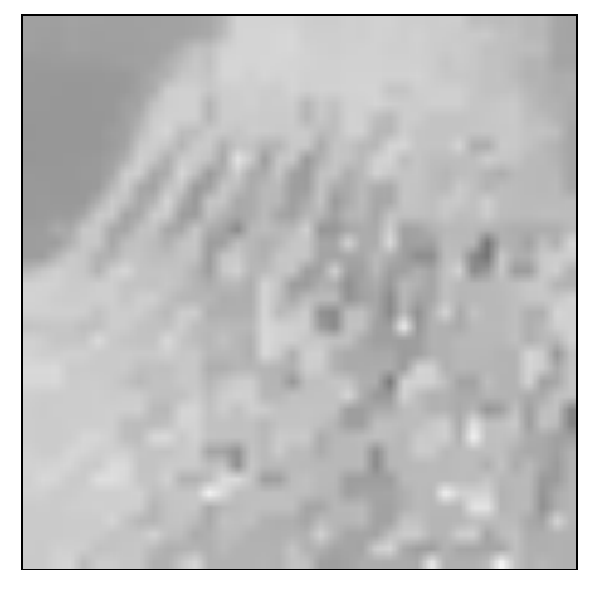}  &
            \includegraphics[width=1.0\linewidth]{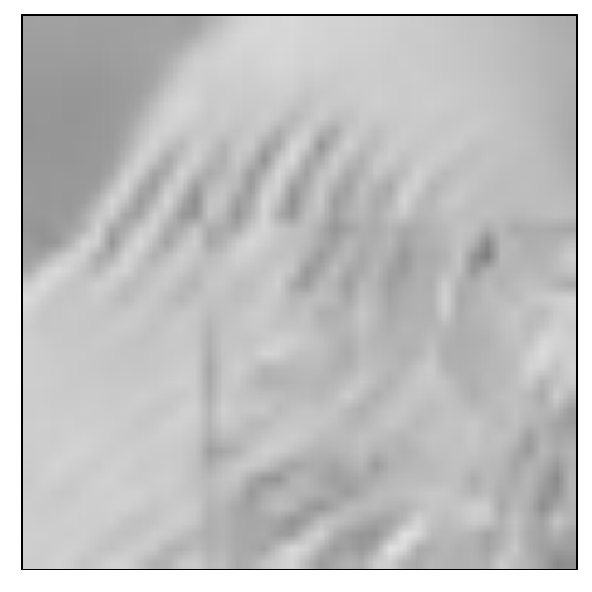} & 
            \includegraphics[width=1.0\linewidth]{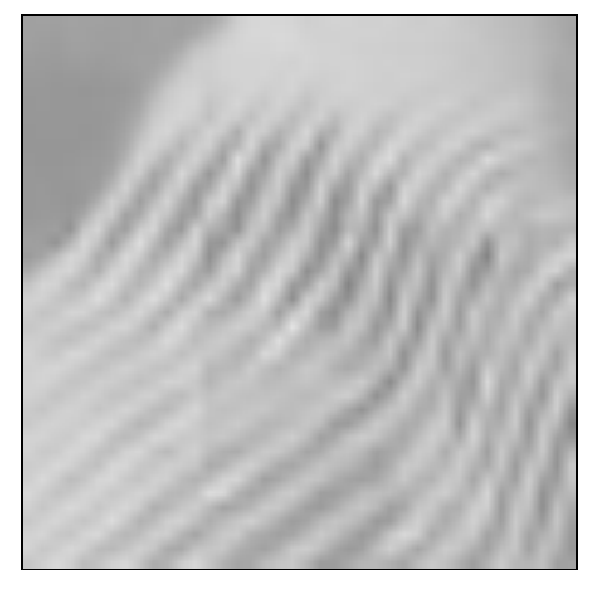} \\[-1pt]
            Original & TVAL3 & ReconNet & ISTA-Net \\ 
            & (26.74, 0.8438) & (26.62, 0.8483) & (32.86, 0.9557)\\ 
            \includegraphics[width=1.0\linewidth]{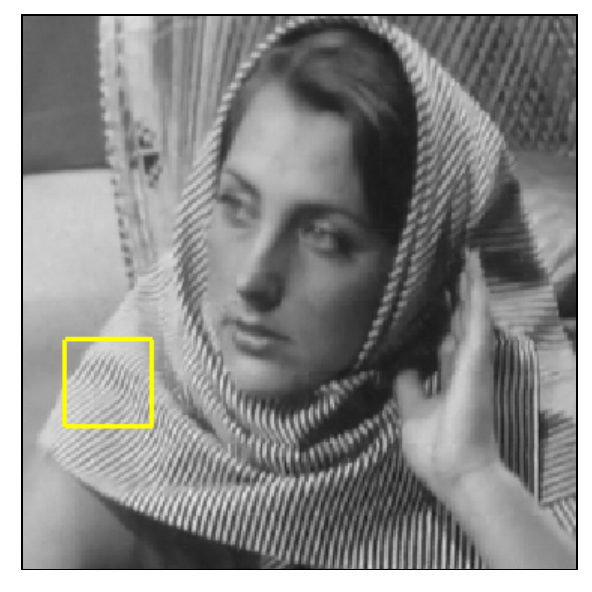}  &
            \includegraphics[width=1.0\linewidth]{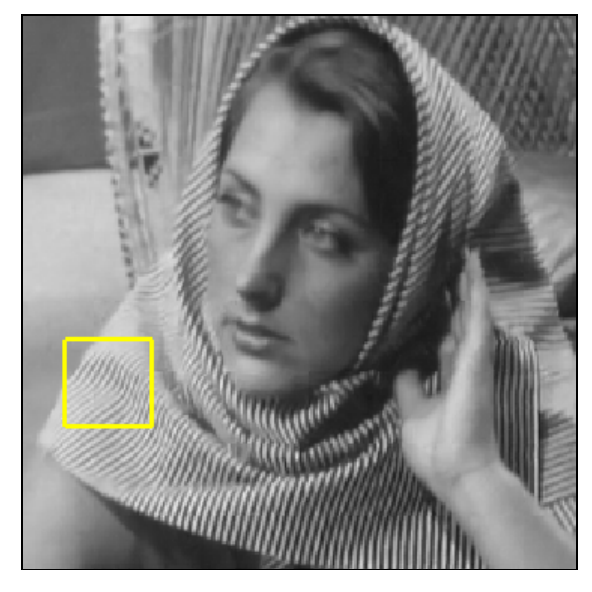}  &
            \includegraphics[width=1.0\linewidth]{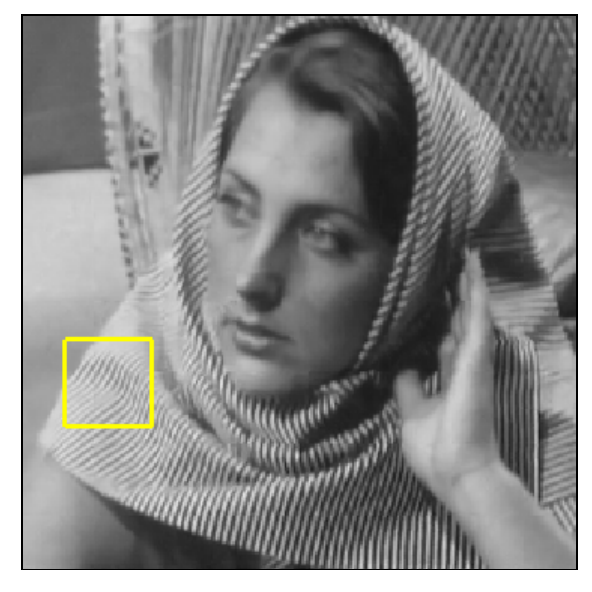} & 
            \includegraphics[width=1.0\linewidth]{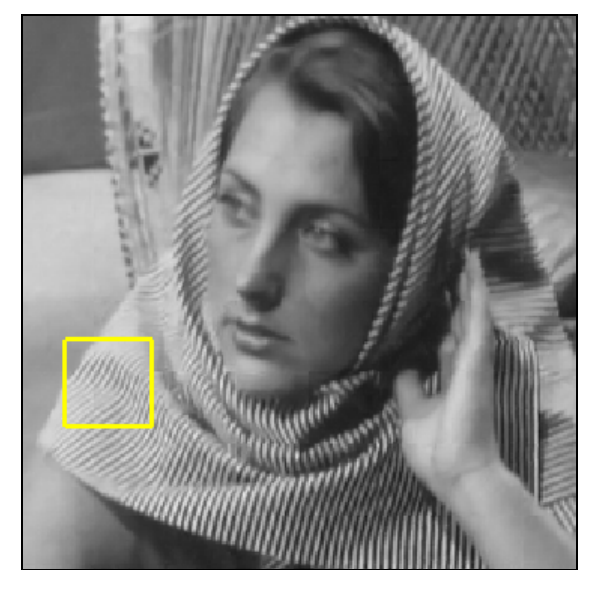} \\[-1pt]
            \includegraphics[width=1.0\linewidth]{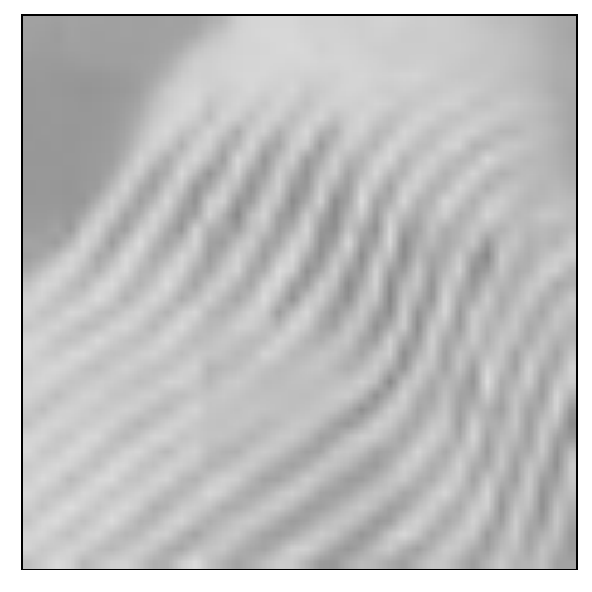}  &
            \includegraphics[width=1.0\linewidth]{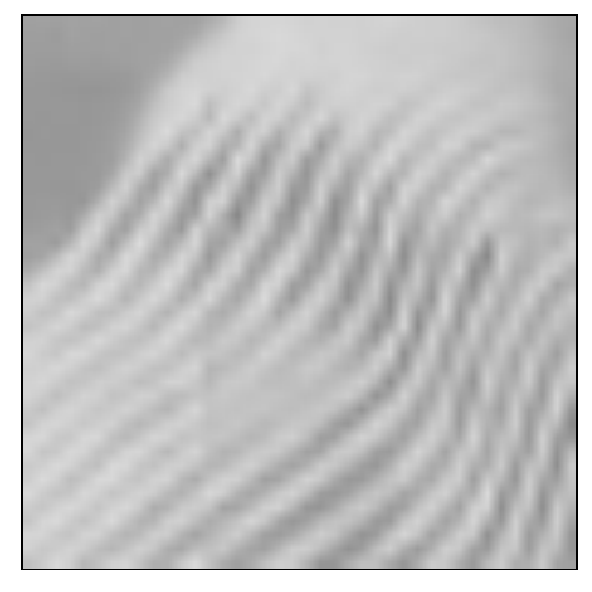}  &
            \includegraphics[width=1.0\linewidth]{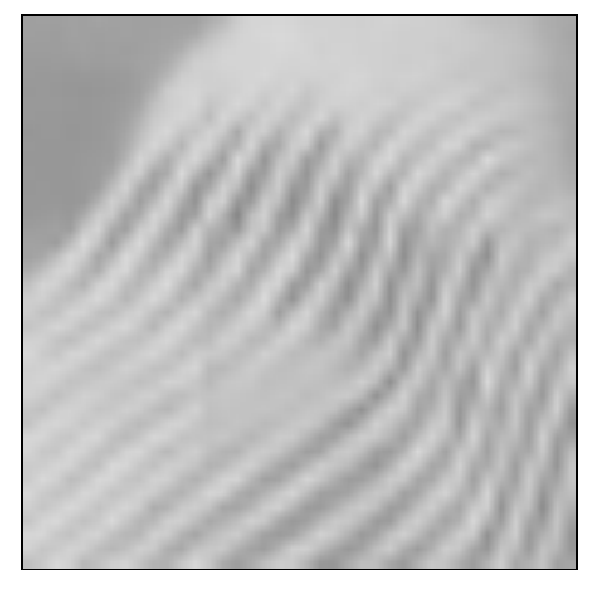} & 
            \includegraphics[width=1.0\linewidth]{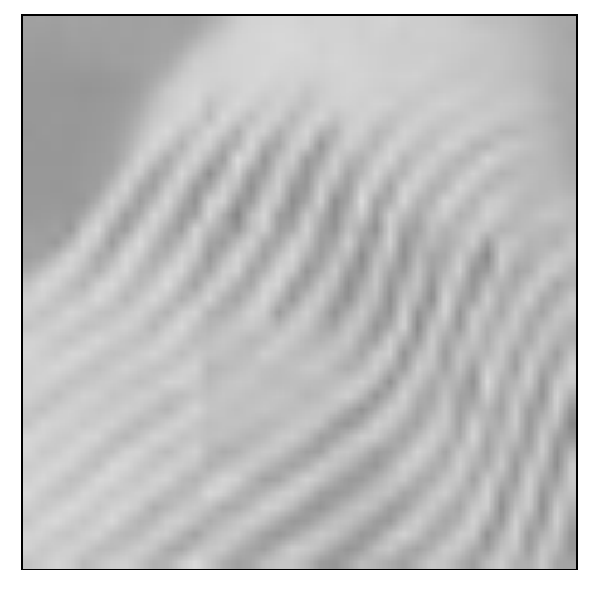} \\[-1pt]
            TF-ISTA-Net & \bf{RTF-ISTA-Net} & TF-FISTA-Net & RTF-FISTA-Net \\
            (34.07, 0.9637) & \textbf{(34.34, 0.9647)} & (34.12, 0.9641) & (33.92, 0.9615)
        \end{tabular} 
    \caption[]{ Examples of reconstructions achieved using various methods under comparison. The numbers indicate (PSNR in dB, SSIM). Unlike the benchmark methods (first row), the proposed methods (second row) do not have prominent artifacts. The texture pattern is also better preserved by the proposed techniques. RTF-ISTA-Net gave the best performance in this example.}
    \label{fig:example_recon_Set11}  
    \end{center}
    \end{figure*}

%------------------------------------------------------------------------

% ----------------------------------------
% ----------------------------------------
% 	      REFERENCES
% ----------------------------------------
% ----------------------------------------

\cleardoublepage
\bibliographystyle{ieeetr}
\bibliography{references}

\end{document}